\documentclass[runningheads]{llncs}

\usepackage[utf8]{inputenc}
\usepackage[T1]{fontenc}

\usepackage{enumerate}
\usepackage{array}
\usepackage{tikz-cd}

\usepackage{booktabs}   

\usepackage{mathtools}
\usepackage{amsfonts}
\usepackage{amssymb}
\usepackage{stmaryrd}
\usepackage{bussproofs}  \EnableBpAbbreviations
\usepackage{cmll}
\usepackage{bm} 

\usepackage{thmtools, thm-restate}

\DeclareSymbolFont{symbolsC}{U}{txsyc}{m}{n}
\DeclareMathSymbol{\circledwedge}{\mathrel}{symbolsC}{84}
\DeclareMathSymbol{\circledvee}{\mathrel}{symbolsC}{85}

\usepackage[%
  none,
  full,		
]{optional}

\opt{short}{\usepackage{hyperref}}
\opt{full}{\usepackage[pagebackref,linktoc=all]{hyperref}}
\opt{full}{\renewcommand*{\backref}[1]{}}
\opt{full}{\renewcommand*{\backrefalt}[4]{{
    \ifcase #1 Not cited.%
          \or Cited on page~#2.%
          \else Cited on pages #2.%
    \fi%
}}}

\hypersetup{
  breaklinks=true,
  colorlinks = true,
  urlcolor = black,
  linkcolor= black,
  citecolor= black,
  filecolor= black,
}

\opt{full}{\setcounter{tocdepth}{3}}

\usepackage{color}

\usepackage{macros}

\usepackage{comment}
\opt{fullproof}{\newcommand\FLAGPROOF{}}
\ifdefined\FLAGPROOF 
  
\else
  \excludecomment{fullproof}
\fi

\opt{draft}{\newcommand\FLAGDRAFTPROOF{}}
\ifdefined\FLAGDRAFTPROOF 
  
\else
  \excludecomment{draftproof}
\fi

\opt{full}{\newcommand\FLAGFULL{}}
\ifdefined\FLAGFULL
  \newenvironment{full}{}{}
\else
  \excludecomment{full}
\fi

\opt{draft}{\newcommand\FLAGDRAFT{}}
\ifdefined\FLAGDRAFT
  
\else
  \excludecomment{draft}
\fi

\spnewtheorem*{remark*}{Remark}{\itshape}{\upshape}

\author{Colin Riba \and Alexandre Kejikian}

\authorrunning{C. Riba and A. Kejikian}

\institute{ENS de Lyon, CNRS, Université Claude Bernard Lyon 1, LIP, UMR 5668,
69342, Lyon cedex 07, France}

\newcommand\mytitle{Infinitary Refinement Types for Temporal Properties in Scott Domains}

\title{\mytitle}

\begin{document}

\maketitle              

\begin{abstract}
We discuss an infinitary refinement type system for input-output
temporal specifications of functions that handle infinite objects
like streams or infinite trees.
Our system is based on a reformulation of Bonsangue and Kok's
infinitary extension 
of Abramsky's Domain Theory in Logical Form to saturated properties.
We show that in an interesting range of cases,
our system is complete without the need of an infinitary rule
introduced by Bonsangue and Kok to reflect the well-filteredness of Scott domains.

\end{abstract}

\section{Introduction}
\label{sec:intro}

\noindent
We are interested in input-output
specifications of higher-order programs that handle infinite 
data, such as streams or non-wellfounded trees.
Consider e.g.
\[
\begin{array}{r c l}
  \filter
& :
& (\BT \to \Bool)
  ~\longto~
  \Stream \BT
  ~\longto~
  \Stream \BT
\\

  \filter\ p\ (a \Colon x)
& =
& \term{if}~ (p\ a)
  ~\term{then}~ a \Colon (\filter\ p\ x)
  ~\term{else}~ (\filter\ p\ x)
\end{array}
\]

\noindent
where $\Stream\BT$ stands for the type of streams on $\BT$.
Assume $p : \BT \to \Bool$ is a function that
tests for a property $\psi$.
If $x$ is a stream on $\BT$,
then $(\filter\ p\ x)$ retains those elements of $x$ which satisfy $\psi$.
The stream produced by $(\filter\ p\ x)$ is thus only partially defined
if $x$ has only finitely many elements satisfying $\psi$.

Logics like $\LTL$ (Linear Temporal Logic),
$\CTL$ (Computation Tree Logic)
or the modal $\mu$-calculus are
widely used to formulate, on infinite objects,
safety and liveness properties (see e.g.~\cite{hr07chapter,bs07chapter}).
Safety properties state that
some ``bad'' event will not occur,
while liveness properties 
specify that ``something good'' will happen 
(see e.g.~\cite{bk08book}).
One typically uses temporal modalities like $\Box$ (\emph{always})
or $\Diam$ (\emph{eventually})
on streams in specifications of programs.

A possible specification for $\filter$
asserts that $(\filter\ p \ x)$ is a totally defined stream whenever
$x$ is a totally defined stream with infinitely many elements satisfying $\psi$.
We express this as follows.
Let $\BT$ be finite and assume given,
for each $a$ of type $\BT$, a formula $\form a$ which holds on $b : \BT$
exactly when $b$ is $a$.
Then $\Box \bigvee_a \form a$ selects those streams on $\BT$
which are totally defined.
The formula $\form\hd\psi$ holds on a stream $(a \Colon x)$ when $\psi$ holds on $a$.
Hence $\Box\Diam \form\hd\psi$ expresses that a stream
has infinitely many elements satisfying $\psi$.
We can thus state that 
\begin{equation}
\label{eq:intro:spec}
\begin{array}{c}
  \text{$x$ satisfies $\Box\bigvee_a \form a$ and $\Box\Diam \form\hd\psi$}
  \quad\longimp\quad
  \text{$(\filter\ p\ x)$ satisfies $\Box\bigvee_a \form a$}
\end{array}
\end{equation}

It is undecidable whether a given higher-order program satisfies
such an input-output temporal property written with formulae
of the modal $\mu$-calculus \cite{ktu10popl}.
A previous work~\cite{jr21esop}
provided a refinement type system for proving such properties.
This type system handles the (negation-free) alternation-free modal $\mu$-calculus
on infinite types such as streams or trees.
But it is based on guarded recursion
and does not allow for non-productive functions such as $\filter$.

In this paper, we consider a fragment of $\FPC$ equipped with general recursion
($\FPC$ extends $\PCF$~\cite{plotkin77tcs} with
recursive types, see e.g.\ \cite{pierce02book}).
We are interested in specifications as in \eqref{eq:intro:spec},
but interpreted at the level of denotational semantics:
In our view, since a stream (as opposed to e.g.\ an integer)
is an inherently infinite object, the above specification for $\filter$
should hold for any stream whatsoever, and not only
for those definable in a given programming language.

This leads us to consider temporal properties
on infinite datatypes in Scott domains.
We noted in \cite{rs24jfla} that the usual rule of Scott induction
(see e.g. \cite[\S 6.2]{ac98book}) does not prove liveness
properties like \eqref{eq:intro:spec} above.
We instead resort to Abramsky's paradigm of
``Domain Theory in Logical Form'' (DTLF) \cite{abramsky91apal}.
We build on~\cite{bk03ic}, in which Bonsangue and Kok extend DTLF
to an infinitary type system which is sound and complete for
a large family of infinitary properties, known as the \emph{saturated} ones.%
\footnote{A subset of a domain is \emph{saturated} when it is upward-closed.}
This includes the \emph{negation-free} formulae of
(suitable adaptations of) the modal $\mu$-calculus \cite{rs24jfla},
and thus the specification in \eqref{eq:intro:spec}.

We present an infinitary refinement type system for saturated properties.
Our system is a reformulation of DTLF which,
in contrast with \cite{abramsky91apal,bk03ic}
(see also \cite[\S 10.5]{ac98book}),
has no specific syntactic entity for compact open sets.
We do isolate a finitary logical fragment,
but it only consists of finite conjunctions and falsity,
without non-empty disjunctions.
As consequence, our version of Abramsky's \emph{coprimeness predicate}
is a \emph{consistency} predicate,
which selects those finite formulae with a non-empty interpretation.
Besides, our consistency predicate has a \emph{positive} (inductive) definition
(cf.\ \cite{bk03ic,ac98book}).

Also, similarly as in~\cite{bk03ic}, the completeness of
our system relies on a topological property of Scott domains
known as \emph{well-filteredness},
and reflected in an infinitary rule ($\ax{WF}$ in \S\ref{sec:wf}).
We show that this rule is actually not needed for an
interesting range of specifications,
including \eqref{eq:intro:spec} for $\filter$,
as well as various specifications for functions on streams and trees
(see Example~\ref{ex:reft:fun} and Theorem~\ref{thm:main}).

Having some control on the rule $\ax{WF}$
is relevant in the context of this work.
We ultimately target a finitary system in which,
similarly as in~\cite{jr21esop},
infinitary behaviours of fixpoint formulae are simulated
by explicit quantifications over the number of unfoldings
of these fixpoints.
To this end, it is important to know that the rule $\ax{WF}$
can be avoided in many interesting cases.

\paragraph{Organization of the paper.}
We devise our refinement type system in \S\ref{sec:system}.
Its (Scott) semantics is presented in \S\ref{sec:sem},
and completeness is handled in \S\ref{sec:compl}.
Finally, \S\ref{sec:conc} discusses our results in the perspective of further works.

\opt{short}{Appendix \ref{sec:app} contains additional technical material.
Proofs are available in the Appendices of the full version~\cite{rk25full}.}%
\opt{full}{Proofs are available in the Appendices.}

\section{A Refinement Type System}
\label{sec:system}

We assume given a collection of sets $\Base$,
which will play the role of \emph{base types}.

\subsection{The Pure System}
\label{sec:pure}

The \emph{pure types}
(notation $\PT,\PTbis,\dots$) are
the closed types over the grammar
\[
\begin{array}{r @{\ \ }c@{\ \ } l}
     \PT
&    \bnf
&    \BT
\gss \PT \times \PT
\gss \PT \arrow \PT
\gss \TV
\gss \rec \TV.\PT 
\end{array}
\]

\noindent
where $\BT \in \Base$,
where $\TV$ ranges over an infinite supply of \emph{type variables},
and where $\rec\TV.\PT$ binds $\TV$ in $\PT$.
We consider terms from the grammar
\[
\begin{array}{r !{~}r!{~} l}
    M,N 
&   \bnf
&   x
\gs \lambda x.M
\gs M N
\gs \fix x.M
\gs \fold(M)
\gs \unfold(M)
\\

&   \mid
&   \pair{M,N}
\gs \pi_1(M)
\gs \pi_2(M)
\gs a
\gs \cse\ M\ \copair{a \mapsto N_a \mid a \in \BT}
\end{array}
\]

\noindent
where $\BT \in \Base$ and $a \in \BT$.
The term formers $\fold, \unfold, \pi_1, \pi_2$
are often written in curried form:
e.g.\ $(\fold M)$ stands for $\fold(M)$.
Write $M \comp N$ for $\lambda x. M(N\, x)$.

Terms are typed as usual, with judgments of the form
$\Env \thesis M : \PT$, where $\Env$ is a list
$x_1:\PTbis_1,\dots,x_n:\PTbis_n$ with $x_i \neq x_j$
if $i \neq j$.
Some typing rules are presented in Figure~\ref{fig:reft:puretyping}.%
\footnote{The set of all typing rules of the pure system is in
Figure~\ref{fig:app:puretyping}, \S\ref{sec:app}.}
Of course, each type $\PT$ is inhabited
by the term $\Omega_\PT \deq \fix x.x : \PT$.

\begin{figure}[t!]
\[
\begin{array}{c}
\dfrac{\Env,x:\PT \thesis M : \PT}
  {\Env \thesis \fix x.M : \PT}
\qquad\quad

\dfrac{\Env \thesis M : \PT[\rec\TV.\PT/\TV]}
  {\Env \thesis \fold(M) : \rec\TV.\PT}

\qquad\quad

\dfrac{\Env \thesis M : \rec\TV.\PT}
  {\Env \thesis \unfold(M) : \PT[\rec\TV.\PT/\TV]}

\\\\

\dfrac{}
  {\Env \thesis a : \BT}
~\text{($\BT \in \Base$ and $a \in \BT$)}

\qquad\quad

\dfrac{ \Env \thesis M : \BT
  \qquad\text{for each $a \in \BT$,\quad} \Env \thesis N_a : \PT}
  {\Env \thesis \cse\ M\ \copair{a \mapsto N_a \mid a \in \BT} : \PT}

\end{array}
\]
\caption{Typing rules of the pure calculus (excerpt).%
\label{fig:reft:puretyping}}
\end{figure}

\begin{full}
\begin{remark*} 
This paper concerns the denotational semantics
of the above fragment of $\FPC$ in Scott domains.
This denotational semantics, to be discussed in~\S\ref{sec:sem},
is compatible with the contextual closure of the 
usual evaluation rules.
\[
\begin{array}{r c l !{\qquad} r c l}
  (\lambda x.M)N
& \rhd
& M[N/x]

& \unfold(\fold M)
& \rhd
& M
\\

  \pi_i \pair{M_1,M_2}
& \rhd
& M_i

& \fix x.M
& \rhd
& M[\fix x.M/x]
\\

  \cse\ a\ \copair{a \mapsto N_a \mid a \in \BT}
& \rhd
&  N_a
\end{array}
\]
\end{remark*}
\end{full}

\begin{example} 
\label{ex:pure}
The type of \emph{streams over $\PTbis$} is
$\Stream\PTbis \deq \rec\TV.\, \PTbis \times \TV$.
It is equipped with the constructor
\(
  \Cons
  \deq
  \lambda h .\lambda t.\fold \pair{h,t}
  :
  \PTbis \arrow \Stream\PTbis \arrow \Stream \PTbis
\).
We use the infix notation $(M \Colon N)$ for $(\Cons M\, N)$.
The usual \emph{head} and \emph{tail} functions
are
$\hd \deq \lambda s.\, \pi_1 (\unfold s) : \Stream\PTbis \arrow \PTbis$
and
$\tl \deq \lambda s.\, \pi_2 (\unfold s) : \Stream\PTbis \arrow \Stream\PTbis$.

The type of \emph{binary trees over $\PTbis$} is
$\Tree\PTbis \deq \rec\TV.\, \PTbis \times (\TV \times \TV)$.
The constructor
$\Node : \PTbis \arrow \Tree\PTbis \arrow \Tree\PTbis \arrow \Tree\PTbis$
and the destructors
$\lbl : \Tree\PTbis \arrow \PTbis$
and
$\lft, \rght : \Tree\PTbis \arrow \Tree\PTbis$
are defined similarly as resp.\ $\Cons, \hd, \tl$ on streams.
\end{example}

\begin{table}[t!]
\[
\begin{array}{c}
\toprule

\begin{array}{r !{~}c!{~} l !{~}c!{~} l}
  \bm{\map}
& \deq
& \lambda f. \fix g. \lambda x.~
  (f\ (\hd x)) \Colon (g\ (\tl x))

& :
& (\PT \arrow \PTbis)
  \longarrow
  \Stream\PT
  \longarrow
  \Stream\PTbis

\end{array}

\\\midrule

\begin{array}{r !{~}l!{~} l}
  \bm{\filter}
& :
& (\PTbis \arrow \Bool)
  \longarrow
  \Stream\PTbis
  \longarrow
  \Stream\PTbis
\\

& \deq
& \lambda p. \fix g.\lambda x.~
  \term{if}~ (p\ (\hd x))
  ~\term{then}~ (\hd x) \Colon (g\ (\tl x))
  ~\term{else}~ (g\ (\tl x))
\end{array}

\\\midrule

\begin{array}{r !{~}l!{~} l !{~}c!{~} l}
  \bm{\diag}
& \deq
& \diagaux (\lambda x.x)
& :
& \Stream(\Stream \PTbis) \longarrow \Stream \PTbis
\\[0.5em]

  \bm{\diagaux}
& :
& \multicolumn{3}{l}{
  (\Stream\PTbis \arrow \Stream\PTbis)
  \longarrow
  \Stream (\Stream \PTbis)
  \longarrow
  \Stream \PTbis
  }
\\

& \deq
& \multicolumn{3}{l}{
  \fix g. \lambda k. \lambda x.~
  \big( (\hd \comp\, k) (\hd x) \big)
  \Colon
  \big( g\ (k \comp \tl)\ (\tl x) \big)
  }
\end{array}

\\\midrule

\begin{array}{c !{\qquad} c}

\begin{array}[t]{l !{~}l!{~} l}

  \bm{\extract}
& :
& \Rou \PTbis \longarrow \PTbis
\\
& \deq 
& \fix e.\lambda c. \unfold\, c\ e
\end{array}

&

\begin{array}[t]{l !{~}l!{~} l}
  \bm{\Over}
& :
& \Rou \PTbis
\\
& \deq
& \fix c. \fold(\lambda k. k\ c)
\end{array}

\end{array}

\\\\

\begin{array}{r !{~}l!{~} l}
  \bm{\bft}
& \deq
& \begin{array}{l !{~}c!{~} l}
    \lambda t.~ \extract (\bftaux\ t\ \Over)
  & :
  & \Tree \PTbis
    \longarrow
    \Stream \PTbis
  \end{array}
\\[0.5em]

  \bm{\bftaux}
& :
& \Tree\PTbis
  \longarrow
  \Rou (\Stream\PTbis)
  \longarrow
  \Rou (\Stream\PTbis)
\\

& \deq
& \fix g.\lambda t.\lambda c.
  \fold \left(
    \lambda k.~
    (\lbl t) \Colon
    \left( \unfold\, c \ \big( k \comp (g (\lft t)) \comp (g (\rght t)) \big) \right)
  \right)

\end{array}

\\\bottomrule
\end{array}
\]
\caption{Functions on Streams and Trees.\label{tab:ex}}
\end{table}

\begin{example}
\label{ex:pure:fun}
Table~\ref{tab:ex} defines functions on streams and trees
for which we will be able to derive specifications which improve on \cite{jr21esop}
(see Examples \ref{ex:reft:fun} and \ref{ex:main}).

On streams, besides the usual $\map$ function,
we consider the $\filter$ function from \S\ref{sec:intro}.
This assumes that $\Base$ contains a set $\Bool = \{\term{tt},\term{ff} \}$
of \emph{Booleans}.
The notation
$\term{if}~ M ~\term{then}~ N_{\term{tt}} ~\term{else}~ N_{\term{ff}}$
stands for the term
$\cse\ M\ \copair{a \mapsto N_a \mid a \in \Bool}$.
Finally, the function $\diag$ computes the diagonal of a stream of streams.
We refer to~\cite[Example 8.3]{jr21esop} for explanations.

On trees, the function $\bft$ implements Martin Hofmann's breadth-first traversal
(see e.g.~\cite{bms19types,jr21esop}).
It uses the recursive type
$\Rou\PTbis \deq \rec\TV.\, (\TV \arrow \PTbis) \arrow \PTbis$.
\end{example}

\begin{full}
\begin{remark*} 
We assumed a term former $\fix x.M$ for term-level fixpoints,
but it is well known that $\fix$ is definable in presence of recursive types
(cf.\ e.g.\ \cite[\S 20.1]{pierce02book}).
\end{remark*}
\end{full}

\subsection{Negation-Free Infinitary Modal Logics}
\label{sec:log}

\begin{figure}[t!]
\[
\begin{array}{c}

\dfrac{\varphi \in \Lang(\PT_1)}
  {\form{\pi_1} \varphi \in \Lang(\PT_1 \times \PT_2)}

\qquad\quad

\dfrac{\varphi \in \Lang(\PT_2)}
  {\form{\pi_2} \varphi \in \Lang(\PT_1 \times \PT_2)}

\qquad\quad

\dfrac{\psi \in \Lang(\PTbis)
  \qquad
  \varphi \in \Lang(\PT)}
  {\psi \realto \varphi \in \Lang(\PTbis \arrow \PT)}

\\\\

\dfrac{\varphi \in \Lang(\PT[\rec \TV.\PT/\TV])}
  {\form\fold \varphi \in \Lang(\rec \TV.\PT)}

\qquad\qquad

\dfrac
  {\text{$\BT \in \Base$ and $a \in \BT$}}
  {\form a \in \Lang(\BT)}

\end{array}
\]
\caption{Modalities.%
\label{fig:modal}}
\end{figure}

\noindent
We consider negation-free infinitary formulae with modalities
as in~\cite{abramsky91apal,bk03ic,jr21esop}.

\begin{definition}[Formulae]
\label{def:form}
Let $\PT$ be a pure type.

The formulae $\varphi \in \Lang(\PT)$
are formed using the modalities in Figure~\ref{fig:modal}
together with arbitrary set-indexed 
conjunctions $\bigwedge_{i \in I} \varphi_i$ and
disjunctions $\bigvee_{i \in I} \varphi_i$.
We write $\True$ (resp.\ $\False$) for the empty conjunction (resp.\ disjunction).

We let $\Lang_\land(\PT)$
consist of those $\varphi \in \Lang(\PT)$
in which all conjunctions are finite and all disjunctions are empty
($\False$ is the only disjunction allowed in $\Lang_\land(\PT)$).

The formulae $\varphi \in \Lang_\Open(\PT)$
are formed from formulae in $\Lang_\land(\PT)$
using arbitrary disjunctions and finite conjunctions.

The \emph{normal forms} $\varphi \in \Norm(\PT)$
are the $\varphi = \bigwedge_{i \in I} \bigvee_{j \in J_i}\psi_{i,j}$
with $\psi_{i,j} \in \Lang_\land(\PT)$
and where $I$ and the $J_i$'s are arbitrary sets.
\end{definition}

The semantics of formulae is defined in~\S\ref{sec:sem:log}.
Their intended meaning is as follows.
The formula $\psi \realto \varphi \in \Lang(\PTbis \arrow \PT)$
is intended to select those $M : \PTbis \arrow \PT$
such that $\varphi$ holds on $M N : \PT$ whenever $\psi$ holds on $N : \PTbis$.
Similarly, $\form\fold \varphi$ holds on $M$ whenever
$\varphi$ holds on $\unfold M$.
For $i = 1,2$,
the formula $\form{\pi_i} \varphi$
selects those $M : \PT_1 \times \PT_2$
such that $\varphi$ holds on $\pi_i M$.
With
$\pair{\varphi_1,\varphi_2} \deq \form{\pi_1}\varphi_1 \land \form{\pi_2} \varphi_2$,
we have a formula which holds on those $M$ such that
$\varphi_i$ holds on $\pi_i M$ for $i = 1,2$.

\begin{example}
\label{ex:form:base}
Given $\BT \in \Base$ and $a \in \BT$,
the formula $\form a$ is intended to hold
on $a$ but not on the $b \in \BT \setminus \{a\}$.
For instance, given $\SP \sle \BT$, the formula
$\bigwedge_{a \in \SP}(\form a \realto \form{\term{tt}})$
is intended to select the $p : \BT \arrow \Bool$
such that $(p\, a)$ is $\term{tt}$ for all $a \in \SP$.
\end{example}

\begin{example}
\label{ex:form:stream}
On streams $\Stream\PTbis$, 
the composite modalities $\form\hd$ and $\form\tl$
are defined as $\form\hd \psi \deq \form\fold \form{\pi_1} \psi$
and $\form\tl \varphi \deq \form\fold \form{\pi_2} \varphi$.
Given $\psi \in \Lang(\PTbis)$ and $\varphi \in \Lang(\Stream\PTbis)$,
the formulae $\form\hd\psi \in \Lang(\Stream\PTbis)$
and $\form\tl\varphi \in \Lang(\Stream\PTbis)$
select those streams $M$ such that $\psi$ holds on $(\hd M)$
and such that $\varphi$ holds on $(\tl M)$, respectively.
In the following, we write $\Next\varphi$ for $\form\tl\varphi$.

Using $\NN$-indexed connectives,
we can define the usual $\LTL$ modalities $\Box$ and $\Diam$ as 
$\Box \varphi \deq \bigwedge_{n \in \NN} \Next^n \varphi$
and
$\Diam \varphi \deq \bigvee_{n \in \NN} \Next^n \varphi$.
Hence, $\Box \varphi$ (resp.\ $\Diam \varphi$)
is intended to hold on those $M : \Stream\PTbis$
such that $\varphi$ holds on $\tl^n M$ for all $n \in \NN$
(resp.\ for some $n \in \NN$).
In particular, $\Box\Diam \form\hd \psi$ (resp.\ $\Diam\Box\form\hd\psi$)
selects those streams with infinitely many (resp.\ ultimately all) elements
satisfying $\psi$.
\end{example}

\begin{example}
\label{ex:form:tree}
Similarly, on trees $\Tree\PTbis$
one can define $\form\lbl$, $\form\lft$ and $\form\rght$
such that
$\form\lbl\psi, \form\lft\varphi, \form\rght\varphi \in \Lang(\Tree\PTbis)$
whenever $\psi \in \Lang(\PTbis)$ and $\varphi \in \Lang(\Tree\PTbis)$.

Moreover, the $\LTL$ stream modalities $\Box,\Diam$ have their
usual $\CTL$ counterparts
$\forall\Box$, $\exists\Box$,
$\forall\Diam$ and $\exists\Diam$.
Namely, given $\varphi \in \Lang(\Tree\PTbis)$,
\[
\begin{array}{r l l !{\qquad\quad} r l l}
  \forall \Box \varphi
& \deq
& \bigwedge_{n \in \NN} 
  (\varphi \land \Land(\pl))^n(\True)

& \forall\Diam \varphi
& \deq
& \bigvee_{n \in \NN} 
  (\varphi \lor \Land(\pl))^n(\False)
\\

  \exists \Box \varphi
& \deq
& \bigwedge_{n \in \NN} 
  (\varphi \land \Lor(\pl))^n(\True)

& \exists\Diam \varphi
& \deq
& \bigvee_{n \in \NN} 
  (\varphi \lor \Lor(\pl))^n(\False)
\end{array}
\]

\noindent
with
$\Land\theta \deq \form\lft \theta \land \form\rght \theta$
and
$\Lor\theta \deq \form\lft \theta \lor \form\rght \theta$,
and where
$(\varphi \land \Lor(\pl))^0(\True)$ is $\True$
and
$(\varphi \land \Lor(\pl))^{n+1}(\True)$
is the formula
$\varphi \land \Lor \left( (\varphi \land \Lor(\pl))^{n}(\True) \right)$.

The intended meaning of $\forall\Box \form\lbl\psi$ is to select
those trees whose node labels all satisfy $\psi$,
while $\exists\Box \form\lbl\psi$ asks $\psi$ to hold on all
labels in some infinite path.
The formula $\exists\Diam \form\lbl\psi$
holds if there is a node whose label satisfies $\psi$,
and $\forall\Diam \form\lbl\psi$
requires that every infinite path has a node label on which $\psi$ holds.
\end{example}

Examples~\ref{ex:form:stream} and~\ref{ex:form:tree}
are generalized in Example~\ref{ex:sem:modalmu} (\S\ref{sec:sem:reft})
to (negation-free) least and greatest fixpoints 
in the style of the modal $\mu$-calculus
(see e.g. \cite{bs07chapter,bw18chapter}).

\begin{figure}[t!]
\[
\begin{array}{c}

\dfrac{\psi \thesis \theta
  \quad
  \theta \thesis \varphi}
  {\psi \thesis \varphi}

\quad

\dfrac{a \neq b}
  {\form a \land \form b \thesis_{\BT} \False}

\quad

\ax{D}
\dfrac{}
  {\bigwedge_{i \in I}\bigvee_{j \in J_i} \varphi_{i,j}
  \,\thesis\,
  \bigvee_{f \in \prod_{i \in I} J_i}\bigwedge_{i \in I} \varphi_{i,f(i)}}

\\\\

\dfrac{}
  {\varphi \thesis \varphi}

\quad

\dfrac{\text{for each $i \in I$, $\psi \thesis \varphi_i$}}
  {\psi \thesis \bigwedge_{i \in I} \varphi_i}

\quad

\dfrac{\psi_i \thesis \varphi}
  {\bigwedge_{i \in I} \psi_i \thesis \varphi}
~(i \in I)

\quad

\dfrac{}
  {\bigwedge_{i \in I} \form\triangle \varphi_i
  \thesis
  \form\triangle \bigwedge_{i \in I} \varphi_i}

\\\\

\dfrac{\psi \thesis \varphi_i}
  {\psi \thesis \bigvee_{i \in I} \varphi_i}
~(i \in I)

\qquad\qquad

\dfrac{\text{for each $i \in I$, $\psi_i \thesis \varphi$}}
  {\bigvee_{i \in I}\psi_i \thesis \varphi}

\qquad\qquad

\dfrac{}
  {\form\triangle \bigvee_{i \in I}\varphi_i
  \thesis
  \bigvee_{i \in I} \form\triangle\varphi_i}

\\\\

\ax{F}
\dfrac{\psi \in \Lang_\land(\PTbis)
  \quad~~
  \varphi_i \in \Lang(\PT)
  \quad~~
  I \neq \emptyset}
  {\psi \realto \left( \bigvee_{i \in I}\varphi_i \right)
  \,\thesis\,
  \bigvee_{i \in I} \left( \psi \realto \varphi_i \right) }

\qquad

\dfrac{\psi' \thesis_{\PTbis} \psi
  \qquad
  \varphi \thesis_{\PT} \varphi'}
  {\psi \realto \varphi \,\thesis_{\PTbis \arrow \PT}\, \psi' \realto \varphi'}

\qquad

\dfrac{\psi \thesis \varphi}
  {\form\triangle \psi \thesis \form\triangle \varphi}

\\\\

\dfrac{}
  {\bigwedge_{i \in I}\left(\psi \realto \varphi_i \right)
  \,\thesis\,
  \psi \realto \left(\bigwedge_{i \in I} \varphi_i\right)}

\qquad\qquad

\dfrac{}
  {\bigwedge_{i \in I}\left( \psi_i \realto \varphi \right)
  \,\thesis\,
  \left(\bigvee_{i \in I} \psi_i \right) \realto \varphi}

\end{array}
\]
\caption{Basic deduction rules, where $\triangle$ is either $\pi_1$, $\pi_2$ or $\fold$.%
\label{fig:log:ded}}
\end{figure}

\begin{definition}[Deduction]
A \emph{sequent} has the form $\psi \thesis_{\PT} \varphi$
where $\varphi,\psi \in \Lang(\PT)$.
We often write $\psi \thesis \varphi$ for $\psi \thesis_{\PT} \varphi$.
\emph{Basic deduction} is defined by the rules in
Fig.~\ref{fig:log:ded}.

We write $\psi \thesisiff \varphi$ when the sequents
$\psi \thesis \varphi$ and $\varphi \thesis \psi$
are both derivable.
\end{definition}

Note that $\varphi \thesis \True$ and $\False \thesis \varphi$
by definition of $\True$ and $\False$.
One can derive that $\thesis$ preserves conjunctions
and disjunctions:
if $\psi_i \thesis \varphi_i$ for all $i \in I$,
then
$\bigwedge_{i \in I} \psi_i \thesis \bigwedge_{i \in I} \varphi_i$
and
$\bigvee_{i \in I} \psi_i \thesis \bigvee_{i \in I} \varphi_i$.

\begin{example}
\label{ex:log:modalnf}
Let $\triangle$ be either $\pi_1,\pi_2$ or $\fold$.
The modality $\form\triangle$ commutes over conjunctions and disjunctions
($\bigwedge_i \form\triangle \varphi_i \thesisiff \form\triangle\bigwedge_i \varphi_i$,
and similarly for $\bigvee$).
In particular, for each normal form $\varphi$
there is a normal form $\psi$ such that $\form\triangle \varphi \thesisiff \psi$.
\end{example}

\begin{example}
\label{ex:log:distr}
As usual, the converse of $\ax{D}$ is derivable, and so is the dual
law
\(
  \bigwedge_{f \in \prod_{i \in I}J_i}\bigvee_{i \in I} \varphi_{i,f(i)}
  \,\thesisiff\,
  \bigvee_{i \in I}\bigwedge_{j \in J_i}\varphi_{i,j}
\)
(see e.g. \cite[Lemma VII.1.10]{johnstone82book}).
\end{example}

\begin{remark}
\label{rem:log:realto}
Taking $I = \emptyset$ in the last two rules of Fig.~\ref{fig:log:ded}
yields $\True \thesis \left(\psi \realto \True \right)$
and $\True \thesis \left(\False \realto \varphi \right)$.
The rule $\ax{F}$ would thus be unsound with $I = \emptyset$ and $\psi = \False$.
Rule $\ax{F}$ differs from usual systems for DTLF
(cf.\ \cite[\S 4.2]{abramsky91apal} \cite[Figure 5]{bk03ic}
and \cite[Figure 10.3]{ac98book}).
The case of $I = \emptyset$
will be handled by rule $\ax{C}$ in \eqref{eq:compl:cc}, \S\ref{sec:compl:fin}.
\end{remark}

\subsection{Refinement Types}
\label{sec:reft}

\emph{Refinement types} (or \emph{types}), notation $\RT,\RTbis,\dots$,
are given by the grammar
\[
\begin{array}{r @{\ \ }c@{\ \ } l}
    \RT
&   \bnf
&   \PT
\gs \reft{\PT \mid \varphi}
\gs \RT \times \RT
\gs \RT \arrow \RT
\end{array}
\]

\noindent
where $\PT$ is a pure type and $\varphi \in \Lang(\PT)$.
We shall consider typing judgments of the form
$\Env \thesis M : \RT$,
where $\Env$ is allowed to mention refinement types.
A judgment
$M : \reft{\PT \mid \varphi}$
is intended to mean that $M$ is of pure type $\PT$ and satisfies $\varphi$.

\begin{example}
\label{ex:reft:base}
Given a base type $\BT \in \Base$ and $\SP \sle \BT$,
a judgment of the form
\(
  p :
  \reft{\BT \arrow \Bool
  \mid
  \bigwedge_{a \in \SP}\left(\form a \realto \form{\term{tt}} \right)}
\)
expresses that 
$(p\, a)$ yields $\term{tt}$ for all $a \in \SP$.
\end{example}

\begin{table}[t!]
\begin{center}
\scalebox{0.9}{\(
\begin{array}{c}
\toprule

\multicolumn{1}{l}{\text{\textbf{Map on streams}
(with $\triangle$ either $\Box$, $\Diam$, $\Diam\Box$ or $\Box\Diam$)}}
\\
\begin{array}{*{7}{l}}
  \map
& :
& \reft{\PT \arrow \PTbis \mid \psi \realto \varphi}
& \longarrow
& \reft{\Stream \PT \mid \triangle \form{\hd}\psi}
& \longarrow
& \reft{\Stream \PTbis \mid \triangle \form{\hd}\varphi}
\end{array}

\\\midrule

\multicolumn{1}{l}{\text{\textbf{Filter on streams}
(with $\triangle$ either $\Box$ or $\Box\Diam$)}}
\\
\begin{array}{*{7}{l}}
  \filter
& :
& \reft{\BT \arrow \Bool \mid \bigwedge_{a \in \SP}(\form a \realto \form{\term{tt}})}
& \longarrow
& \reft{\Stream \PTbis \mid \triangle \form{\hd}\bigvee_{a \in \SP}\form a}
& \longarrow
& \reft{\Stream \PTbis \mid \Box \form{\hd}\bigvee_{a \in \SP}\form a}
\end{array}

\\\midrule

\multicolumn{1}{l}{\text{\textbf{Diagonal of streams of streams}
(with $\triangle$ either $\Box$ or $\Diam\Box$)}}
\\
\begin{array}{l l l l l}
  \diag
& :
& \reft{\Stream (\Stream \PTbis) \mid \triangle \form\hd \Box \form{\hd}\varphi}
& \longto 
& \reft{\Stream \PTbis \mid \triangle \form{\hd}\varphi}
\\
\end{array}

\\\midrule

\multicolumn{1}{l}{\text{\textbf{Breadth-first tree traversal}
(see Example~\ref{ex:reft:fun} for $\triangle$ and $\overline\triangle$)}}
\\

\begin{array}{l l r c l}
  \bft
& :
& \reft{\Tree\PTbis \mid \triangle \form\lbl \varphi}
& \longto
& \reft{\Stream\PTbis \mid \overline\triangle \form\hd \varphi}
\end{array}

\\\bottomrule

\end{array}\)}
\end{center}
\caption{Some judgments with refinement types
(functions defined in Table~\ref{tab:ex}).%
\label{tab:reft}}
\end{table}

\begin{example}
\label{ex:reft:fun}
Table~\ref{tab:reft} presents some specifications,
expressed as refinement types, for functions defined in Table~\ref{tab:ex}
(see Example~\ref{ex:pure:fun}).

For the $\map$ function,
assuming $f : \reft{\PTbis \arrow \PT \mid \psi \realto \varphi}$,
if $\triangle$ is $\Box$ (resp.\ $\Diam, \Box\Diam, \Diam\Box$),
then the judgment expresses
that $(\map f)$
takes a stream with all (resp.\ some, infinitely many, ultimately all)
elements satisfying $\psi$ to a stream with all
(resp.\ some, infinitely many, ultimately all)
elements satisfying~$\varphi$.

The specifications for $\filter$ are the expected ones.
Let $p : \BT \to \Bool$ such that
$(p\, a)$ yields $\term{tt}$ for all $a \in \SP$.
If $\triangle$ is $\Box$ (resp.\ $\Box\Diam$)
then the judgment means that
$(\filter p)$ takes a stream with all (resp.\@ infinitely many) elements in $\SP$
to a stream with all elements in $\SP$.
Recalling that the stream formula $\Box \form\hd \form a$ amounts to
$\bigwedge_{n \in \NN} \Next^n \form\hd \form a$,
note that none of the formulae $\Next^n \form\hd \form a$ hold on
$\Omega_{\Stream\BT} : \Stream\BT$.

Concerning the diagonal,
if $\triangle$ is $\Box$ (resp.\ $\Diam\Box$),
then the judgment expresses that $\diag$ takes a stream whose component streams
all (resp.\@ ultimately all) satisfy $\Box \form\hd \varphi$
to a stream whose elements all (resp.\ ultimately all)
satisfy $\varphi$.

For the tree traversal $\bft$ we can allow for any sound
combination of $\triangle$ and $\overline\triangle$.
This includes all pairs $(\triangle,\overline\triangle)$
among
$(\forall \Box, \Box)$,
$(\exists \Box, \Box\Diam)$,
$(\exists \Diam, \Diam)$,
$(\forall \Diam, \Diam)$
and
$(\forall \Box \exists \Diam, \Box\Diam)$.
For instance,
if $\triangle$ is $\forall\Box$
(resp.\ $\exists\Diam, \forall\Box\exists\Diam$),
then
the judgment says that $\bft$
takes a tree with all (resp.\ some, infinitely many) node
labels satisfying $\varphi$ to a stream with all (resp.\ some, infinitely may)
elements satisfying~$\varphi$.
\end{example}

\begin{figure}[t!]
\[
\begin{array}{c}

\dfrac{}
  {\RT \subtype \UPT\RT}

\qquad

\dfrac{}
  {\PT \subtype \reft{\PT \mid \True}}

\qquad

\dfrac{\psi \thesis_{\PT} \varphi}
  {\reft{\PT \mid \psi} \subtype \reft{\PT \mid \varphi}}

\qquad

\dfrac{\RT \subtype \RTbis
  \qquad
  \RTbis \subtype \RTter}
  {\RT \subtype \RTter}

\\\\

\dfrac{\RT \subtype \RT'
  \qquad
  \RTbis \subtype \RTbis'}
  {\RT \times \RTbis \subtype \RT' \times \RTbis'}

\qquad

\dfrac{}{
  \reft{\PT \mid \varphi}
  \times
  \reft{\PTbis \mid \psi}
  \eqtype
  \reft{\PT \times \PTbis \mid \pair{\varphi,\psi}}}

\qquad

\dfrac{}
  {\RT \subtype \RT}

\\\\

\dfrac{\RTbis' \subtype \RTbis
  \qquad
  \RT \subtype \RT'}
  {\RTbis \arrow \RT \subtype \RTbis' \arrow \RT'}

\qquad

\dfrac{}
  {\reft{\PTbis \mid \psi} \arrow \reft{\PT \mid \varphi}
  \eqtype
  \reft{\PTbis \arrow \PT \mid \psi \realto \varphi}}

\end{array}
\]

\caption{Subtyping.%
\label{fig:reft:subtyping}}
\end{figure}

\begin{figure}[t!]
\begin{center}
\scalebox{0.89}{\(
\begin{array}{c}
\dfrac{(x:\RT) \in \Env}
  {\Env \thesis x:\RT}

\qquad\qquad

\dfrac{\Env,x:\RTbis \thesis M : \RT}
  {\Env \thesis \lambda x.M : \RTbis \arrow \RT}

\qquad\qquad

\dfrac{\Env \thesis M : \RTbis \arrow \RT
  \qquad
  \Env \thesis N : \RTbis}
  {\Env \thesis M N : \RT}

\\\\

\dfrac{\Env \thesis M : \RT
  \qquad
  \Env \thesis N : \RTbis}
  {\Env \thesis \pair{M,N} : \RT \times \RTbis}

\qquad\qquad

\dfrac{\Env \thesis M : \RT \times \RTbis}
  {\Env \thesis \pi_1(M) : \RT}

\qquad\qquad

\dfrac{\Env \thesis M : \RT \times \RTbis}
  {\Env \thesis \pi_2(M) : \RTbis}

\\\\

\dfrac{
  \begin{array}{l}
  \UPT\Env \thesis M : \PT
  \\
  \text{for each $i \in I$,}\quad \Env \thesis M : \reft{\PT \mid \varphi_i}
  \end{array}}
  {\Env \thesis M : \reft{\PT \mid \bigwedge_{i \in I} \varphi_i}}

\qquad\quad

\dfrac{
  \begin{array}{l}
  \UPT\Env, x:\PTbis, \UPT{\Env'} \thesis M : \UPT\RT
  \\
  \text{for each $i \in I$,}\quad
  \Env, x:\reft{\PTbis \mid \psi_i},\Env' \thesis M : \RT
  \end{array}}
  {\Env, x : \reft{\PTbis \mid \bigvee_{i \in I} \psi_i} , \Env' \thesis M : \RT}

\\\\

\dfrac{
  \Env' \subtype \Env
  \quad 
  \RT \subtype \RT'
  \quad
  \Env \thesis M : \RT}
  {\Env' \thesis M : \RT'}

\qquad

\dfrac{\Env \thesis \fix x.M : \reft{\PT \mid \psi}
  \quad
  \Env, x: \reft{\PT \mid \psi} \thesis M : \reft{\PT \mid \varphi}}
  {\Env \thesis \fix x.M : \reft{\PT \mid \varphi}}
~(\varphi,\psi \in \Lang_\land)

\\\\

\dfrac{\Env \thesis M : \reft{\PT_1 \times \PT_2 \mid \form{\pi_i} \varphi}}
  {\Env \thesis \pi_i(M) : \reft{\PT_i \mid \varphi}}
~(i=1,2)

\qquad\quad

\dfrac{\Env \thesis M_i : \reft{\PT_i \mid \varphi}
  \qquad
  \Env \thesis M_{3-i} : \PT_{3-i}}
  {\Env \thesis \pair{M_1,M_2} : \reft{\PT_1 \times \PT_2 \mid \form{\pi_i} \varphi}}
~(i=1,2)
\\\\

\dfrac{}
  {\Env \thesis a : \reft{\BT \mid \form a}}

\qquad\quad

\dfrac{
  \Env \thesis M : \reft{\BT \mid \form b}
  \qquad
  \Env \thesis N_b : \RT
  \qquad
  \text{for each $a \in A$,\quad} \UPT\Env \thesis N_a : \UPT\RT}
  {\Env \thesis \cse\ M\ \copair{a \mapsto N_a \mid a \in \BT} : \RT}

\\\\

\dfrac{\Env \thesis M : \reft{\PT[\rec\TV.\PT/\TV] \mid \varphi}}
  {\Env \thesis \fold(M) : \reft{\rec\TV.\PT \mid \form\fold \varphi}}

\qquad\quad

\dfrac{\Env \thesis M : \reft{\rec\TV.\PT \mid \form\fold \varphi}}
  {\Env \thesis \unfold(M) : \reft{\PT[\rec\TV.\PT/\TV] \mid \varphi}}

\end{array}\)}
\end{center}
\caption{Typing with refinement types.%
\label{fig:reft:reftyping}}
\end{figure}

Each refinement type $\RT$ has an \emph{underlying pure type} $\UPT\RT$
defined by induction from $\UPT\PT \deq \PT$
and $\UPT{\reft{\PT \mid \varphi}} \deq \PT$.
We write $\UPT\Env$ for the extension of $\UPT{\pl}$ to $\Env$.

We derive typing judgments $\Env \thesis M : \RT$
using the rules in Figure~\ref{fig:reft:reftyping}
augmented with all the typing rules of the pure system (\S\ref{sec:pure}).
Deduction on formulae (\S\ref{sec:log})
enters the type system via a subtyping relation $\RTbis \subtype \RT$.
Subtyping rules are presented in Figure~\ref{fig:reft:subtyping},
where $\RTbis \eqtype \RT$
stands for the conjunction of $\RTbis \subtype \RT$ and $\RT \subtype \RTbis$.
Note that for each $\PT$ we have $\PT \eqtype \reft{\PT \mid \True}$.
Subtyping is extended to typing contexts:
given $\Env = x_1:\RTbis_1,\dots,x_n:\RTbis_n$
and $\Env' = x_1:\RTbis'_1,\dots,x_n:\RTbis'_n$,
we let $\Env \subtype \Env'$ when $\RTbis_i \subtype \RTbis'_i$
for all $i =1,\dots,n$.
Note that if $\Env \thesis M :\RT$ is derivable
then so is $\UPT\Env \thesis M : \UPT\RT$.

The rules in Figures~\ref{fig:reft:reftyping} and~\ref{fig:reft:subtyping}
are direct adaptations of those in~\cite{abramsky91apal,bk03ic,jr21esop}.
In particular, the rule for $\fix$
(in which $\varphi,\psi \in \Lang_\land(\PT)$)
comes from~\cite{abramsky91apal}.

\begin{example}
The following rules are derived using the 
last rule in Figure~\ref{fig:reft:subtyping}.
\[
\begin{array}{c}
\dfrac{\Env,x : \reft{\PTbis \mid \psi} \thesis M : \reft{\PT \mid \varphi}}
  {\Env \thesis \lambda x.M : \reft{\PTbis \arrow \PT \mid \psi \realto \varphi}}

\quad

\dfrac{\Env \thesis M : \reft{\PTbis \arrow \PT \mid \psi \realto \varphi}
  \quad
  \Env \thesis N : \reft{\PTbis \mid \psi}}
  {\Env \thesis M N : \reft{\PT \mid \varphi}}
\end{array}
\]
\end{example}

\begin{restatable}{lemma}{LemReft}
\label{lem:reft}%
For each type $\RT$, there is a $\varphi \in \Lang(\UPT\RT)$
such that $\RT \eqtype \reft{\UPT\RT \mid \varphi}$.
\end{restatable}

Our goal is to devise extensions of this type system which
are sound and complete w.r.t.\ the usual Scott semantics,
the sense that given $\thesis M : \PT$,
\[
\begin{array}{c !{\quad}c!{\quad} c}
  \thesis M : \reft{\PT \mid \varphi}
& \text{if, and only if,}
& \text{$\varphi$ holds on $\I M$ in the Scott semantics.}
\end{array}
\]

\noindent
The Scott semantics is recalled in~\S\ref{sec:sem},
while \S\ref{sec:compl} discusses completeness.
In particular, all 
typing judgments in Table~\ref{tab:reft} (Example~\ref{ex:reft:fun})
will be derivable.
Those for $\filter$ and $\bft$ improve on \cite{jr21esop}
(see Example~\ref{ex:main}).

\section{Semantics}
\label{sec:sem}

\subsubsection{Scott Domains.}
\label{sec:scott}
We shall interpret pure types as Scott domains and terms as Scott-continuous functions.
We mostly use the terminology of~\cite[\S 1]{ac98book}.
A \emph{dcpo} is a poset with all directed suprema.
A \emph{cpo} is dcpo with a least element (often denoted $\bot$).
A function between dcpos is \emph{Scott-continuous}
if it preserves the order (i.e.\ is monotone) as well as directed suprema.
A Scott-continuous function is \emph{strict} if it preserves
least elements.

\begin{definition}[Scott Domain]
\label{def:scott}
A \emph{Scott domain} is a bounded-complete algebraic cpo.
$\Scott$ is the category of Scott domains and Scott-continuous functions.
\end{definition}

Recall that a cpo $X$ is bounded-complete if
any two $x,y \in X$ have a sup (or \emph{least} upper bound)
$x \vee y \in X$ whenever they have an upper bound in $X$.

An element $x$ of a dcpo $X$ is \emph{finite}
if for all directed $D \sle X$ such that $x \leq \bigvee D$,
we have $d \in \down D$ (i.e.\ $x \leq d$ for some $d \in D$).%
\footnote{Finite elements are called \emph{compact} in~\cite{ac98book}.}
Note that $\bot$ is always finite,
and that if $d,d' \in X$ are finite,
then $d \vee d'$ is finite whenever it exists.
A dcpo $X$ is \emph{algebraic} if for each $x \in X$,
the set $\{ d\in X \mid \text{$d$ finite and $\leq x$} \}$
is directed and has sup $x$.

The category $\Scott$ is Cartesian-closed
(see e.g.\ \cite[Corollary 4.1.6]{aj95chapter}).

\subsubsection{Semantics of the Pure System.}
\label{sec:sem:pure}
Typed terms $\Env \thesis M : \PT$ of the pure system (\S\ref{sec:pure})
are interpreted as morphisms $\I M \colon \I\Env \to \I\PT$ in $\Scott$,
where $\I{\Env} = \prod_{i=1}^n\I{\PTbis_i}$ when
$\Env = x_1:\PTbis_1,\dots,x_n:\PTbis_n$.
This is well-known.

Base types $\BT \in \Base$ are interpreted as \emph{flat domains}
$\I\BT \deq \BT_\bot$,
where
$\BT_\bot$ is $\BT+\{\bot\}$ with $\BT$ discrete.
For each $a \in \BT$, we let
$\I{a} \colon \one \to \I{\BT}$
be the constant map of value $a$.
The term 
$\cse\ M\ \copair{a \mapsto N_a \mid a \in \BT}$
is interpreted using the strict Scott-continuous function
which takes $b \in \BT$ and $(y_a)_a \in X^{\BT}$ to $y_b$.

\begin{full}
Product types $\PT \times \PTbis$ are interpreted using
the Cartesian product of $\Scott$, i.e.\ the Cartesian product
of sets equipped with component-wise order.
Arrow types $\PTbis \to \PT$ are interpreted using the
closed structure of $\Scott$,
given by equipping each homset $\Scott(X,Y)$ with the pointwise order.
\end{full}

We refer to~\cite{ac98book,aj95chapter,streicher06book}
for the interpretation of recursive types $\rec\TV.\PT$.%
\opt{full}{\footnote{See Appendix~\ref{sec:proof:sem:pure} for details.}}%
\opt{short}{\footnote{See also \cite{rk25full} for details.}}%

Term-level fixpoints $\fix x.M$ are interpreted
using the usual fixpoint combinators $\term Y \colon (X \to X) \to X$
taking $f \colon X \to X$ to $\term Y(f) \deq \bigvee_{n \in \NN} f^n(\bot)$.

\begin{example}
\label{ex:scott:stream-tree}
The domain $\I{\Stream\PTbis}$
of streams 
(resp.\ $\I{\Tree\PTbis}$ of trees)
is $\I\PTbis^K$ equipped with the pointwise order,
where $K = \NN$ (resp.\ $K = \two^*$).
The finite elements are those $z \in \I\PTbis^K$ such that $z(p)$ is finite in
$\I\PTbis$ for all $p \in K$,
and $z(p) \neq \bot$ for at most finitely many $p \in K$.

Given $x \in \I{\Stream\PTbis}$,
we have $\I\hd(x) = x(0)$
while $\I\tl(x)$ is the stream taking $n \in \NN$ to $x(n+1) \in \I\PTbis$.
Moreover, $x = \I\Cons(\I\hd(x),\I\tl(x))$.%
\footnote{Note that $\I{\Stream\BT}$ differs from the usual \emph{Kahn domain}
$\BT^* \cup \BT^\omega$
(see e.g.~\cite[Definition 3.7.5 and Example 5.4.4]{vickers89book}
or~\cite[\S 7.4]{dst19book}, see also~\cite{vvk05concur}).}

Similarly, if $y \in \I{\Tree\PTbis}$
then $\I\lbl(y) = y(\es)$ is the root label of $y$,
while $\I\lft(y)$ and $\I\rght(y)$ are the left- and right-subtrees
of $y$, respectively.
\end{example}

\subsubsection{Scott Topology.}
\label{sec:top}
The semantics of refinement types involves some topology.
We refer to e.g.\ \cite[\S 1.2]{ac98book},
\cite[\S 2.3]{aj95chapter} or \cite[\S 7.1]{gg24book}.
See also~\cite{rs24jfla}.

Let $(X,\leq)$ be a dcpo.
A set $\SP \sle X$ is \emph{Scott-open}
if $\SP$ is upward-closed
(if $x \in \SP$ and $x \leq y$ in $X$, then $y \in \SP$),
and if moreover $\SP$ is inaccessible by directed sups,
in the sense that if $\bigvee D \in \SP$
with $D \sle X$ directed, then $D \cap \SP \neq \emptyset$.
This equips $X$ with a topology, called the \emph{Scott topology}.%
\footnote{Moreover, we have $x \leq y$ if, and only if,
$x \in \SP$ implies $y \in \SP$ for every Scott-open $\SP$.}
A function between dcpos is Scott-continuous
precisely when it is continuous for the Scott topology.

If $X$ is algebraic, then the Scott-opens are exactly the unions of sets of the form
$\up d = \{ x \in X \mid d \leq x\}$, with $d$ finite in $X$.
Note that $\up d$ is a compact subset of $X$ when $d$ is finite in $X$.
If $X$ is a Scott domain, then $\up d \cap \up d'$
is compact for all finite $d,d' \in X$
(by bounded-completeness,
if $\up d \cap \up d'$ is non-empty, then $d \vee d'$
is defined, finite and such that $\up (d \vee d') = \up d \cap \up d'$).%
\footnote{It is well-known that Scott domains are \emph{coherent} topological spaces
(see \cite[Proposition 4.2.17, \S 4.2.3]{aj95chapter},
and also \cite[Definition 5.2.21]{goubault13book}
and \cite[\S 2.3]{gg24book}).}

A set $\SP \sle X$ is \emph{saturated}
if $\SP$ is upward-closed,
or equivalently if $\SP$ is an intersection of Scott-open sets
(see e.g.\ \cite[Proposition 4.2.9]{goubault13book}).

\subsubsection{Semantics of Formulae.}
\label{sec:sem:log}
For each $\varphi \in \Lang(\PT)$ we define a set
$\I\varphi \sle \I\PT$
using the following \emph{semantic modalities}:
$\I{\form a} \deq \{a\} \sle \I\BT$
for $\BT \in \Base$ and $a \in \BT$,
and
\begin{equation*}
\begin{array}{r !{~}c!{~} l l l}
  \SP \in \Po(\I{\PT_i})
& \longmapsto
& \I{\form{\pi_i}}(\SP)
& \deq
& \left\{
  x \in \I{\PT_1 \times \PT_2}
  \mid
  \pi_i(x) \in \SP
  \right\}
\\

  \SP \in \Po(\I{\PT[\rec\TV.\PT/\TV]})
& \longmapsto
& \I{\form\fold}(\SP)
& \deq
& \left\{
  x \in \I{\rec\TV.\PT}
  \mid
  \I\unfold(x) \in \SP
  \right\}
\\

  \SP \in \Po(\I\PTbis)
  \,,\,
  \SPbis \in \Po(\I\PT)
& \longmapsto
& (\SP \realto \SPbis)
& \deq
& \left\{
  f \in \I{\PTbis \arrow \PT}
  \mid
  \forall x \in \SP,~ f(x) \in \SPbis
  \right\}
\end{array}
\end{equation*}

\noindent
We let
$\I{\form{\pi_i}\varphi} \deq \I{\form{\pi_i}}(\I\varphi)$,
$\I{\form{\fold}\varphi} \deq \I{\form{\fold}}(\I\varphi)$,
and
$\I{\psi \realto \varphi} \deq \I\psi \realto \I\varphi$.
Conjunctions and disjunctions are interpreted as intersections
and unions. 

\begin{example}
\label{ex:sem:modalmu}
Assume given \emph{propositional variables}
$p^\PT,\dots$ for each pure type $\PT$.
If a formula $\varphi(p^\PT)$ of type $\PT$ is positive in $p^\PT$,
then it induces a monotone function on $(\Po(\I\PT),\sle)$
with least and greatest fixpoints
\(
  \I{\mu p.\varphi}
  =
  \llbracket \bigvee_{\alpha \leq |\Po(\I\PT)|} \varphi^\alpha(\False) \rrbracket
\)
and
\(
  \I{\nu p.\varphi}
  =
  \llbracket\bigwedge_{\alpha \leq |\Po(\I\PT)|} \varphi^\alpha(\True) \rrbracket
\)
\cite[\S 20]{gtw02alig}.
This generalizes Examples~\ref{ex:form:stream}, \ref{ex:form:tree}.
\end{example}

Lemmas \ref{lem:top:char:fin}, \ref{lem:top:char}
below are semantic characterizations of 
the classes of formulae in Definition~\ref{def:form} (\S\ref{sec:log}).
This yields the soundness of the rule $\ax{F}$
in Figure~\ref{fig:log:ded} (\S\ref{sec:log}).

\begin{restatable}{lemma}{LemTopCharFin}
\label{lem:top:char:fin}%
Given $\varphi \in \Lang_\land(\PT)$,
if $\I\varphi \neq \emptyset$ then
$\I\varphi = \up d$ for some finite $d \in \I\PT$.
Conversely, if $d \in \I\PT$ is finite, then $\up d = \I\varphi$ for some
$\varphi \in \Lang_{\land}(\PT)$.
\end{restatable}

\begin{restatable}{lemma}{LemTopChar}
\label{lem:top:char}%
A set $\SP \sle \I\PT$ is saturated (resp.\ Scott-open)
if, and only if, there is a formula $\varphi \in \Lang(\PT)$
(resp.\@ $\varphi \in \Lang_\Open(\PT)$)
such that $\SP = \I\varphi$.

In particular, for each $\varphi \in \Lang(\PT)$
we have $\I\varphi = \I\psi$ for some $\psi \in \Norm(\PT)$.
\end{restatable}

\begin{restatable}[Soundness of Deduction]{proposition}{PropSemSoundDed}
\label{prop:sem:sound:ded}%
If $\psi \thesis \varphi$ is derivable in the basic deduction system 
in Figure~\ref{fig:log:ded} (\S\ref{sec:log}),
then $\I\psi \sle \I\varphi$.
\end{restatable}

\begin{proof}
We only detail the case of $\ax{F}$.
If $\I\psi = \emptyset$, then for all $\SP \sle \I\PT$
we have $\I\psi \realto \SP = \I{\PTbis \arrow \PT}$,
and we are done since $I$ is assumed to be non-empty.

Otherwise, we have $\I\psi = \up d$ by Lemma~\ref{lem:top:char:fin}.
Let $f \in \I{\PTbis \arrow \PT}$.
If $\up d$ is included in
$f^{-1}(\I{\bigvee_i \varphi_i}) = \bigcup_{i} f^{-1}(\I{\varphi_i})$,
then $d \in f^{-1}(\I{\varphi_i})$ for some $i$.
Hence $\up d \sle f^{-1}(\I{\varphi_i})$
as $f^{-1}(\I{\varphi_i})$ is saturated
($f^{-1}(\I{\varphi_i})$ is saturated
since $f$ is monotone and since $\I{\varphi_i}$ is saturated
by Lemma~\ref{lem:top:char}).
\qed
\end{proof}

\subsubsection{Semantics of Refinement Types.}
\label{sec:sem:reft}
The interpretation $\I\RT \sle \I{\UPT\RT}$
of a type $\RT$ is defined as
$\I{\reft{\PT \mid \varphi}} \deq \I\varphi$,
$\I{\RT \times \RTbis} \deq \I\RT \times \I\RTbis$
and
$\I{\RTbis \arrow \RT} \deq \I\RTbis \realto \I\RT$.

\begin{definition}[Sound Typing Judgement]
\label{def:sound:typing}
A judgment $\Env \thesis M : \RT$ with
$\Env = x_1:\RTbis_1,\dots,x_n:\RTbis_n$
is \emph{sound} if $\UPT\Env \thesis M :\UPT\RT$ is derivable and
if moreover $\I M(u_1,\dots,u_n) \in \I\RT$
whenever $u_i \in \I{\RTbis_i}$ for all $i =1,\dots,n$.
\end{definition}

\noindent
The judgments in
Tab.~\ref{tab:reft} (Ex.~\ref{ex:reft:fun}) are sound.
Also, derivable judgments are sound.

\begin{theorem}[Soundness of Typing]
\label{thm:sem:sound:reft}
If $\Env \thesis M : \RT$ is derivable in the system of~\S\ref{sec:reft},
then $\Env \thesis M : \RT$ is sound.
\end{theorem}

\section{Completeness}
\label{sec:compl}

\subsubsection{The Finite Case.}
\label{sec:compl:fin}
Since the rule $\ax{F}$ assumes $I\neq \emptyset$,
it does not allow us to derive 
$(\psi \realto \False) \thesis \False$.
This sequent is sound only when $\I\psi \neq \emptyset$.
In \cite{abramsky91apal},
Abramsky introduced \emph{coprimeness predicates} 
which select those finite $\varphi$ with $\I\varphi \neq \emptyset$.
We extend our basic deduction system (Figure \ref{fig:log:ded} in \S\ref{sec:log})
with the predicate $\C$ and the rules in eq.\ \eqref{eq:compl:cc} below.
Recall that $\pair{\varphi,\psi} = \form{\pi_1}\varphi \land \form{\pi_2}\psi$.
\begin{equation}
\label{eq:compl:cc}
\scalebox{.9125}{\text{$\begin{array}{c}

\dfrac{}{\C(\True)}

\quad

\dfrac
  {\text{$\BT \in \Base$ and $a \in \BT$}}
  {\C(\form a)}

\quad

\dfrac{\C(\varphi)}
  {\C(\form\fold \varphi)}

\quad

\dfrac{\C(\varphi) 
  \quad
  \C(\psi)}
  {\C(\pair{\varphi,\psi})}

\quad

\dfrac{\C(\psi)
  \quad
  \psi \thesis \varphi
  \quad
  \varphi \in \Lang_\land}
  {\C(\varphi)}
\\\\

\ax{C}
\dfrac{\C(\psi)}
  {(\psi \realto \False) \,\thesis\, \False}

\quad

\dfrac{\begin{array}{l}
  \text{$I$ finite and $\forall i \in I$,}~
  \C(\psi_i) 
  ~\text{and}~
  \C(\varphi_i) ;
  \\
  \text{$\forall J \sle I$,}~
  \bigwedge_{j \in J} \psi_j \thesis \False
  ~~\text{or}~~
  \C\left( \bigwedge_{j \in J} \varphi_j \right)
  \end{array}}
  {\C\left( \bigwedge_{i \in I}(\psi_i \realto \varphi_i) \right)}
\end{array}$}}
\end{equation}

\noindent
In contrast with \cite{abramsky91apal,bk03ic,ac98book},
our $\C$ is a consistency predicate rather than a coprimeness predicate.
Note that the clauses defining $\C$ are positive.%
\footnote{Compare with \cite[Figure 3]{bk03ic} and \cite[Figure 10.3]{ac98book}.}

\begin{restatable}{proposition}{PropComplFinDed}
\label{prop:compl:fin:ded}
In the extension of Figure~\ref{fig:log:ded} (\S\ref{sec:log})
with eq.\ \eqref{eq:compl:cc}:
\begin{enumerate}[(1)]
\item
for all $\varphi,\psi \in \Lang_\land(\PT)$,
we have
$\psi \thesis_\PT \varphi$
if, and only if,
$\I\psi \sle \I\varphi$;

\item
for all $\varphi \in \Lang_\land$,
we have
$\C(\varphi)$
if, and only if,
$\I\varphi \neq \emptyset$.
\end{enumerate}
\end{restatable}

\noindent
In particular, for each $\varphi \in \Lang_\land$,
either $\C(\varphi)$ or $\varphi \thesis \False$ is derivable.

A type is \emph{finite} 
if it only contains formulae $\varphi \in \Lang_\land$.
A typing context $x_1:\RTbis_1,\dots,x_n : \RTbis_n$
is finite if so are all $\RTbis_i$'s.
Completeness for finite types can be obtained
from minor adaptations to \cite{abramsky91apal}.

\begin{restatable}[Abramsky \cite{abramsky91apal}]{theorem}{ThmComplFin}
\label{thm:compl:fin}%
Assume $\Env$ and $\RT$ are finite.
If $\Env \thesis M : \RT$ is sound,
then $\Env \thesis M : \RT$ is derivable in the system of \S\ref{sec:reft}
extended with eq.\ \eqref{eq:compl:cc}.
\end{restatable}

\subsubsection{Well-Filteredness.}
\label{sec:wf}
Following~\cite{bk03ic},
completeness for types with infinitary formulae relies on the fact that
Scott domains are \emph{well-filtered} spaces.
The latter is stated in \cite[Corollary 7.1.11]{aj95chapter}
and \cite[Proposition 8.3.5]{goubault13book} as a consequence
of the Hofmann-Mislove (or Scott-open filter) Theorem.
It can also be obtained from \cite[Theorem 7.38]{gg24book}.
A subset $F$ of a poset $P$ is \emph{filtering} if $F$ is directed in $P^\op$.

\begin{proposition}[Well-Filteredness]
\label{prop:wf}
Let $X$ be an algebraic dcpo,%
\footnote{More generally, this result holds for any sober space $X$
(with $\SP$ open in $X$).}
and let $\Filt$ be a set of compact saturated subsets of $X$.
If $\Filt$ is filtering in $\Po(X)$
and $\bigcap \Filt \sle \SP$ for some Scott-open $\SP$,
then $Q \sle \SP$ for some $Q \in \Filt$.
\end{proposition}

Proposition~\ref{prop:wf} yields the soundness of the following deduction rule.
\[
\ax{WF}~
\dfrac{\text{for all $i\in I$, $\psi_i \in \Lang_\land(\PTbis)$}
  \qquad
  \varphi \in \Lang_\Open(\PT)}
  {
  \left( \bigwedge_{i \in I} \psi_i \right)
  \realto
  \varphi
  \,\thesis\,
  \bigvee_{\text{$J \sle I$, $J$ finite}}
  \left(
  \left( \bigwedge_{j \in J} \psi_j \right)
  \realto
  \varphi
  \right)
  }
\]

\begin{restatable}{lemma}{LemWf}
\label{lem:wf}%
The rule $\ax{WF}$ is sound.
\end{restatable}

\subsubsection{Main Results.}
\label{sec:main}
Theorem~\ref{thm:main} below
gives sufficient conditions for the completeness of the system in \S\ref{sec:reft}
extended with eq.\ \eqref{eq:compl:cc}.
This relies on Well-Filteredness (Proposition~\ref{prop:wf}),
but does not use the rule $\ax{WF}$.
In \S\ref{sec:conc}, we discuss why it might be useful for future work
to avoid that rule.
Proofs of Lemma~\ref{lem:compl:nf} and Theorem~\ref{thm:main}
are given in Appendix \ref{sec:app:main}.

\begin{restatable}{lemma}{LemComplNf}
\label{lem:compl:nf}%
Given $\varphi,\psi \in \Norm(\PT)$,
if $\I\psi \sle \I\varphi$,
then $\psi \thesis_\PT \varphi$
is derivable in the extension of Figure~\ref{fig:log:ded} (\S\ref{sec:log})
with eq.\ \eqref{eq:compl:cc}.
\end{restatable}

\begin{definition}
A type is \emph{normal} if it is pure
or $\reft{\PT \mid \varphi}$ with $\varphi \in \Norm(\PT)$.
A typing context $x_1:\RTbis_1,\dots,x_n : \RTbis_n$
is \emph{normal} if so are all $\RTbis_i$'s.

The \emph{first-order over normal forms} (\emph{fonf}) types
are generated by the grammar
\[
\begin{array}{r @{\ \ }c@{\ \ } l}
    \RT
&   \bnf
&   \RTbis
\gs \RT \times \RT
\gs \RTbis \arrow \RT
\end{array}
\]

\noindent
with $\RTbis$ normal.
A judgment $\Env \thesis M : \RT$ is \emph{normal}
if $\Env$ is normal and $\RT$ is fonf.
\end{definition}

We shall see that if $\Env \thesis M : \RT$ is sound and normal,
then it is derivable.
The idea is to reduce to the finite case (Theorem~\ref{thm:compl:fin})
by using Proposition~\ref{prop:wf},
but without using the rule $\ax{WF}$.
We first show that $\RT$ can be assumed to be normal.
To each normal judgment $\Env \thesis M : \RT$
we associate a set of normal judgments $\eta(\Env \thesis M : \RT)$.
We let
\(
  \eta\left(\Env \thesis M : \RT \right)
  \deq
  \left\{\Env \thesis M : \RT \right\}
\)
if $\RT$ is normal, and
\[
\begin{array}{r !{\quad\deq\quad} l}

  \eta\left(\Env \thesis M : \RT_1 \times \RT_2 \right)
& \eta\left(\Env \thesis \pi_1 M : \RT_1 \right)
  \cup
  \eta\left(\Env \thesis \pi_2 M : \RT_2 \right)
\\

  \eta\left(\Env \thesis M : \RTbis \arrow \RT \right)
& \eta\left( \Env, x: \RTbis \thesis M x : \RT \right)

\end{array}
\]

\noindent
Note that for each $(\Env' \thesis M' : \RT') \in \eta(\Env \thesis M : \RT)$,
the type $\RT'$ is normal.

\begin{restatable}{proposition}{PropMainEta}
\label{prop:main:eta}%
A normal judgment $\Env \thesis M : \RT$ is sound (resp.\ derivable)
if, and only if, so are all $(\Env' \thesis M' : \RT') \in \eta(\Env \thesis M : \RT)$.
\end{restatable}

\begin{restatable}[Main Result]{theorem}{ThmMain}
\label{thm:main}%
If $\Env \thesis M : \RT$ is sound and normal
then $\Env \thesis M : \RT$ is derivable in the system of \S\ref{sec:reft}
extended with eq.\ \eqref{eq:compl:cc}.
\end{restatable}

\begin{example}
\label{ex:main}
Using Examples \ref{ex:log:modalnf} and \ref{ex:log:distr},
the judgments for $\filter$, $\diag$
and $\bft$ in Table~\ref{tab:reft} (Example~\ref{ex:reft:fun})
can be assumed to be normal whenever so is $\varphi$.
Hence our Main Theorem~\ref{thm:main} applies and these
judgments are derivable 
in the system of \S\ref{sec:reft} extended with eq.\ \eqref{eq:compl:cc},
but without the rule $\ax{WF}$.
This improves on \cite{jr21esop},
which does not handle $\filter$,
and which handles $\bft$ only when $\triangle$ is $\forall\Box$.

As for $\map$,
one has to assume that $\psi \in \Lang_\land$ (in addition to $\varphi \in \Norm$).
\end{example}

\subsubsection{The General Case.}
\label{sec:compl:general}
Using $\ax{WF}$ and Example~\ref{ex:log:modalnf},
any formula is \emph{provably} equivalent to a $\psi \in \Norm$.
This yields the completeness result of Bonsangue \& Kok \cite{bk03ic}.

\begin{restatable}{lemma}{LemComplNfWf}
\label{lem:compl:nf:wf}
For each $\varphi \in \Lang(\PT)$, there is a $\psi \in \Norm(\PT)$
such that $\varphi \thesisiff \psi$
in the extension of Figure~\ref{fig:log:ded} (\S\ref{sec:log})
with eq.\ \eqref{eq:compl:cc} and $\ax{WF}$.
\end{restatable}

\begin{corollary}[Bonsangue \& Kok \cite{bk03ic}]
If $\Env \thesis M : \RT$ is sound 
then $\Env \thesis M : \RT$ is derivable in the system of \S\ref{sec:reft}
extended with eq.\ \eqref{eq:compl:cc} and $\ax{WF}$.
\end{corollary}

\section{Future Work}
\label{sec:conc}

We think of the present infinitary system as an intermediary
between denotational semantics and finitary type systems
in the style of \cite{jr21esop}.
In the latter, the logic uses fixpoints in the spirit of the modal $\mu$-calculus
(cf.\ Example~\ref{ex:sem:modalmu}).
When fixpoints are \emph{alternation-free}%
\footnote{This corresponds to
``alternation depth $1$'' in~\cite[\S 2.2 \& \S 4.1]{bw18chapter}.
See also~\cite[\S 7]{bs07chapter} and~\cite{sv10apal}.}
(which includes $\LTL$ on $\Stream\BT$ and $\CTL$ on $\Tree \BT$),
their semantics can computed by iteration up to $\omega$.
In order to reason syntactically over (finite) unfoldings of alternation-free fixpoints,
the system of \cite{jr21esop} uses a term language over natural
numbers (with quantifications over these).

We target a similar finitarization of our system,
in which alternation-free fixpoints $\mu p.\varphi(p)$
and $\nu p.\varphi(p)$ would be seen as
$(\exists k)\varphi^k(\False)$ and $(\forall k)\varphi^k(\True)$.
Rules $\ax{WF}$ and $\ax{D}$ may turn out to be problematic.
Our Main Theorem~\ref{thm:main} shows that $\ax{WF}$
is not needed in an interesting range of cases.
On the other hand, in view of Example~\ref{ex:log:distr}
we think rule $\ax{D}$ could be handled (under appropriate assumptions)
using enough fresh Skolem symbols, as in
\[
\tag{$f$ fresh function symbol}
\begin{array}{c}
\dfrac{(\forall k)\psi(k,f(k)) ~\thesis~ \varphi}
  {(\forall k)(\exists \ell)\psi(k,\ell) ~\thesis~ \varphi}

\qquad

\dfrac{\psi ~\thesis~ (\exists k)\varphi(k,f(k))}
  {\psi ~\thesis~ (\exists k)(\forall \ell)\varphi(k,\ell)}

\end{array}
\]

Further, we expect to handle alternation-free modal $\mu$-properties
on (finitary) polynomial types, thus targeting a system which as a whole
would be based on $\FPC$.
But polynomial types involve sums,
and sums are not universal in $\Scott$.%
\footnote{See e.g. \cite[Exercise 6.1.10]{ac98book}.}
We think of working with Call-By-Push-Value (CBPV)~\cite{levy03book,levy22siglog}
for the usual adjunction between dcpos and cpos with strict functions.
In the long run, it would be nice if this basis could extend to
enriched models of CBPV,
so as to handle further computational effects.
Print and global store are particularly relevant,
as an important trend in proving temporal properties
considers programs generating streams of events.
Major works in this line include
\cite{ssv08jfp,hc14lics,hl17lics,nukt18lics,kt14lics,ust17popl,nukt18lics,%
su23popl}.
In contrast with ours, these approaches are based on trace semantics
of syntactic expressions rather than denotational domains.%
\footnote{See e.g.~\cite[Theorem 4.1 (and Figure 6)]{nukt18lics}
or~\cite[Theorem 1 (and Definition 20 from the full version)]{su23popl}.}

In a different direction, we think the approach of this paper
could extend to linear types~\cite{hjk00mscs,nw03concur,winskel04llcs},
possibly relying on the categorical study of~\cite{bf06book}.

\subsubsection{Acknowledgments.}
We thank the anonymous referees for constructive comments which have hopefully helped
to improve the presentation of the paper.
Research supported by the ANR project
\href{https://anr.fr/Projet-ANR-20-CE48-0005}{\texttt{QuaReMe}} (ANR-20-CE48-0005).

\appendix
\section{Additional Material}
\label{sec:app}

\begin{figure}[t!]
\[
\begin{array}{c}
\dfrac{(x:\PT) \in \Env}
  {\Env \thesis x:\PT}

\qquad\quad

\dfrac{\Env,x:\PTbis \thesis M : \PT}
  {\Env \thesis \lambda x.M : \PTbis \arrow \PT}

\qquad\quad

\dfrac{\Env \thesis M : \PTbis \arrow \PT
  \qquad
  \Env \thesis N : \PTbis}
  {\Env \thesis M N : \PT}

\\\\

\dfrac{\Env \thesis M : \PT
  \qquad
  \Env \thesis N : \PTbis}
  {\Env \thesis \pair{M,N} : \PT \times \PTbis}

\qquad\quad

\dfrac{\Env \thesis M : \PT \times \PTbis}
  {\Env \thesis \pi_1(M) : \PT}

\qquad\quad

\dfrac{\Env \thesis M : \PT \times \PTbis}
  {\Env \thesis \pi_2(M) : \PTbis}

\\\\

\dfrac{\Env,x:\PT \thesis M : \PT}
  {\Env \thesis \fix x.M : \PT}

\qquad\quad

\dfrac{\Env \thesis M : \PT[\rec\TV.\PT/\TV]}
  {\Env \thesis \fold(M) : \rec\TV.\PT}

\qquad\quad

\dfrac{\Env \thesis M : \rec\TV.\PT}
  {\Env \thesis \unfold(M) : \PT[\rec\TV.\PT/\TV]}

\\\\

\dfrac{}
  {\Env \thesis a : \BT}
~\text{($\BT \in \Base$ and $a \in \BT$)}

\qquad\quad

\dfrac{ \Env \thesis M : \BT
  \qquad\text{for each $a \in \BT$,\quad} \Env \thesis N_a : \PT}
  {\Env \thesis \cse\ M\ \copair{a \mapsto N_a \mid a \in \BT} : \PT}

\end{array}
\]
\caption{Typing Rules of the Pure Calculus.%
\label{fig:app:puretyping}}
\end{figure}

\noindent
Figure~\ref{fig:app:puretyping} gathers all rules of the pure calculus.

\subsection{\nameref{sec:wf}}
\label{sec:app:wf}

A give a proof of Lemma~\ref{lem:wf} 
in order to illustrate Proposition~\ref{prop:wf} in a simple case.

\LemWf*

\begin{proof}
\begin{full}
Recall that the rule $\ax{WF}$ is
\[
\dfrac{\text{for all $i\in I$, $\psi_i \in \Lang_\land(\PTbis)$}
  \qquad
  \varphi \in \Lang_\Open(\PT)}
  {
  \left( \bigwedge_{i \in I} \psi_i \right)
  \realto
  \varphi
  \,\thesis\,
  \bigvee_{\text{$J \sle I$, $J$ finite}}
  \left(
  \left( \bigwedge_{j \in J} \psi_j \right)
  \realto
  \varphi
  \right)
  }
\]
\end{full}

We shall apply Proposition~\ref{prop:wf} to the Scott domain $\SI\PTbis$.
Let $f \in \SI{\PTbis \arrow \PT}$
such that $\SI{\bigwedge_{i \in I} \psi_i} \sle f^{-1}(\SI\varphi)$.
Note that $f^{-1}(\SI\varphi)$ is Scott-open since $f$ is Scott-continuous
while $\SI\varphi$ is Scott-open by Lemma~\ref{lem:top:char}.

Let $\Filt$ be the set of all
$\SI{\bigwedge_{j \in J} \psi_j}$, where $J$ ranges over all finite subsets of $I$.
We check the assumptions of Proposition~\ref{prop:wf}.
\begin{itemize}
\item
First, $\Filt$ is filtering since it is non-empty
(as $\emptyset$ is a finite subset of $I$)
and since given 
$\SI{\bigwedge_{j \in J} \psi_j}$
and
$\SI{\bigwedge_{k \in K} \psi_k}$
in $\Filt$,
we have
$\SI{\bigwedge_{\ell \in J \cup K} \psi_\ell} \in \Filt$
with
\(
  \SI{\bigwedge_{\ell \in J \cup K} \psi_\ell}
  \sle
  \SI{\bigwedge_{j \in J} \psi_j}
  ,
  \SI{\bigwedge_{k \in K} \psi_k}
\).

\item
Second, it follows from 
Lemmas~\ref{lem:top:char:fin} and \ref{lem:top:char}
that $\Filt$ consists of compacts saturated subsets of $\SI\PTbis$.

\item
Third, we have
\[
\begin{array}{l l l}
  \I{\bigwedge_{i \in I} \psi_i}
& =
& \bigcap_{i \in I} \I{\psi_i}
\\

& =
& \bigcap_{J \sle_\fin I} \bigcap_{j \in J} \I{\psi_j}
\\

& =
& \bigcap_{J \sle_\fin I} \I{\bigwedge_{j \in J} \psi_j}
\\

& =
& \bigcap \Filt
\end{array}
\]
\end{itemize}

\noindent
Now we are done since by Proposition~\ref{prop:wf} there is some
$\SI{\bigwedge_{j \in J} \psi_j} \in \Filt$
such that 
$\SI{\bigwedge_{j \in J} \psi_j} \sle f^{-1}(\I\varphi)$.
\qed
\end{proof}

\subsection{\nameref{sec:main}}
\label{sec:app:main}

We prove Lemma~\ref{lem:compl:nf} and our Main Theorem~\ref{thm:main}.

\LemComplNf*

\begin{proof}
The general strategy is to reduce to the finite case
(Proposition~\ref{prop:compl:fin:ded}),
by using Proposition~\ref{prop:wf}, but without using the rule $\ax{WF}$.

Let $\varphi,\psi \in \Norm(\PT)$ such that $\I\psi \sle \I\varphi$.
Since $\varphi \in \Norm(\PT)$, we have
$\varphi = \bigwedge_{k \in K} \varphi_k$
with $\varphi_k \in \Lang_\Open(\PT)$.
Hence, for each $k \in K$ we have
$\I\psi \sle \I{\varphi_k}$.
Thanks to the right-rule for $\bigwedge$ in Figure~\ref{fig:log:ded},
we can therefore reduce to the case of $\varphi \in \Lang_{\Open}(\PT)$.

We now assume $\varphi \in \Lang_\Open(\PT)$,
with $\varphi = \bigvee_{k \in K} \varphi_k$
and $\varphi_k \in \Lang_\land(\PT)$.
Since $\psi \in \Norm(\PT)$, using Example~\ref{ex:log:distr}
we can actually put $\psi$ in $\bigvee\bigwedge$-form:
we have
$\psi \thesisiff \bigvee_{i \in I} \bigwedge_{j \in J_i} \psi_{i,j}$
with $\psi_{i,j} \in \Lang_\land(\PT)$.
If $\I\psi \sle \I\varphi$,
then for all $i \in I$ we have
$\I{\psi_i} \sle \I\varphi$.
Thanks to the left-rule for $\bigvee$ in Figure~\ref{fig:log:ded},
we can therefore reduce to the case where $\psi$ is
of the form $\bigwedge_{i \in I} \psi_i$ with $\psi_i \in \Lang_\land(\PT)$.

Assume $\SI{\bigwedge_{i \in I} \psi_i} \sle \I\varphi$
with $\psi_i \in \Lang_\land(\PT)$,
and with $\varphi \in \Lang_\Open(\PT)$ as above.
We use Proposition~\ref{prop:wf}.
Similarly as in the proof of Lemma~\ref{lem:wf},
let $\Filt$ be the set of all
$\SI{\bigwedge_{j \in J} \psi_j}$, where $J$ ranges over all finite subsets of $I$.
The assumptions of Proposition~\ref{prop:wf} are checked similarly as in the proof of
Lemma~\ref{lem:wf}.
Again similarly as in the proof of Lemma~\ref{lem:wf},
there is some finite $J \sle I$ such that
$\SI{\bigwedge_{j \in J} \psi_j} \sle \I\varphi$.

Since $\bigwedge_{i \in I}\psi_i \thesis \bigwedge_{j \in J}\psi_j$,
we are done if we show that 
$\bigwedge_{j \in J}\psi_i \thesis \varphi$
is derivable.
Note that $\bigwedge_{j \in J}\psi_j \in \Lang_\land(\PT)$
since $J$ is finite.

Assume that $\SI{\bigwedge_{j \in J}\psi_j} = \emptyset$.
By Proposition~\ref{prop:compl:fin:ded}
we have $\bigwedge_{j \in J}\psi_j \thesis \False$,
from which we get
$\bigwedge_{j \in J}\psi_i \thesis \varphi$.

Otherwise, by Lemma~\ref{lem:top:char:fin}
there is some finite $d \in \I\PT$ such that
$\up d = \SI{\bigwedge_{j \in J}\psi_j}$.
Hence $d \in \I\varphi$,
and there is some $k \in K$ such that $d \in \I{\varphi_k}$.
But this implies
$\SI{\bigwedge_{j \in J}\psi_j} \sle \I{\varphi_k}$,
and
$\bigwedge_{j \in J}\psi_j \thesis \varphi_k$
is derivable by
Proposition~\ref{prop:compl:fin:ded}.
We then obtain
$\bigwedge_{j \in J}\psi_i \thesis \varphi$
using the right-rule for $\bigvee$.
\qed
\end{proof}

We now turn to our main result (Theorem~\ref{thm:main}).
Given a typing context $\Env = x_1:\RTbis_1,\dots,x_n:\RTbis_n$,
we write $\I\Env$ for $\I{\RTbis_1} \times \dots \times \I{\RTbis_n}$.

\ThmMain*

\begin{proof}
Thanks to Proposition~\ref{prop:main:eta},
we only have to consider the case of a normal judgment $\Env \thesis M : \RT$
in which the type $\RT$ is normal.
The general idea of the proof is somehow similar to that of
Lemma~\ref{lem:compl:nf}:
we reduce to the finite case (Theorem~\ref{thm:compl:fin}),
by using Proposition~\ref{prop:wf}, but without using the rule $\ax{WF}$.

Since $\RT$ is normal, it is of the form $\reft{\PT \mid \varphi}$
with $\varphi \in \Norm(\PT)$.
Similarly as in Lemma~\ref{lem:compl:nf}
(but using the right-rule for $\bigwedge$ in Figure~\ref{fig:reft:reftyping}),
we can reduce to the case of $\varphi \in \Lang_\Open(\PT)$,
with $\varphi$ of the form $\bigvee_{k \in K} \varphi_k$,
where $\varphi_k \in \Lang_\land(\PT)$.

Assume $\Env = \Env', x : \RTbis$.
Since $\RTbis$ is normal, it is of the form
$\reft{\PTbis \mid \psi}$, with $\psi \in \Norm(\PTbis)$.
Again similarly as in Lemma~\ref{lem:compl:nf},
using Example~\ref{ex:log:distr}
we can actually put $\psi$ in $\bigvee\bigwedge$-form:
we have
$\psi \thesisiff \bigvee_{i \in I} \bigwedge_{j \in J_i} \psi_{i,j}$
with $\psi_{i,j} \in \Lang_\land(\PTbis)$.
For each $i \in I$,
the judgment
$\Env', x : \reft{\PTbis \mid \psi_i} \thesis M : \reft{\PT \mid \varphi}$
is sound (since so is $\Env \thesis M : \RT$).
Hence, using the left-rule for $\bigvee$
in Figure~\ref{fig:reft:reftyping},
we can reduce to the case where $\psi$ is a $\bigwedge$
of formulae in $\Lang_\land(\PTbis)$.

Repeating the above for each declaration $(x:\RTbis) \in \Env$,
we can assume that $\Env$ is of the form
$x_1:\RTbis_1,\dots,x_n:\RTbis_n$,
where $\RTbis_i = \reft{\PTbis_i \mid \psi_i}$
with $\psi_i = \bigwedge_{j \in J_i} \psi_{i,j}$
and $\psi_{i,j} \in \Lang_\land(\PTbis_i)$.

We shall now apply Proposition~\ref{prop:wf}
to the Scott domain $\I{\UPT\Env}$.
Note that $\I M$ is a Scott-continuous function $\I{\UPT\Env} \to \I\PT$,
so that $\SP \deq \I M^{-1}(\I\varphi)$
is open in $\I{\UPT\Env}$ by Lemma~\ref{lem:top:char}.
Let $\Filt$ consist of all the
\[
\begin{array}{l l l}
  \I{\bigwedge_{\ell \in L_1} \psi_{1,\ell}}
& \times \dots \times
& \I{\bigwedge_{\ell \in L_n} \psi_{n,\ell}}
\end{array}
\]

\noindent
where $L_1,\dots,L_n$ range over all finite subsets
of $J_1,\dots,J_n$, respectively.
It is easy to see
that $\Filt$ and $\SP$ meet the assumptions of Proposition~\ref{prop:wf},
namely that $\Filt$ is a filtering family of compact saturated
subsets of $\I{\UPT\Env}$ such that $\bigcap \Filt \sle \SP$.
\begin{full}
\begin{description} 
\item[$\Filt$ is filtering.]
Indeed, $\Filt$ is non-empty.
Moreover, given
\[
\begin{array}{l !{\quad\text{and}\quad} l}
  \I{\bigwedge_{\ell \in L_1} \psi_{1,\ell}}
  \times \dots \times
  \I{\bigwedge_{\ell \in L_n} \psi_{n,\ell}}

& \I{\bigwedge_{\ell \in L'_1} \psi_{1,\ell}}
  \times \dots \times
  \I{\bigwedge_{\ell \in L'_n} \psi_{n,\ell}}
\end{array}
\]

\noindent
in $\Filt$,
we have
\[
\begin{array}{l l l}
  \I{\bigwedge_{\ell \in L_1 \cup L'_1} \psi_{1,\ell}}
  \times \dots \times
  \I{\bigwedge_{\ell \in L_n \cup L'_n} \psi_{n,\ell}}
& \in
& \Filt
\end{array}
\]

\noindent
with
\[
\begin{array}{l l l}
  \I{\bigwedge_{\ell \in L_i \cup L'_i} \psi_{i,\ell}}
& \sle
& \I{\bigwedge_{\ell \in L_i} \psi_{i,\ell}}
  \,,\, 
  \I{\bigwedge_{\ell \in L'_i} \psi_{i,\ell}}
\end{array}
\]

\noindent
for all $i = 1,\dots,n$.

\item[$\Filt$ consists of compact saturated subsets.]

First, it is clear that $\Filt$ consists of saturated sets
since each $\SI{\bigwedge_{\ell \in L_i} \psi_{i,\ell}}$
is saturated by Lemma~\ref{lem:top:char},
while $\I{\UPT\Env}$ is equipped with the pointwise order.

Moreover, $\Filt$ consists of compact sets since
by Lemma~\ref{lem:top:char:fin}
each element of $\Filt$ is a (finite) product of
compact sets.

\item[We have $\bigcap \Filt \sle \SP$.]
Indeed, by assumption we have
$\prod_i \I{\psi_i} \sle \SP$,
while
\[
\begin{array}{l l l}
  \prod_i \I{\psi_i}
& =
& \bigcap_{j_1 \in J_1} \dots \bigcap_{j_n \in J_n}
  \prod_i \I{\psi_{i,j_i}}
\\

& =
& \bigcap_{L_1 \sle_\fin J_1} \dots \bigcap_{L_n \sle_\fin J_n}
  \prod_i \I{\bigwedge_{j \in L_i}\psi_{i,j}}
\\

& =
& \bigcap \Filt
\end{array}
\]
\end{description}
\end{full}

Hence Proposition~\ref{prop:wf} applies,
and there are 
finite
$L_1 \sle J_1,\dots,L_n \sle J_n$ s.t.\
\[
\begin{array}{l l l}
  \I{\bigwedge_{j \in L_1}\psi_{1,j}}
  \times \dots \times
  \I{\bigwedge_{j \in L_n}\psi_{n,j}}
& \sle
& \SP
\end{array}
\]

Using subtyping, we can therefore reduce to the
sound judgment
\[
\begin{array}{c}
  x_1 : \reft{\PTbis_1 ~\middle|~ \bigwedge_{j \in L_1}\psi_{1,j}}
  ,\dots,
  x_n : \reft{\PTbis_n ~\middle|~ \bigwedge_{j \in L_n}\psi_{n,j}}
  \thesis
  M : \reft{\PT \mid \varphi}
\end{array}
\]

Assume that for some $i$ we have
$\SI{\bigwedge_{j \in L_i}\psi_{i,j}} = \emptyset$.
Then Proposition~\ref{prop:compl:fin:ded}
yields
$\bigwedge_{j \in L_i}\psi_{i,j} \thesis \False$
and we can conclude using the left-rule for $\bigvee$
in Figure~\ref{fig:reft:reftyping}.

Otherwise, by Lemma~\ref{lem:top:char:fin}
for each $i$ there is some finite $e_i \in \I{\PTbis_i}$
such that
$\up e_i = \SI{\bigwedge_{j \in L_i}\psi_{i,j}}$.
Recall that $\varphi = \bigvee_{k \in K}\varphi_k$
with $\varphi_k \in \Lang_\land(\PT)$.
We have
\[
\begin{array}{*{5}{l}}
  \up(e_1,\dots,e_n)
& =
& \I{\bigwedge_{j \in L_1}\psi_{1,j}}
  \times \dots \times
  \I{\bigwedge_{j \in L_n}\psi_{n,j}}
& \sle
& \bigcup_{k \in K} \I M^{-1}(\I{\varphi_k})
\end{array}
\]

\noindent
Hence, for some $k \in K$ the judgment
\[
\begin{array}{c}
  x_1 : \reft{\PTbis_1 ~\middle|~ \bigwedge_{j \in L_1}\psi_{1,j}}
  ,\dots,
  x_n : \reft{\PTbis_n ~\middle|~ \bigwedge_{j \in L_n}\psi_{n,j}}
  \thesis
  M : \reft{\PT \mid \varphi_k}
\end{array}
\]

\noindent
is sound.
We can now conclude by Theorem~\ref{thm:compl:fin} and subtyping.
\qed
\end{proof}

%
\bibliographystyle{splncs04}
\bibliography{bibliographie}

\providecommand{\noopsort}[1]{}
\begin{thebibliography}{10}
\providecommand{\url}[1]{\texttt{#1}}
\providecommand{\urlprefix}{URL }
\providecommand{\doi}[1]{https://doi.org/#1}

\bibitem{abramsky91apal}
Abramsky, S.: {Domain Theory in Logical Form}. Ann. Pure Appl. Log.
  \textbf{51}(1-2),  1--77 (1991). \doi{10.1016/0168-0072(91)90065-T}

\bibitem{aj95chapter}
Abramsky, S., Jung, A.: {Domain Theory}. In: Abramsky, S., Gabbay, D., T.S.E.,
  M. (eds.) {Handbook of Logic in Computer Science}, chap.~1. Clarendon Press,
  Oxford (1995)

\bibitem{ac98book}
Amadio, R.M., Curien, P.L.: {Domains and Lambda-Calculi}. Cambridge Tracts in
  Theoretical Computer Science, Cambridge University Press (1998)

\bibitem{bk08book}
Baier, C., Katoen, J.P.: {Principles of Model Checking}. The MIT Press (2008)

\bibitem{bms19types}
Berger, U., Matthes, R., Setzer, A.: {Martin Hofmann's Case for Non-Strictly
  Positive Data Types}. In: Dybjer, P., Esp{\'i}rito~Santo, J., Pinto, L.
  (eds.) {Proceedings of TYPES 2018}, Leibniz International Proceedings in
  Informatics (LIPIcs), vol.~130, pp. 1:1--1:22. {Schloss Dagstuhl -
  Leibniz-Zentrum fuer Informatik} (2019). \doi{10.4230/LIPIcs.TYPES.2018.1}

\bibitem{bk03ic}
Bonsangue, M.M., Kok, J.N.: Infinite intersection types. Information and
  Computation  \textbf{186}(2),  285--318 (2003).
  \doi{https://doi.org/10.1016/S0890-5401(03)00143-3}

\bibitem{bs07chapter}
Bradfield, J., Stirling, C.: {Modal Mu-Calculi}. In: Blackburn, P., {Van
  Benthem}, J., Wolter, F. (eds.) Handbook of Modal Logic, Studies in Logic and
  Practical Reasoning, vol.~3, pp. 721--756. Elsevier (2007).
  \doi{https://doi.org/10.1016/S1570-2464(07)80015-2}

\bibitem{bw18chapter}
Bradfield, J.C., Walukiewicz, I.: {The mu-calculus and Model Checking}. In:
  Clarke, E.M., Henzinger, T.A., Veith, H., Bloem, R. (eds.) Handbook of Model
  Checking, pp. 871--919. Springer (2018)

\bibitem{bf06book}
Bunge, M., Funk, J.: {Singular Coverings of Toposes}, Lecture Notes in
  Mathematics, vol.~1890. Springer Berlin Heidelberg (2006)

\bibitem{dst19book}
Dickmann, M., Schwartz, N., Tressl, M.: {Spectral Spaces}. New Mathematical
  Monographs, Cambridge University Press, Cambirdge (2019).
  \doi{10.1017/9781316543870}

\bibitem{gg24book}
Gehrke, M., van Gool, S.: Topological Duality for Distributive Lattices: Theory
  and Applications. Cambridge Tracts in Theoretical Computer Science, Cambridge
  University Press, Cambridge (2024)

\bibitem{goubault13book}
Goubault-Larrecq, J.: {Non-Hausdorff Topology and Domain Theory: Selected
  Topics in Point-Set Topology}. New Mathematical Monographs, Cambridge
  University Press, Cambridge (2013). \doi{10.1017/CBO9781139524438}

\bibitem{gtw02alig}
Gr{\"a}del, E., Thomas, W., Wilke, T. (eds.): {Automata, Logics, and Infinite
  Games: A Guide to Current Research}, LNCS, vol.~2500. Springer (2002)

\bibitem{hr07chapter}
Hodkinson, I., Reynolds, M.: {Temporal Logic}. In: Blackburn, P., {Van
  Benthem}, J., Wolter, F. (eds.) Handbook of Modal Logic, Studies in Logic and
  Practical Reasoning, vol.~3, pp. 655--720. Elsevier (2007).
  \doi{https://doi.org/10.1016/S1570-2464(07)80014-0}

\bibitem{hc14lics}
Hofmann, M., Chen, W.: {Abstract interpretation from B{\"{u}}chi automata}. In:
  Henzinger, T.A., Miller, D. (eds.) Joint Meeting of the Twenty-Third {EACSL}
  Annual Conference on Computer Science Logic {(CSL)} and the Twenty-Ninth
  Annual {ACM/IEEE} Symposium on Logic in Computer Science (LICS), {CSL-LICS}
  '14, Vienna, Austria, July 14 - 18, 2014. pp. 51:1--51:10. {ACM} (2014).
  \doi{10.1145/2603088.2603127}

\bibitem{hl17lics}
Hofmann, M., Ledent, J.: A cartesian-closed category for higher-order model
  checking. In: 32nd Annual {ACM/IEEE} Symposium on Logic in Computer Science,
  {LICS} 2017, Reykjavik, Iceland, June 20-23, 2017. pp. 1--12. {IEEE} Computer
  Society (2017). \doi{10.1109/LICS.2017.8005120}

\bibitem{hjk00mscs}
Huth, M., Jung, A., Keimel, K.: Linear types and approximation. Mathematical
  Structures in Computer Science  \textbf{10}(6),  719--745 (2000)

\bibitem{jr21esop}
Jaber, G., Riba, C.: {Temporal Refinements for Guarded Recursive Types}. In:
  Yoshida, N. (ed.) {Proceedins of ESOP'21}. Lecture Notes in Computer Science,
  vol. 12648, pp. 548--578. Springer (2021).
  \doi{10.1007/978-3-030-72019-3\_20}

\bibitem{johnstone82book}
Johnstone, P.: {Stone Spaces}. Cambridge Studies in Advanced Mathematics,
  Cambridge University Press (1982)

\bibitem{ktu10popl}
Kobayashi, N., Tabuchi, N., Unno, H.: {Higher-Order Multi-Parameter Tree
  Transducers and Recursion Schemes for Program Verification}. In: POPL '10:
  Proceedings of the 37th annual ACM SIGPLAN-SIGACT symposium on Principles of
  programming languages. pp. 495--508. Association for Computing Machinery, New
  York, NY, USA (2010). \doi{10.1145/1707801.1706355}

\bibitem{kt14lics}
Koskinen, E., Terauchi, T.: {Local Temporal Reasoning}. In: Proceedings of the
  Joint Meeting of the Twenty-Third EACSL Annual Conference on Computer Science
  Logic (CSL) and the Twenty-Ninth Annual ACM/IEEE Symposium on Logic in
  Computer Science (LICS). CSL-LICS'14, Association for Computing Machinery,
  New York, NY, USA (2014). \doi{10.1145/2603088.2603138}

\bibitem{levy03book}
Levy, P.B.: {Call-By-Push-Value}. Semantics Structures in Computation,
  Springer, Dordrecht (2003). \doi{https://doi.org/10.1007/978-94-007-0954-6}

\bibitem{levy22siglog}
Levy, P.B.: {Call-by-Push-Value}. ACM SIGLOG News  \textbf{9}(2),  7--29 (may
  2022). \doi{10.1145/3537668.3537670}

\bibitem{maclane98book}
Mac~Lane, S.: {Categories for the Working Mathematician}. Springer, 2nd edn.
  (1998)

\bibitem{nukt18lics}
Nanjo, Y., Unno, H., Koskinen, E., Terauchi, T.: {A Fixpoint Logic and
  Dependent Effects for Temporal Property Verification}. In: Proceedings of the
  33rd Annual ACM/IEEE Symposium on Logic in Computer Science. pp. 759--768.
  LICS'18, Association for Computing Machinery, New York, NY, USA (2018).
  \doi{10.1145/3209108.3209204}

\bibitem{nw03concur}
Nygaard, M., Winskel, G.: {Full Abstraction for {HOPLA}}. In: Amadio, R.,
  Lugiez, D. (eds.) Proceedings of {CONCUR} 2003. Lecture Notes in Computer
  Science, vol.~2761, pp. 378--392. Springer (2003).
  \doi{10.1007/978-3-540-45187-7\_25}

\bibitem{pierce02book}
Pierce, B.C.: {Types and Programming Languages}. The MIT Press, 1st edn. (2002)

\bibitem{plotkin77tcs}
Plotkin, G.: {LCF Considered as a Programming Language}. Theoretical Computer
  Science  \textbf{5},  223--256 (1977)

\bibitem{rs24jfla}
Riba, C., Stern, S.: {Liveness Properties in Geometric Logic for
  Domain-Theretic Streams}. In: Proceedings of JFLA'24 (2024),
  \url{https://inria.hal.science/hal-04407194}, full version available on arXiv
  (\url{https://arxiv.org/abs/2310.12763})

\bibitem{sv10apal}
Santocanale, L., Venema, Y.: Completeness for flat modal fixpoint logics. Ann.
  Pure Appl. Logic  \textbf{162}(1),  55--82 (2010)

\bibitem{su23popl}
Sekiyama, T., Unno, H.: {Temporal Verification with Answer-Effect Modification:
  Dependent Temporal Type-and-Effect System with Delimited Continuations}.
  Proc. ACM Program. Lang.  \textbf{7}(POPL) (jan 2023), full version available
  on arXiv at \url{https://arxiv.org/abs/2207.10386}

\bibitem{ssv08jfp}
Skalka, C., Smith, S., Van~Horn, D.: {Types and Trace Effects of Higher Order
  Programs}. J. Funct. Program.  \textbf{18}(2),  179--249 (Mar 2008).
  \doi{10.1017/S0956796807006466}

\bibitem{streicher06book}
Streicher, T.: {Domain-Theoretic Foundations of Functional Programming}. World
  Scientific (2006). \doi{10.1142/6284},
  \url{https://www.worldscientific.com/doi/abs/10.1142/6284}

\bibitem{ust17popl}
Unno, H., Satake, Y., Terauchi, T.: Relatively complete refinement type system
  for verification of higher-order non-deterministic programs. Proc. {ACM}
  Program. Lang.  \textbf{2}({POPL}),  12:1--12:29 (2018).
  \doi{10.1145/3158100}

\bibitem{vickers89book}
Vickers, S.: {Topology via Logic}. Cambridge University Press, USA (1989)

\bibitem{vvk05concur}
V{\"{o}}lzer, H., Varacca, D., Kindler, E.: {Defining Fairness}. In: Abadi, M.,
  de~Alfaro, L. (eds.) Proceedings of {CONCUR} 2005. Lecture Notes in Computer
  Science, vol.~3653, pp. 458--472. Springer (2005). \doi{10.1007/11539452\_35}

\bibitem{winskel04llcs}
Winskel, G.: Linearity and nonlinearity in distributed computation. In:
  Ehrhard, T., Girard, J.Y., Ruet, P. (eds.) Linear Logic in Computer Science.
  London Mathematical Society Lecture Note Series, Cambridge University Press
  (2004)

\end{thebibliography}

\opt{draft,full}{\newpage}

\opt{full}{
\section{Proofs of \S\ref{sec:system} (\nameref{sec:system})}

\subsection{Proofs of \S\ref{sec:log} (\nameref{sec:log})}
\label{sec:proof:log}

\begin{figure}[t!]
\[
\begin{array}{c}

\dfrac{\text{for each $i \in I$, $\psi_i \,\thesis\, \varphi_i$}}
  {\bigwedge_{i \in I} \psi_i \,\thesis\, \bigwedge_{i \in I}\varphi_i}

\qquad\qquad

\dfrac{\text{for each $i \in I$, $\psi_i \thesis \varphi_i$}}
  {\bigvee_{i \in I} \psi_i \thesis \bigvee_{i \in I}\varphi_i}

\\\\

  \form\triangle \bigwedge_{i \in I} \varphi_i
  \,\thesisiff\,
  \bigwedge_{i \in I} \form\triangle \varphi_i

\qquad\qquad

  \bigvee_{i \in I} \form\triangle\varphi_i
  \,\thesisiff\,
  \form\triangle \bigvee_{i \in I}\varphi_i

\\\\

  \bigwedge_{i \in I} \bigvee_{j \in J_i} \varphi_{i,j}
  \,\thesisiff\,
  \bigvee_{f \in \prod_{i \in I}J_i}\bigwedge_{i \in I} \varphi_{i,f(i)}

\\\\

  \bigwedge_{f \in \prod_{i \in I}J_i}\bigvee_{i \in I} \varphi_{i,f(i)}
  \,\thesisiff\,
  \bigvee_{i \in I}\bigwedge_{j \in J_i}\varphi_{i,j}
\end{array}
\]
\caption{Some derivable rules and sequents,
where $\triangle$ is either $\pi_1$, $\pi_2$ or $\fold$.%
\label{fig:proof:log:derivable}}
\end{figure}

\noindent
In this Appendix~\ref{sec:proof:log},
we give details on Figure~\ref{fig:proof:log:derivable},
which gathers some derivable rule and sequents
(including those of Examples~\ref{ex:log:modalnf} and~\ref{ex:log:distr}).

\begin{lemma}
\label{lem:proof:log:functprop}
The following rules are derivable
\[
\begin{array}{c}

\dfrac{\text{for each $i \in I$, $\psi_i \thesis \varphi_i$}}
  {\bigwedge_{i \in I} \psi_i \thesis \bigwedge_{i \in I}\varphi_i}

\qquad\qquad

\dfrac{\text{for each $i \in I$, $\psi_i \thesis \varphi_i$}}
  {\bigvee_{i \in I} \psi_i \thesis \bigvee_{i \in I}\varphi_i}

\end{array}
\]
\end{lemma}

\begin{proof}
The premise of the first rule
yields
$\bigwedge_{i \in I} \psi_i \thesis \varphi_i$
for all $i \in I$, from which we obtain
$\bigwedge_{i \in I} \psi_i \thesis \bigwedge_{i \in I} \varphi_i$.
The second rule is handled similarly.
\end{proof}

\begin{lemma}
The following sequents are derivable,
where $\triangle$ is either $\pi_1$, $\pi_2$ or $\fold$:
\[
\begin{array}{c}

  \form\triangle \bigwedge_{i \in I} \varphi_i
  \,\thesis\,
  \bigwedge_{i \in I} \form\triangle \varphi_i

\qquad\text{and}\qquad

  \bigvee_{i \in I} \form\triangle\varphi_i
  \,\thesis\,
  \form\triangle \bigvee_{i \in I}\varphi_i

\end{array}
\]
\end{lemma}

\begin{proof}
For each $i \in I$ we have
\begin{center}
\AXC{}
\UIC{$\varphi_i \,\thesis\, \varphi_i$}
\UIC{$\bigwedge_{i \in I} \varphi_i \,\thesis\, \varphi_i$}
\UIC{$\form\triangle \bigwedge_{i \in I} \varphi_i \,\thesis\, \form\triangle \varphi_i$}
\DisplayProof
\end{center}

\noindent
from which we obtain the first sequent.
The other one is derived similarly.
\qed
\end{proof}

\begin{lemma}
The following sequents are derivable
\[
\begin{array}{r !{~}l!{~} l}
  \bigvee_{f \in \prod_{i \in I}J_i}\bigwedge_{i \in I} \varphi_{i,f(i)}
& \thesis
& \bigwedge_{i \in I} \bigvee_{j \in J_i} \varphi_{i,j}
\\

  \bigvee_{i \in I}\bigwedge_{j \in J_i}\varphi_{i,j}
& \thesis
& \bigwedge_{f \in \prod_{i \in I}J_i} \bigvee_{i \in I} \varphi_{i,f(i)}
\end{array}
\]
\end{lemma}

\begin{proof}
We only discuss the first, as the other one can be dealt-with similarly.
Let $f \in \prod_{i \in I}J_i$.
For each $i \in I$, derive
\begin{center}
\AXC{}
\UIC{$\varphi_{i,f(i)} \,\thesis\, \varphi_{i,f(i)}$}
\UIC{$\varphi_{i,f(i)} \,\thesis\, \bigvee_{j \in J_i} \varphi_{i,j}$}
\DisplayProof
\end{center}

\noindent
Hence Lemma~\ref{lem:proof:log:functprop}
gives
\[
\begin{array}{l l l}
  \bigwedge_{i \in I} \varphi_{i,f(i)}
& \thesis
& \bigwedge_{i \in I} \bigvee_{j \in J_i} \varphi_{i,j}
\end{array}
\]

\noindent
Then we are done since this holds for each $f \in \prod_{i \in I}J_i$.
\qed
\end{proof}

\begin{lemma}
\label{lem:proof:log:distr}
The following sequent is derivable
\[
\dfrac{}
  {\bigwedge_{f \in \prod_{i \in I}J_i}\bigvee_{i \in I} \varphi_{i,f(i)}
  \thesis
  \bigvee_{i \in I}\bigwedge_{j \in J_i}\varphi_{i,j}}
\]
\end{lemma}

\begin{proof}
This sequent amounts a well-known fact on completely distributive
complete lattices,
see e.g.~\cite[Lemma VII.1.10]{johnstone82book}.
We nevertheless offer a detailed proof.
Using the distributive law $\ax{D}$, we have
\begin{equation*}
\begin{array}{l !{~}l!{~} l}
  \bigwedge_{f \in \prod_{i \in I}J_i}\bigvee_{i \in I} \varphi_{i,f(i)}
& \thesis
& \bigvee_{F\colon (\prod_{i \in I}J_i) \to I}
  \bigwedge_{f \in \prod_{i \in I}J_i}
  \varphi_{F(f),f(F(f))}
\end{array}
\end{equation*}

\noindent
Hence we are done if we show
\begin{equation*}
\begin{array}{l !{~}l!{~} l}
  \bigvee_{F\colon (\prod_{i \in I}J_i) \to I}
  \bigwedge_{f \in \prod_{i \in I}J_i}
  \varphi_{F(f),f(F(f))}
& \thesis
& \bigvee_{i \in I}\bigwedge_{j \in J_i}\varphi_{i,j}
\end{array}
\end{equation*}

\noindent
So let $F \colon \left(\prod_{i \in I}J_i\right)  \to I$
and assume toward a contradiction that 
\begin{equation*}
\begin{array}{l !{~}l!{~} l}
  \bigwedge_{f \in \prod_{i \in I}J_i}
  \varphi_{F(f),f(F(f))}
& \not\thesis
& \bigvee_{i \in I}\bigwedge_{j \in J_i}\varphi_{i,j}
\end{array}
\end{equation*}

\noindent
It follows that for each $i \in I$, there is some $j \in J_i$ such that
\begin{equation*}
\begin{array}{l !{~}l!{~} l}
  \bigwedge_{f \in \prod_{i \in I}J_i}
  \varphi_{F(f),f(F(f))}
& \not\thesis
& \varphi_{i,j}
\end{array}
\end{equation*}

\noindent
Using the Axiom of Choice, we get a function $g \in \prod_{i \in I}J_i$
such that for all $i \in I$,
\begin{equation*}
\begin{array}{l !{~}l!{~} l}
  \bigwedge_{f \in \prod_{i \in I}J_i}
  \varphi_{F(f),f(F(f))}
& \not\thesis
& \varphi_{i,g(i)}
\end{array}
\end{equation*}

\noindent
In particular,
\begin{equation*}
\begin{array}{l !{~}l!{~} l}
  \bigwedge_{f \in \prod_{i \in I}J_i}
  \varphi_{F(f),f(F(f))}
& \not\thesis
& \varphi_{F(g),g(F(g))}
\end{array}
\end{equation*}

\noindent
a contradiction.
\qed
\end{proof}

\subsection{Proofs of \S\ref{sec:reft} (\nameref{sec:reft})}
\label{sec:proof:reft}

Lemma~\ref{lem:reft}
will be useful for completeness (\S\ref{sec:compl} and \S\ref{sec:proof:compl}).

\LemReft*

\begin{proof}
The proof is by induction on $\RT$.
The base case of $\reft{\PT \mid \varphi}$ is trivial.
In the base case of $\PT$, one can take $\varphi = \True$.
In the cases of $\RT \times \RTbis$ and $\RTbis \arrow \RT$,
by induction hypotheses we get $\varphi \in \Lang(\UPT\RT)$
and $\psi \in \Lang(\UPT\RTbis)$ such that
$\RT \eqtype \reft{\UPT\RT \mid \varphi}$
and
$\RTbis \eqtype \reft{\UPT\RTbis \mid \varphi}$.
We then conclude with
\[
\begin{array}{r c l}
  \RT \times \RTbis
& \eqtype
& \reft{\UPT\RT \times \UPT\RTbis \mid \pair{\varphi,\psi}}
\\

  \RTbis \arrow \RT
& \eqtype
& \reft{\UPT\RTbis \arrow \UPT\RT \mid \psi \realto \varphi}
\end{array}
\]
\qed
\end{proof}

}
\opt{full}{
\section{Proofs of \S\ref{sec:sem} (\nameref{sec:sem})}
\label{sec:proof:sem}

\subsection{\nameref{sec:sem:pure}}
\label{sec:proof:sem:pure}

\subsubsection{Solutions of Recursive Domain Equations.}
We review the usual solution of recursive domain equations.
We refer to~\cite{ac98book,aj95chapter,streicher06book}.

\paragraph{Categories of Domains.}
In the following, $\DCPO$ is the category of those
posets with all directed suprema, and with Scott-continuous
functions as morphisms.
$\CPO$ is the full subcategory of $\DCPO$
on posets with a least element.
Note that $\Scott$ is a full subcategory of $\CPO$.

Recall that $\DCPO$, $\CPO$ and $\Scott$
have finite products
(equipped with the component-wise order).
See \cite[Theorem 3.3.3, Theorem 3.3.5 and Corollary 4.1.6]{aj95chapter}.
Hence for each $n \in \NN$, the categories
$\Scott^n$, $\CPO^n$ and $\DCPO^n$
are (not full) subcategories of $\Scott$, $\CPO$ and $\DCPO$
respectively.

\begin{lemma}
\label{lem:proof:scott:enrich}
If $n \in \NN$ then
$\DCPO^n$,
$\CPO^n$, $\Scott^n$ are enriched in $\DCPO$.
\end{lemma}

\begin{proof}
The result for $n=1$ follows from the
Cartesian-closure of $\DCPO$, $\CPO$ and $\Scott$
(\cite[Theorem 3.3.3, Theorem 3.3.5 and Corollary 4.1.6]{aj95chapter}).
In the cases of $n \neq 1$, the result follows from the fact that
in $\DCPO, \CPO, \Scott$, finite products are Cartesian products
of sets equipped with the component-wise order.
\qed
\end{proof}

Let $\cat C$ be a category enriched over $\DCPO$.
Given objects $X,Y \in \cat C$,
an \emph{embedding-projection} pair $X \to Y$
is a pair of morphisms $\ladj f: X\rightleftarrows Y:\radj f$
where $\radj f \comp \ladj f = \id_X$ and $\ladj f \comp \radj f \leq \id_Y$.
The morphism $\ladj f$ is an \emph{embedding}
(it reflects (as well as preserves) the order),
while $\radj f$ is a \emph{projection}.
Note that if $X$ (resp.\ $Y$) has a least element,
then so does $Y$ (resp.\ $X$) and $\ladj f$ (resp $\radj f$)
is strict.
It is well-known that $\ladj f$ completely determines $\radj f$
and reciprocally, see~\cite[\S 7.1]{ac98book}
(cf.\ also~\cite[\S 3.1.4]{aj95chapter} and~\cite[\S 9]{streicher06book}).
Given an embedding $e$ (resp.\ a projection $p$),
we write $\radj e$ (resp.\ $\ladj p$)
for the corresponding projection (resp.\ embedding).

We write $\cat C^\ep$ for the category with the same objects as $\cat C$,
and with embedding-projection pairs as morphisms.
Note that we have faithful functors
$\ladj{(\pl)} \colon \cat C^\ep \to \cat C$
and
$\radj{(\pl)} \colon \cat C^\ep \to \cat C^\op$
(taking $(\ladj f,\radj f)$ to $\ladj f$ and to $\radj f$,
respectively).
Given a functor $H$ of codomain $\cat C^\ep$,
we write $\radj H$ for $\radj{(\pl)} \comp H$,
and similarly for $\ladj H$.

\paragraph{The Limit-Colimit Coincidence.}
The following (crucial and) well-known fact
is \cite[Theorem 7.1.10]{ac98book}
(see also \cite[Theorem 3.3.7]{aj95chapter}).

\begin{theorem}
\label{thm:proof:scott:limcolim}
Let $K \colon \omega \to \cat C^\ep$ be a functor
where $\cat C$ is enriched over $\DCPO$.
Each limiting cone $\varpi \colon \Lim \radj K \to \radj K$
for $\radj K \colon \omega^\op \to \cat C$ consists of projections,
and the $(\ladj{(\varpi_n)},\varpi_n)_n$
form a colimiting cocone
$K \to \Colim K$ in $\cat C^\ep$.
\end{theorem}

\begin{proof}
Let
$K \colon \omega \to \cat C^\ep$
and consider a limiting cone
\begin{equation}
\label{diag:proof:scott:lim}
\begin{array}{c}
\begin{tikzcd}[column sep=2em] 
&
& \Lim \radj K
  \arrow{dll}[above]{\varpi_0}
  \arrow{dl}{\varpi_1}
  \arrow{d}{\varpi_2}
  \arrow{dr}[below]{\varpi_n}
  \arrow{drr}{\varpi_{n+1}}

\\

  \radj K(0)
& \radj K(1)
  \arrow{l}[below]{\radj{(k_0)}}
& \radj K(2)
  \arrow{l}[below]{\radj{(k_1)}}
& \radj K(n)
  \arrow[dashed]{l}
& \radj K(n+1)
  \arrow{l}[below]{\radj{(k_n)}}
& \phantom{F}
  \arrow[dashed]{l}
\end{tikzcd}
\end{array}
\end{equation}

\noindent
in $\cat C$.
The components of the colimiting cocone
\begin{equation}
\label{diag:proof:scott:colim}
\begin{array}{c}
\begin{tikzcd}[column sep=2em]
&
& \Colim K

\\

  K(0)
  \arrow{r}[below]{k_0}
  \arrow{urr}[above]{\gamma_0}
& K(1)
  \arrow{r}[below]{k_1}
  \arrow{ur}[below]{\gamma_1}
& K(2)
  \arrow[dashed]{r}
  \arrow{u}{\gamma_2}
& K(n)
  \arrow{r}[below]{k_n}
  \arrow{ul}[below]{\gamma_n}
& K(n+1)
  \arrow[dashed]{r}
  \arrow{ull}[above]{\gamma_{n+1}}
& \phantom{F}
\end{tikzcd}
\end{array}
\end{equation}

\noindent
in $\cat C^\ep$ are given by $\radj{(\gamma_n)} = \varpi_n$
for projections.

Concerning embeddings,
for each $n \in \NN$ we build a cone with vertex $K(n) = \radj K(n)$
as follows.
Given $m \in \NN$, we have a morphism $h_{n,m} \colon K(n) \to K(m)$
obtained by composing $\radj{(k_i)}$'s or $\ladj{(k_i)}$'s
according to whether $m \leq n$ or $n \leq m$.
The $h_{n,m}$'s can be made so that $h_{n,m} = \radj{(k_m)} \comp h_{n,m+1}$.
The universal property of limits in $\cat C$ then yields
a unique morphism $c_n$ from $K(n) = \radj K(n)$ to $\Lim \radj K(n)$
such that $\varpi_m \comp c_n = h_{n,m}$ for all $m \in \NN$.

We are going to show that $c_n = \ladj{(\varpi_n)}$.
Note that $\varpi_n \comp c_n$ is the identity by definition of $c_n$.
It remains to show that $c_n \comp \varpi_n \leq \id_{\Lim \radj K}$.
We first show that
$(c_n \comp \varpi_n)_n$ forms an increasing sequence
in $\cat C(\Lim \radj K, \Lim \radj K)$.
To this end, note that $\varpi_n = \radj{(k_n)} \comp \varpi_{n+1}$ 
(since $\varpi$ is a cone).
We moreover have
$c_n = c_{n+1} \comp \ladj{(k_n)}$
since
\(
  \varpi_m \comp c_{n+1} \comp \ladj{(k_n)}
  =
  h_{n+1,m} \comp \ladj{(k_n)}
  =
  h_{n,m}
\)
for all $m \in \NN$.
We compute
\[
\begin{array}{*{5}{l}}
  c_n \comp \varpi_n
& =
& c_{n+1} \comp \ladj{(k_n)} \comp \radj{(k_n)} \comp \varpi_{n+1}
& \leq
& c_{n+1} \comp \varpi_{n+1}
\end{array}
\]

Let $\ell = \bigvee_{n}(c_n \comp \varpi_n)$.
We now claim that $\ell$ is the identity.
This will yield that $c_n \comp \varpi_n \leq \id_{\Lim \radj K}$.
In order to show that $\ell = \id_{\Lim \radj K}$,
we show that $\varpi_m \comp \ell = \varpi_m$ for all $m \in \NN$,
and use the universal property of limits in $\cat C$.
We have
\[
\begin{array}{l l l}
  \varpi_m \comp \ell
& =
& \varpi_m \comp \bigvee_{n}(c_n \comp \varpi_n)
\\

& =
& \bigvee_n \left(
  \varpi_m \comp c_n \comp \varpi_n
  \right)
\\

& =
& \bigvee_n \left(
  h_{n,m} \comp \varpi_n
  \right)
\\

& =
& \bigvee_{n\geq m} \left(
  h_{n,m} \comp \varpi_n
  \right)
\end{array}
\]

\noindent
But by definition of $h_{n,m}$, we have
$h_{n,m} \comp \varpi_n = \varpi_m$ when $m \leq n$.

We can thus set $\gamma_n = (c_n, \varpi_n)$.
Moreover, $\gamma = (\gamma_n)_n$ is indeed a cocone since
$c_n = c_{n+1} \comp \ladj{(k_n)}$ (see above).

We now claim that $\gamma \colon K \to \Lim \radj K$ is colimiting.
To this end, consider a cocone $\tau \colon K \to C$.
We thus get a cone $\radj\tau \colon \radj K \to \radj C$ in $\cat C$,
and the universal property of limits yields a unique
$p \colon \radj C \to \Lim \radj K$ such that
$\varpi_n \comp p = \radj{(\tau_n)}$ for all $n \in \NN$.
We show that $p$ is a projection.
We define a morphism $e \colon \Lim \radj K \to C$
as $e = \bigvee_{n}(\ladj{(\tau_n)} \comp \varpi_n)$.
We have
\[
\begin{array}{l l l}
  e \comp p
& =
& \left( \bigvee_{n} \ladj{(\tau_n)} \comp \varpi_n \right) \comp p
\\

& =
& \bigvee_n \ladj{(\tau_n)} \comp \radj{(\tau_n)}
\\

& \leq
& \id_{C}
\end{array}
\]

\noindent
On the other hand, given $m \in \NN$ we have
\[
\begin{array}{l l l}
  \varpi_m \comp p \comp e
& =
& \bigvee_{n} \radj{(\tau_m)} \comp \ladj{(\tau_n)} \comp \varpi_n
\\

& =
& \bigvee_{n \geq m} \radj{(\tau_m)} \comp \ladj{(\tau_n)} \comp \varpi_n
\\

& =
& \bigvee_{n \geq m} h_{n,m} \comp \radj{(\tau_n)} \comp \ladj{(\tau_n)} \comp \varpi_n
\\

& =
& \bigvee_{n \geq m} h_{n,m} \comp \varpi_n
\\

& =
& \bigvee_{n \geq m} \varpi_m
\\

& =
& \varpi_m
\end{array}
\]

\noindent
so that $p \comp e = \id_{\Lim \radj K}$
by the universal property of limits in $\cat C$.

Moreover, for all $n \in \NN$ we have
\[
\begin{array}{l l l}
  e \comp c_n
& =
& \bigvee_m \ladj{(\tau_m)} \comp \varpi_m \comp c_n
\\

& =
& \bigvee_m \ladj{(\tau_m)} \comp h_{n,m}
\\

& =
& \bigvee_{m\geq n} \ladj{(\tau_m)} \comp h_{n,m}
\\

& =
& \bigvee_{m\geq n} \ladj{(\tau_n)}
\\

& =
& \ladj{(\tau_n)}
\end{array}
\]

Consider now a morphism $\ell \colon \Lim \radj K \to C$
in $\cat C^\ep$
such that 
$\varpi_n \comp \radj\ell = \radj{(\tau_n)}$
and
$\ladj\ell \comp c_n = \ladj{(\tau_n)}$
for all $n \in \NN$.
The universal property of limits in $\cat C$
yields $\radj\ell = p$, so that $\ladj\ell = e$
since $e$ is uniquely determined from $p$.
\qed
\end{proof}

\paragraph{Solutions of Domain Equations.}
We shall use Theorem~\ref{thm:proof:scott:limcolim}
in the following situation.
Consider a functor
\[
\begin{array}{*{5}{l}}
  G
& :
& \cat D^\ep \times \cat C^\ep
& \longto
& \cat C^\ep
\end{array}
\]

\noindent
where $\cat C$ and $\cat D$ are enriched over $\DCPO$.
We moreover assume that $\cat C$ has a terminal object $\one$
which is initial in $\cat C^\ep$.
We are going to define a functor
\[
\begin{array}{*{5}{l}}
  K
& \colon
& \cat D^\ep \times \omega
& \longto
& \cat C^\ep
\end{array}
\]

\noindent
Given an object $B$ of $\cat D^\ep$,
$K(B,\pl)$ is the $\omega$-chain in $\cat C^\ep$
obtained by iterating $G_B = G(B,\pl)$ from the initial object $\one$ of $\cat C^\ep$:
\begin{equation}
\label{diag:proof:scott:chain}
\begin{tikzcd} 
  \one
  \arrow{r}[above]{\one}
& G_B(\one)
  \arrow{r}[above]{G_B(\one)}
& G^2_B(\one)
  \arrow[dashed]{r}
& G^n_B(\one)
  \arrow{r}[above]{G_B^n(\one)}
& G^{n+1}_B(\one)
  \arrow[dashed]{r}
& \phantom{F}
\end{tikzcd}
\end{equation}

Given a morphism $f \colon B \to B'$ in $\cat D^\ep$,
$K(f,\pl)$ is obtained by commutativity of the following.
\begin{equation}
\label{diag:proof:scott:natdiag}
\begin{array}{c}
\begin{tikzcd} 

  \one
  \arrow{r}[above]{\one}
  \arrow{d}{\one}
& G_B(\one)
  \arrow{r}[above]{G_B(\one)}
  \arrow{d}{G_f(\one)}
& G^2_B(\one)
  \arrow{d}{G_f^2(\one)}
  \arrow[dashed]{r}
& G^n_B(\one)
  \arrow{r}[above]{G_B^n(\one)}
  \arrow{d}{G_f^n(\one)}
& G^{n+1}_B(\one)
  \arrow{d}{G_f^{n+1}(\one)}
  \arrow[dashed]{r}
& \phantom{F}

\\

  \one
  \arrow{r}[below]{\one}
& G_{B'}(\one)
  \arrow{r}[below]{G_{B'}(\one)}
& G^2_{B'}(\one)
  \arrow[dashed]{r}
& G^n_{B'}(\one)
  \arrow{r}[below]{G_{B'}^n(\one)}
& G^{n+1}_{B'}(\one)
  \arrow[dashed]{r}
& \phantom{F}
\end{tikzcd}
\end{array}
\end{equation}

Assume now that $\cat C$ has limits of $\omega^\op$-chains
of projections.
Then Theorem~\ref{thm:proof:scott:limcolim}
yields that each $K(B,\pl)$ has a colimit in $\cat C^\ep$.
Since $K$ is a functor $\cat D^\ep \times \omega \to \cat C^\ep$,
it follows from \cite[Theorem V.3.1]{maclane98book}
that these colimits assemble into a functor
\[
\begin{array}{l l r c l}
  \Fix G
& :
& \cat D^\ep
& \longto
& \cat C^\ep
\\

&
& B
& \longmapsto
& \Colim_{n \in \omega} K(B,n)
\end{array}
\]

If $G(B,\pl)$ preserves colimits of $\omega$-chains,
then the universal property of colimits gives an isomorphism
\(
  \fold^\ep
  :
  G(B,\Fix G(B))
  \rightleftarrows
  \Fix G(B)
  :
  \unfold^\ep
\)
in $\cat C^\ep$.

We are going to prove the following.

\begin{proposition}
\label{prop:proof:scott:contfunct}
If $G \colon \cat D^\ep \times \cat C^\ep \to \cat C^\ep$
preserves colimits of $\omega$-chains,
then so does $\Fix G \colon \cat D^\ep \to \cat C^\ep$.
\end{proposition}

The proof of Proposition~\ref{prop:proof:scott:contfunct}
is split into the following lemmas.
Fix a functor $G \colon \cat D^\ep \times \cat C^\ep \to \cat C^\ep$
which preserves colimits of $\omega$-chains.

\begin{lemma}
\label{lem:proof:scott:contdiag}
The diagonal functor $\Delta \colon \cat D^\ep \to \cat D^\ep \times \cat D^\ep$
preserves colimits of $\omega$-chains.
\end{lemma}

\begin{proof}
Since colimits are pointwise in functor categories
(\cite[Corollary V.3]{maclane98book}).%
\footnote{Note that \cite[Corollary V.3]{maclane98book} only gives the result for limits.
But recall that the opposite of a functor category $[\cat C,\cat D]$
is the functor category $[\cat C^\op, \cat D^\op]$.}
\end{proof}

Lemma~\ref{lem:proof:scott:contdiag}
entails in particular that each functor
$G^n_{(\pl)}(\one) \colon \cat D^\ep \to \cat C^\ep$
preserves colimits of $\omega$-chains
($G^{n+1}_{(\pl)}(\one)$ is
$G(\pl,G^{n}_{(\pl)}(\one)) \comp \Delta$).

Proposition~\ref{prop:proof:scott:contfunct}
relies on the fact that the functor
$K \colon \cat D^\ep \to \funct{\omega,\cat C^\ep}$
preserves colimits of $\omega$-chains.
This involves some notation.

Let $W \colon \omega \to \cat D^\ep$ be an $\omega$-chain,
with colimiting cocone $\gamma \colon W \to \Colim W$.
In the following, we write $w_m \colon W(m) \to W(m+1)$
for the connecting morphisms of $W$.
The cocone $K \gamma \colon K(W) \to K(\Colim W)$
has component at $m \in \NN$
the commutative diagram in \eqref{diag:proof:scott:natdiag}
where on takes $\gamma_m \colon W(m) \to \Colim W$
for $f \colon B \to B'$.

\begin{lemma}
\label{lem:proof:scott:colimiting}
The cocone $K\gamma \colon K(W) \to K(\Colim W)$
is colimiting.
\end{lemma}

\begin{proof}
First, it follows from the above that each
$G_\gamma^n(\one) \colon G_W^n(\one) \to G_{\Colim W}^n(\one)$
is colimiting.

Consider now a cocone
$\tau \colon K(W) \to H$ in $\funct{\omega,\cat C^\ep}$.
For each $m \in \NN$, we have
$\tau_m = \tau_{m+1} \comp K(w_m)$,
that is
\begin{equation}
\label{diag:proof:scott:taucocone}
\begin{array}{c}
\begin{tikzcd}[column sep=2.14em, row sep=large]

  \one
  \arrow{r}[above]{\one}
  \arrow{d}{\one}
& G_{B}(\one)
  \arrow{r}[above]{G_{B}(\one)}
  \arrow{d}{G_{w_m}(\one)}
& G^2_{B}(\one)
  \arrow{d}{G_{w_m}^2(\one)}
  \arrow[dashed]{r}
& G^n_{B}(\one)
  \arrow{r}[above]{G_{B}^n(\one)}
  \arrow{d}{G_{w_m}^n(\one)}
& G^{n+1}_{B}(\one)
  \arrow{d}{G_{w_m}^{n+1}(\one)}
  \arrow[dashed]{r}
& \phantom{F}

\\

  \one
  \arrow{r}[above]{\one}
  \arrow{d}{(\tau_{m+1})_0}
& G_{B'}(\one)
  \arrow{r}[above]{G_{B'}(\one)}
  \arrow{d}{(\tau_{m+1})_1}
& G^2_{B'}(\one)
  \arrow{d}{(\tau_{m+1})_2}
  \arrow[dashed]{r}
& G^n_{B'}(\one)
  \arrow{r}[above]{G_{B'}^n(\one)}
  \arrow{d}{(\tau_{m+1})_n}
& G^{n+1}_{B'}(\one)
  \arrow{d}{(\tau_{m+1})_{n+1}}
  \arrow[dashed]{r}
& \phantom{F}

\\

  H(0)
  \arrow{r}[below]{h(0)}
& H(1)
  \arrow{r}[below]{h(1)}
& H(2)
  \arrow[dashed]{r}
& H(n)
  \arrow{r}[below]{h(n)}
& H(n+1)
  \arrow[dashed]{r}
& \phantom{F}
\end{tikzcd}
\end{array}
\end{equation}

\noindent
where $B$ is $W(m)$, $B'$ is $W(m+1)$
and the $h(n) \colon H(n) \to H(n+1)$ are the connective morphisms of $H$.
In particular, for each $m \in \NN$ and each $n \in \NN$,
we have $(\tau_m)_n = (\tau_{m+1})_n \comp G_{w_m}^n(\one)$.
Hence, for each $n \in \NN$ we have a cocone
$((\tau_m)_n)_m \colon G_W^n(\one) \to H(n)$,
and the universal property of $G_\gamma^n(\one)$
gives a unique morphism $\ell_n \colon G_{\Colim W}^n(\one) \to H(n)$
such that $(\tau_m)_n = \ell_n \comp G_{\gamma_m}^n(\one)$ for all $m \in \NN$.

We show that the $\ell_n$'s assemble into a morphism
$\ell \colon K(\Colim W) \to H$ in $\funct{\omega,\cat C^\ep}$.
We thus have to show that the following commutes
\begin{equation*}
\begin{array}{c}
\begin{tikzcd} 

  \one
  \arrow{r}[above]{\one}
  \arrow{d}{\ell_0}
& G_{B'}(\one)
  \arrow{r}[above]{G_{B'}(\one)}
  \arrow{d}{\ell_1}
& G^2_{B'}(\one)
  \arrow{d}{\ell_2}
  \arrow[dashed]{r}
& G^n_{B'}(\one)
  \arrow{r}[above]{G_{B'}^n(\one)}
  \arrow{d}{\ell_n}
& G^{n+1}_{B'}(\one)
  \arrow{d}{\ell_{n+1}}
  \arrow[dashed]{r}
& \phantom{F}

\\

  H(0)
  \arrow{r}[below]{h(0)}
& H(1)
  \arrow{r}[below]{h(1)}
& H(2)
  \arrow[dashed]{r}
& H(n)
  \arrow{r}[below]{h(n)}
& H(n+1)
  \arrow[dashed]{r}
& \phantom{F}

\end{tikzcd}
\end{array}
\end{equation*}

\noindent
where $B'$ is $\Colim W$.
We show that $\ell_{n+1} \comp G_{B'}^n(\one) = h(n) \comp \ell_n$
for all $n \in \NN$.
For each $m \in \NN$, 
by commutativity of \eqref{diag:proof:scott:natdiag}
and \eqref{diag:proof:scott:taucocone}
we have
\[
\begin{array}{l l l}
  \ell_{n+1} \comp G_{B'}^{n}(\one) \comp G_{\gamma_m}^{n}(\one)
& =
& \ell_{n+1} \comp G_{\gamma_m}^{n+1}(\one) \comp G_B^{n}(\one)
\\

& =
& (\tau_m)_{n+1} \comp G_B^{n}(\one)
\\

& =
& h(n) \comp (\tau_m)_n
\\

& =
& h(n) \comp \ell_n \comp G_{\gamma_m}^n(\one)
\end{array}
\]

\noindent
where $B$ is $W(m)$.
Then we are done by the universal property of $G_{\gamma_m}^n(\one)$.

Consider finally a morphism $f \colon K(\Colim W) \to H$ in
$\funct{\omega,\cat C^\ep}$
such that $f \comp K(\gamma) = \tau$.
Then for all $m \in \NN$ we have
$f \comp K(\gamma_m) = \tau_m$,
and thus
$f_n \comp G_{\gamma_m}^n(\one) = (\tau_m)_n$
for all $n \in \NN$.
It follows that $f_n = \ell_n$, so that $f = \ell$.
\qed
\end{proof}

We can now prove Proposition~\ref{prop:proof:scott:contfunct}.

\begin{proof}[of Proposition~\ref{prop:proof:scott:contfunct}]
Let $W \colon \omega \to \cat D^\ep$ be an $\omega$-chain.
By Lemma~\ref{lem:proof:scott:colimiting},
and since colimits always commute over colimits,
we have
\[
\begin{array}{l l l}
  \Fix G(\Colim W)
& =
& \Colim_{n \in \omega} K(\Colim W,n)
\\

& \cong
& \Colim_{n \in \omega} \Colim_{m \in \omega} K(W(m),n)
\\

& \cong
& \Colim_{m \in \omega} \Colim_{n \in \omega} K(W(m),n)
\\

& \cong
& \Colim_{m \in \omega} \Fix G(W(m))
\end{array}
\]
\qed
\end{proof}

\paragraph{Local Continuity.}
Functors $G \colon \cat D^\ep \times \cat C^\ep \to \cat E^\ep$
will be obtained from ``mixed-variance'' functors
\[
\begin{array}{*{5}{l}}
  F
& :
& \cat D^\op \times \cat C
& \longto
& \cat E
\end{array}
\]

\noindent
where
$\cat D,\cat C, \cat E$ are enriched over $\DCPO$.

\begin{definition}
\label{def:proof:scott:loc}
We say that $F$
is \emph{locally} \emph{monotone} (resp.\ \emph{continuous})
if each hom-function
\[
\begin{array}{r c l}
  \cat D(B',B) \times \cat C(A,A')
& \longto
& \cat C(F(B,A), F(B',A'))
\\

  (g,f)
& \longmapsto
& F(g,f)
\end{array}
\]

\noindent
is monotone (resp.\ Scott-continuous).
\end{definition}

\noindent
We refer to \cite[Definition 5.2.5]{aj95chapter},
\cite[Definition 7.1.15]{ac98book}
and \cite[Definition 9.1]{streicher06book}.
The following is a straightforward adaptation of \cite[Proposition 7.1.19]{ac98book}
(see also \cite[Proposition 5.2.6]{aj95chapter}).

\begin{lemma}
\label{lem:proof:scott:lift}
Let
$F \colon \cat D^\op \times \cat C \to \cat E$
be locally monotone.
Then $F$ lifts to a covariant functor
\[
\begin{array}{*{5}{l}}
  F^\ep
& :
& \cat D^\ep \times \cat C^\ep
& \longto
& \cat E^\ep
\end{array}
\]

\noindent
with $F^\ep(B,A) = F(B,A)$ on objects and
$F^\ep(g,f) = (F(\radj g,\ladj f) \,,\, F(\ladj g,\radj f))$
on morphisms.

If moreover $F$ is locally continuous, then $F^\ep$
preserves colimits of $\omega$-chains.
\end{lemma}

\subsubsection{Interpretation of Pure Types.}
A \emph{pure type expression} is a possibly open production of the
grammar of pure types (\S\ref{sec:pure}), namely
\[
\begin{array}{r @{\ \ }c@{\ \ } l}
     \PT
&    \bnf
&    \BT
\gss \PT \times \PT
\gss \PT \arrow \PT
\gss \TV
\gss \rec \TV.\PT 
\end{array}
\]

\noindent
where $\BT \in \Base$,
where $\TV$ is a type variable,
and where $\rec\TV.\PT$ binds $\TV$ in $\PT$.

Consider a pure type expression $\PT$ with free
type variables $\vec \TV = \TV_1,\dots,\TV_n$.
We are going to interpret $\PT$ as a functor
\[
\begin{array}{*{5}{l}}
  \I\PT
& :
& \left( \Scott^\ep \right)^{n}
& \longto
& \Scott^\ep
\end{array}
\]

\noindent
which preserves colimits of $\omega$-chains.

\paragraph{Preliminaries.}
Recall that the category $\Scott$ is Cartesian-closed
(products and homsets are equipped with pointwise orders),
see \cite[Corollary 4.1.6]{aj95chapter} or \cite[\S 1.4]{ac98book}.
This yields functors
$\Scott(\pl,\pl) \colon \Scott^\op \times \Scott \to \Scott$
and
$(\pl) \times (\pl) \colon \Scott \times \Scott \to \Scott$.
These functors are locally continuous
(\cite[Example 7.1.16]{ac98book}).
By combining Lemma~\ref{lem:proof:scott:enrich},
Lemma~\ref{lem:proof:scott:lift} and Lemma~\ref{lem:proof:scott:contdiag},
we obtain functors
\[
\begin{array}{*{5}{l}}
  \left(\Scott(\pl,\pl) \right)^\ep
  ,\,
  \left( (\pl) \times (\pl) \right)^\ep
& :
& \Scott^\ep \times \Scott^\ep
& \longto
& \Scott^\ep
\end{array}
\]

\noindent
which preserve colimits of $\omega$-chains.

Moreover, $\Scott$ has limits of $\omega^\op$-chains
of projections (in the embedding-projection sense),
see \cite[Theorem 3.3.7, Theorem 3.3.11 and Proposition 4.1.3]{aj95chapter}.
More precisely, the full inclusion $\Scott \emb \DCPO$
creates limits for $\omega^\op$-chains of projections.%
\footnote{The notion of creation of limits has to be understood in the usual
sense of \cite[Definition V.1]{maclane98book}.}
In particular, $\Scott$
is closed in $\DCPO$ under limits of $\omega^\op$-chains of projections.
Note that the category $\DCPO$ has all limits,
and that they are created by the forgetful functor to the category
of posets (and monotone functions),
see \cite[Theorem 3.3.1]{aj95chapter}.
It follows that given $K \colon \omega \to \cat \Scott^\ep$,
the limit of $\radj K \colon \omega^\op \to \Scott$
is
\[
  \left\{
  (x_i)_i \in \prod_{i \in \NN} K(i)
  \mathrel{\Big|}
  \radj K(i \leq j)(x_j) = x_i
  \right\}
\]

\noindent
equipped with the pointwise order.
Moreover, the limiting
cone $\Lim \radj K \to \radj K$ 
consists in set-theoretic projections.%
\footnote{These are also projections in the embedding-projection sense
by Theorem~\ref{thm:proof:scott:limcolim}.}
In view of Theorem~\ref{thm:proof:scott:limcolim},
we also get that $\Scott$ is closed in $\DCPO$
under colimits of $\omega$-chains of embeddings.

The terminal object $\one$ of $\Scott$ is initial in $\Scott^\ep$
(\cite[Proposition 7.1.9]{ac98book}).

\paragraph{Definition of the Interpretation.}
Let $\PT$ be a (pure) type expression with free
type variables $\vec \TV = \TV_1,\dots,\TV_n$.
The interpretation $\I\PT \colon \left(\Scott^\ep\right)^n \to \Scott^\ep$
is defined by induction on $\PT$.
\begin{itemize}
\item
In the case of $\PT = \TV_i$,
we let $\I\PT$ take $\vec X = X_1,\dots,X_n$ to $X_i$.

\item
In the case of $\BT \in \Base$,
we let $\I\PT(\vec X)$ be the flat domain $\BT_\bot$,
where $\BT_\bot$ is $\BT + \{\bot\}$ with $\BT$ discrete.

\item
In the cases of $\PT \times \PTbis$
and $\PTbis \arrow \PT$,
the induction hypotheses give us
functors
\[
\begin{array}{*{5}{l}}
  \I\PT, \I\PTbis
& :
& \left( \Scott^\ep \right)^{n}
& \longto
& \Scott^\ep
\end{array}
\]

\noindent
which preserve colimits of $\omega$-chains.
We can thus set
\[
\begin{array}{r c l}
  \I{\PTbis \arrow \PT}(\vec X)
& =
& \left( \Scott \left( \I{\PTbis}(\vec X) ,\, \I{\PT}(\vec X) \right) \right)^\ep
\\

  \I{\PT \times \PTbis}(\vec X)
& =
& \left(
  \I\PT(\vec X) \times \I\PTbis(\vec X)
  \right)^\ep
\end{array}
\]

\item
In the case of $\rec \TV.\PT$,
the induction hypothesis gives a functor
\[
\begin{array}{*{5}{l}}
  \I\PT
& :
& \left( \Scott^\ep \right)^{n} \times \Scott^\ep
& \longto
& \Scott^\ep
\end{array}
\]

\noindent
which preserves colimits of $\omega$-chains.
Theorem~\ref{thm:proof:scott:limcolim}
gives a functor
\[
\begin{array}{l l r c l}
  \I{\rec \TV.\PT}
& :
& \left( \Scott^\ep \right)^{n} 
& \longto
& \Scott^\ep
\\

&
& \vec X
& \longmapsto
& \Fix (\I{\PT}(\vec X))
\end{array}
\]

\noindent
This functor preserves colimits of $\omega$-chains by
Proposition~\ref{prop:proof:scott:contfunct}.
Moreover, since $\I\PT$ preserves
colimits of $\omega$-chains,
we obtain canonical isomorphisms
\(
  \I\fold
  :
  \I{\PT[\rec \TV.\PT/\TV]}(\vec X)
  \rightleftarrows
  \I{\rec \TV.\PT}(\vec X)
  :
  \I\unfold
\)
by taking
$\I\fold = \ladj{(\fold^\ep)}$
and
$\I\unfold = \ladj{(\unfold^\ep)}$.
\end{itemize}

\begin{figure}[t!]
\[
\begin{array}{c}

\dfrac{}
  {\bot \in \Fin(\I\PT(\vec X))}  

\qquad\qquad

\dfrac{\text{$\BT \in \Base$ and $a \in \BT$}}
  {a \in \Fin(\I\BT(\vec X))}

\qquad\qquad

\dfrac{\text{$d$ finite in $X_i$}}
  {d \in \Fin(\I{\TV_i}(\vec X))}

\\\\

\dfrac{d \in \Fin(\I\PT(\vec X))
  \qquad
  e \in \Fin(\I\PTbis(\vec X))}
  {(d,e) \in \Fin(\I{\PT \times \PTbis}(\vec X))}

\qquad\qquad

\dfrac{d \in \Fin(\I{\PT[\rec\TV.\PT/\TV]}(\vec X))}
  {\I\fold(d) \in \Fin(\I{\rec\TV.\PT}(\vec X))}

\\\\

\dfrac{\begin{array}{l}
  \text{for all $i \in I$,~
  $d_i \in \Fin(\I{\PT})(\vec X)$
  ~and~
  $e_i \in \Fin(\I{\PTbis})(\vec X)$ \@;}
  \\
  \text{for all $J \sle I$,~
  $\bigvee_{j \in J} d_j$ defined in $\I{\PT}(\vec X)$
  ~$\imp$~
  $\bigvee_{j \in J} e_j$ defined in $\I{\PTbis}(\vec X)$}
  \end{array}}
  {\bigvee_{i \in I}(d_i \step e_i) \in \Fin(\I{\PT \arrow \PTbis}(\vec X))}
~(\text{$I$ finite})

\end{array}
\]
\caption{Inductive description of the finite elements of $\I\PT(\vec X)$.%
\label{fig:proof:sem:finelt}}
\end{figure}

\paragraph{Description of the Finite Elements.}
For each (pure) type expression $\PT$
with free variables $\vec\PT = \PT_1,\dots,\PT_n$,
we define a set $\Fin(\I{\PT}(\vec X))$.
The definition is by induction on derivations with
the rules in Figure~\ref{fig:proof:sem:finelt}.
The set $\Fin(\I{\PT}(\vec X))$
describes the finite elements of $\I{\PT}(\vec X)$.
This relies on the following.

Given $\BT \in \Base$,
the finite elements of the flat domain $\I\BT$
are exactly the elements of $\BT$.

Let $X,Y \in \Scott$.
The finite elements in the product $X \times Y$
are exactly the pairs of finite elements.
The finite elements of $\Scott(X,Y)$ are exactly the finite sups of step functions.
Given finite $d \in X$ and $e \in Y$,
the \emph{step function} $(d \step e) \colon X \to Y$
is defined as $(d \step e)(x) = e$ if $x \geq d$ and
$(d \step e)(x)= \bot$ otherwise.
Recall that the sup $\bigvee_{i \in I}(d_i \step e_i)$ of 
a finite family of step functions exists
if, and only if,
for every $J \sle I$, the set $\{e_j \mid j \in J\}$ has an upper bound
whenever so does $\{d_j \mid j \in J\}$.
See \cite[Theorem 1.4.12]{ac98book}.

Concerning recursive types, let 
\[
\begin{array}{*{5}{l}}
  G
& :
& \Scott^\ep
& \longto
& \Scott^\ep
\end{array}
\]

\noindent
be a functor which preserves colimits of $\omega$-chains.
Recall that $\Fix G$ is the colimit in \eqref{diag:proof:scott:colim}
where $K \colon \omega \to \Scott^\ep$
takes $n$ to $G^n(\one)$
(similarly as in \eqref{diag:proof:scott:chain}).
We have seen that $\Scott$ is closed in $\DCPO$
under colimits of $\omega$-chains of embeddings.
Hence it follows from \cite[Theorem 3.3.11]{aj95chapter}
that the finite elements of $\Fix G$
are the images of the finite elements of the $G^n(\one)$'s
under the components of the colimiting cocone
$\gamma \colon K \to \Fix G$.

We thus have the following.

\begin{proposition}
\label{prop:proof:scott:fin}
$\Fin(\I\PT(\vec X))$ is the set of finite elements of $\I\PT(\vec X)$.
\end{proposition}

\paragraph{Example.}
We now provide some details on Example~\ref{ex:scott:stream-tree}
on $\I{\Stream\PTbis}$ and $\I{\Tree\PTbis}$,
where $\PTbis$ is a pure type.
We handle streams and binary trees
uniformly by considering the covariant functor
\[
\begin{array}{l l r c l}
  F
& :
& \Scott
& \longto
& \Scott
\\

&
& X
& \longmapsto
& \I\PTbis \times X^\Dir
\end{array}
\]

\noindent
where $\Dir$ is a finite set.
In view of Theorem~\ref{thm:proof:scott:limcolim},
$\Fix F$ is the limit of the $\omega^\op$-chain
\[
\begin{tikzcd} 
  \one
& F(\one)
  \arrow{l}[above]{\one}
& F^2(\one)
  \arrow{l}[above]{F(\one)}
& F^n(\one)
  \arrow[dashed]{l}
& F^{n+1}(\one)
  \arrow{l}[above]{F^n(\one)}
& \phantom{F}
  \arrow[dashed]{l}
\end{tikzcd}
\]

\noindent
where $\one$ is the terminal Scott domain $\{\bot\}$.
Hence, $\Fix F$ is
\[
\begin{array}{c}
  \left\{
  x \in \prod_{n \in \NN} F^n(\one)
  \mathrel{\Big|}
  x(n) = F^n(\one)(x(n+1))
  \right\}
\end{array}
\]

\noindent
equipped with the pointwise order.
We show that $\Fix F$ is isomorphic to $\I\PTbis^{\Dir^*}$.
Define for each $n \in \NN$ an
isomorphism $\iota_n \colon \I\PTbis^{\Dir^n} \to F^n(\one)$
as $\iota_0 = \one \colon \one \to \one$
and
\[
\begin{array}{l l r c l}
  \iota_{n+1}
& :
& \I\PTbis^{\Dir^{n+1}}
& \longto
& F^{n+1}(\one) = \I\PTbis \times \left( F^n(\one) \right)^\Dir
\\

&
& T
& \longmapsto
& \left( T(\es), (\iota_n(u \mapsto T(d \cdot u)))_{d \in \Dir}  \right)
\end{array}
\]

\noindent
and to observe that the following commutes
\[
\begin{tikzcd}
  \I\PTbis^{\Dir^{n+1}}
  \arrow{d}[left]{T \mapsto T\restr \Dir^n}
  \arrow{r}{\iota_{n+1}}
& F^{n+1}(\one)
  \arrow{d}{F^n(\one)}
\\
  \I\PTbis^{\Dir^{n}}
  \arrow{r}[below]{\iota_{n}}
& F^n(\one)
\end{tikzcd}
\]

The characterization of the finite elements then follows from
Proposition~\ref{prop:proof:scott:fin}.

\subsection{\nameref{sec:sem:log}}
\label{sec:proof:sem:log}

First note that if $\triangle$ is either $\pi_1$, $\pi_2$ or $\fold$,
then
since $\I{\form\triangle}$ acts by inverse image
(of resp.\ $\pi_1$, $\pi_2$ and $\I\unfold$), we directly
have that $\I{\form\triangle}$
is monotone (w.r.t.\ inclusion) and preserves all unions and all intersections.

We now consider Lemma~\ref{lem:top:char:fin}.

\LemTopCharFin*

The proof of Lemma~\ref{lem:top:char:fin} is split into the next
two lemmas.

\begin{lemma}
\label{lem:proof:top:fin:compact-open}
Given $\varphi \in \Lang_\land(\PT)$, if $\I\varphi \neq \emptyset$ then
$\I\varphi = \up d$ for some finite $d \in \I\PT$.
\end{lemma}

\begin{proof}
The proof is by induction on $\varphi \in \Lang_\land(\PT)$.
We rely on the description of finite elements given
by Proposition~\ref{prop:proof:scott:fin}
(see Figure~\ref{fig:proof:sem:finelt}).
\begin{description}
\item[Case of $\True$.]
In this case, we have $\I\varphi = \up \bot$.

\item[Case of $\varphi \land \psi$.]
First, note that $\I\varphi, \I\psi$ are non-empty since so is their intersection.
By induction hypothesis, there are finite $d,e \in \I\PT$
such that $\I\varphi = \up d$ and $\I\psi = \up e$.
Since $\up d \cap \up e$ is non-empty, and since $\I\PT$ is a Scott domain,
we get that $d \vee e$ is defined, finite, and such that
$\up(d\vee e) = \up d \cap \up e$.
Hence $\I{\varphi \land \psi} = \up(d \vee e)$.

\item[Case of $\form a$ (with $a \in \BT$ for $\BT \in \Base$).]
Since $\I{\form a} = \up a$.

\item[Case of $\form{\triangle}\varphi$
with $\triangle$ either $\pi_1$, $\pi_2$ or $\fold$.]
Note that $\I\varphi$ is non-empty since so is
$\I{\form\triangle\varphi} = \I{\form\triangle}(\I\varphi)$.
Hence by induction hypothesis, there is some finite $d$
such that $\I\varphi = \up d$.

Consider first the case of $\triangle = \fold$.
Then, since $\I\unfold$ is an isomorphism with inverse $\I\fold$
we have
\[
\begin{array}{l l l}
  \I{\form\triangle\varphi}
& =
& \I{\form\triangle}(\up d)
\\

& =
& \left\{
  x \in \I{\rec\TV.\PT} \mid \I\unfold(x) \geq d
  \right\}
\\

& =
& \up \I\fold(d)
\end{array}
\]

\noindent
The result then follows from
Proposition~\ref{prop:proof:scott:fin}.

Consider now the case of $\triangle = \pi_i$,
say $\triangle = \pi_1$ (the other case is symmetric).
Since the order in $\I{\PT_1 \times \PT_2}$ is pointwise,
we have
\[
\begin{array}{l l l}
  \I{\form{\pi_1}\varphi}
& =
& \I{\form{\pi_1}}(\up d)

\\
& =
& \left\{
  x \in \I{\PT_1 \times \PT_2} \mid \pi_1(x) \geq d
  \right\}
\\

& =
& \up(d,\bot)
\end{array}
\]

\noindent
and the result again follows from
Proposition~\ref{prop:proof:scott:fin}.

\item[Case of $\psi \realto \varphi$.]
First, if $\I\psi = \emptyset$,
then $\I{\psi \realto \varphi} = \up \bot$.

Assume now that $\I\psi \neq \emptyset$.
In this case, we must also have $\I\varphi \neq \emptyset$.
Hence by induction hypothesis there are $d,e$ finite
such that
$\up e = \I\psi$ and $\up d = \I\varphi$.
Then we are done since
\[
\begin{array}{l l l}
  \I{\psi \realto \varphi}
& =
& \left\{ f \mid \forall x \in \I\psi,~ f(x) \in \I\varphi\right\}
\\

& =
& \left\{ f \mid \forall x \geq e,~ f(x) \geq d\right\}
\\

& =
& \up \left(e \step d \right)
\end{array}
\]
\qed
\end{description}
\end{proof}

\begin{lemma}
\label{lem:proof:top:compact-open:fin}
If $d \in \I\PT$ is finite,
then there is
$\varphi \in \Lang_\land(\PT)$
such that $\up d = \I\varphi$.
\end{lemma}

\begin{proof}
We rely on Proposition~\ref{prop:proof:scott:fin}
and on the inductive definition of $\Fin(\I\PT)$ in Figure~\ref{fig:proof:sem:finelt}.
We reason by cases on the rules 
in Figure~\ref{fig:proof:sem:finelt}.
\begin{description}
\item[Case of]
\[
\dfrac{}
  {\bot \in \Fin(\I\PT)}  
\]

\noindent
Since $\up \bot = \I{\True}$.

\item[Case of]
\[
\dfrac{\text{$\BT \in \Base$ and $a \in \BT$}}
  {a \in \Fin(\I\BT)}
\]

\noindent
Since $\up a = \I{\form a}$.

\item[Case of]
\[
\dfrac{d \in \Fin(\I\PT)
  \qquad
  e \in \Fin(\I\PTbis)}
  {(d,e) \in \Fin(\I{\PT \times \PTbis})}
\]

\noindent
By induction hypothesis, we
have $\varphi \in \Lang_\land(\PT)$
and $\psi \in \Lang_\land(\PTbis)$
such that
$\I\varphi = \up d$
and
$\I\psi = \up e$.
Since the order in $\I{\PT \times \PTbis}$
is pointwise, we get
\[
\begin{array}{l l l}
  \up(d,e)
& =
& \up d \times \up e
\\

& =
& (\up d \times \I\PTbis)
  \cap
  (\I\PT \times \up e)
\\

& =
& \I{\form{\pi_1}\varphi \land \form{\pi_2}\psi}
\end{array}
\]

\item[Case of]
\[
\dfrac{d \in \Fin(\I{\PT[\rec\TV.\PT/\TV]})}
  {\I\fold(d) \in \Fin(\I{\rec\TV.\PT})}
\]

\noindent
By induction hypothesis, there is
$\varphi \in \Lang_\land(\PT[\rec\TV.\PT/\TV])$
such that $\I\varphi = \up d$.
We thus have $\up \I\fold(d) = \I{\form\fold \varphi}$.

\item[Case of]
\[
\dfrac{\begin{array}{l}
  \text{for all $i \in I$,~
  $d_i \in \Fin(\I{\PT})$
  ~and~
  $e_i \in \Fin(\I{\PTbis})$ \@;}
  \\
  \text{for all $J \sle I$,~
  $\bigvee_{j \in J} d_j$ defined in $\I{\PT}$
  ~$\imp$~
  $\bigvee_{j \in J} e_j$ defined in $\I{\PTbis}$}
  \end{array}}
  {\bigvee_{i \in I}(d_i \step e_i) \in \Fin(\I{\PT \arrow \PTbis})}
\]

\noindent
where $I$ is a finite set.

By induction hypothesis, for each $i \in I$
there are $\varphi_i \in \Lang_\land(\PTbis)$
and $\psi_i \in \Lang_\land(\PT)$
such that $\up e_i = \I{\varphi_i}$
and $\up d_i = \I{\psi_i}$.
Note that
\[
\begin{array}{l l l}
  \up(d_i \step e_i)
& =
& \left\{
  f \colon \I\PT \to \I\PTbis \mid
  \forall x \geq d_i,~
  f(x) \geq e_i
  \right\}
\\

& =
& \left\{
  f \colon \I\PT \to \I\PTbis \mid
  \forall x \in \I{\psi_i},~
  f(x) \in \I{\varphi_i}
  \right\}
\\

& =
& \I{\psi_i \realto \varphi_i}
\end{array}
\]

\noindent
The result then follows from the fact that
\[
\begin{array}{l l l}
  \up \left(
  \bigvee_{i \in I} d_i \step e_i
  \right)
& =
& \bigcap_{i \in I} \up(d_i \step e_i)
\end{array}
\]
\qed
\end{description}
\end{proof}

We now turn to Lemma~\ref{lem:top:char}.
We first recall its statement.

\LemTopChar*

\begin{proof}
Consider first the case of a set $\SP \sle \I\PT$
which is open (resp.\ saturated).
Then $\SP$ is a union (resp.\ an intersection of unions) of
sets of the form $\up d$ with $d \in \I\PT$ finite.
Then result then follows from Lemma~\ref{lem:proof:top:compact-open:fin}
using the closure of $\Lang_\Open(\PT)$ under arbitrary disjunctions
(resp.\ the closure of $\Lang(\PT)$ under arbitrary disjunctions
and conjunctions).

The converse is proven by induction on formulae.
Since opens are stable under unions and finite intersections,
the case of $\varphi \in \Lang_\Open(\PT)$
directly follows from Lemma~\ref{lem:proof:top:fin:compact-open}.
As for $\varphi \in \Lang(\PT)$,
since saturated sets are stable under all unions and intersections,
we only have to consider the cases of modalities.
\begin{description}
\item[Case of $\form a$ (with $a \in \BT$ for $\BT \in \Base$).]
Trivial, since $\I{\form a}$ is compact open.

\item[Case of $\form{\triangle}\varphi$
with $\triangle$ either $\pi_1$, $\pi_2$ or $\fold$.]
Similarly as in Lemma~\ref{lem:proof:top:fin:compact-open},
one can apply the induction hypothesis and use the fact
that $\I{\form\triangle}$ preserves all unions and all intersections.

\item[Case of $\psi \realto \varphi$.]
By induction hypothesis, $\I\varphi$ is upward-closed.
Hence so is $\I{\psi \realto \varphi}$.
\end{description}

For the last part of the statement,
let $\varphi \in \Lang(\PT)$.
Since $\I\varphi$ is saturated,
it is an intersection of unions of sets
of the form $\up d$ with $d \in \I\PT$ finite.
Lemma~\ref{lem:proof:top:compact-open:fin} yields
that such $\up d$'s are definable in $\Lang_\land(\PT)$,
whence the result.
\qed
\end{proof}

We finally discuss Proposition~\ref{prop:sem:sound:ded}.

\PropSemSoundDed*

\begin{proof}
The proof is by induction on $\psi \thesis \varphi$,
and by cases on the rules in Figure~\ref{fig:log:ded}.
The cases of the rules for (infinitary) propositional
logic directly follow from the definition of the interpretation.
So we just have to discuss modalities.

Let $\triangle$ be either $\pi_1$, $\pi_2$ or $\fold$.
Since $\I{\form\triangle}$ acts by inverse image
(of resp.\ $\pi_1$, $\pi_2$ and $\I\unfold$), we directly
have that $\I{\form\triangle}$
is monotone (w.r.t.\ inclusion) and preserves all unions and all intersections.
This handles all the rules for $\form\triangle$.

The rule $\ax{F}$ has already been discussed,
and the other rules for $\realto$ are straightforward to check.
\qed
\end{proof}

}
\opt{full}{
\section{Proofs of \S\ref{sec:compl} (\nameref{sec:compl})}
\label{sec:proof:compl}

\subsection{\nameref{sec:compl:fin}}
\label{sec:proof:compl:fin}

We begin with Proposition~\ref{prop:compl:fin:ded}.

\PropComplFinDed*

The proof of Proposition~\ref{prop:compl:fin:ded}
is split into Lemmas \ref{lem:proof:compl:fin:ded:cor},
\ref{lem:proof:compl:fin:ded:c-false}
and
\ref{lem:proof:compl:fin:ded:compl}.
Namely:
\begin{itemize}
\item
Lemma~\ref{lem:proof:compl:fin:ded:cor} is the soundness of
the system made of Figure~\ref{fig:log:ded} (\S\ref{sec:log})
and eq.~\eqref{eq:compl:cc} for formulae in $\Lang_\land$.

\item
Lemma~\ref{lem:proof:compl:fin:ded:c-false} is a form of dichotomy:
for every $\psi \in \Lang_\land$,
either $\I\psi = \emptyset$ and $\psi \thesis \False$ is derivable,
or $\I\psi \neq \emptyset$ and $\C(\psi)$ is derivable.

\item
Lemma~\ref{lem:proof:compl:fin:ded:compl} is the completeness
of the deduction relation for formulae in $\Lang_\land$.
\end{itemize}

\begin{lemma}
\label{lem:proof:compl:fin:ded:cor}
In the extension of Figure~\ref{fig:log:ded} (\S\ref{sec:log})
with eq.\ \eqref{eq:compl:cc}:
\begin{enumerate}[(1)]
\item
for all $\varphi,\psi \in \Lang_\land(\PT)$,
we have
$\I\psi \sle \I\varphi$
if
$\psi \thesis_\PT \varphi$;

\item
for all $\varphi \in \Lang_\land$,
we have
$\I\varphi \neq \emptyset$
if
$\C(\varphi)$.
\end{enumerate}
\end{lemma}

\begin{proof}
We reason by mutual induction the definition of $\thesis$ and $\C$.
Thanks to Proposition~\ref{prop:sem:sound:ded},
we do not have to consider the rules in Figure~\ref{fig:log:ded}.
We reason by cases on the last applied rule.
\begin{description}
\item[Cases of]
\[
\begin{array}{c}

\dfrac{}{\C(\True)}

\quad

\dfrac
  {\text{$\BT \in \Base$ and $a \in \BT$}}
  {\C(\form a)}

\quad

\dfrac{\C(\varphi)}
  {\C(\form\fold \varphi)}

\quad

\dfrac{\C(\psi)
  \quad
  \psi \thesis \varphi}
  {\C(\varphi)}

\end{array}
\]

\noindent
Trivial.

\item[Case of]
\[
\dfrac{\C(\varphi) 
  \quad
  \C(\psi)}
  {\C(\pair{\varphi,\psi})}
\]

\noindent
Recall that $\pair{\varphi,\psi} = \form{\pi_1}\varphi \land \form{\pi_2}\psi$.
Hence, if $\I{\pair{\varphi,\psi}} = \emptyset$ then we must have
either $\I\varphi = \emptyset$ or $\I\psi = \emptyset$,
and we conclude by induction hypothesis.

\item[Case of]
\[
\ax{C}
\dfrac{\C(\psi)}
  {(\psi \realto \False) \,\thesis\, \False}
\]

By induction hypothesis, we have $\I\psi \neq \emptyset$.
Hence $\I{\psi \realto \False} = \emptyset$.

\item[Case of]
\[
\dfrac{\begin{array}{l}
  \text{$I$ finite and $\forall i \in I$,}~
  \C(\psi_i) 
  ~\text{and}~
  \C(\varphi_i) ;
  \\
  \text{$\forall J \sle I$,}~
  \bigwedge_{j \in J} \psi_j \thesis \False
  ~\text{or}~
  \C\left( \bigwedge_{j \in J} \varphi_j \right)
  \end{array}}
  {\C\left( \bigwedge_{i \in I}(\psi_i \realto \varphi_i) \right)}
\]

Let $\PT,\PTbis$ such that
$\varphi_i \in \Lang_\land(\PT)$ and $\psi_i \in \Lang_\land(\PTbis)$
for all $i \in I$.

First, by induction hypothesis 
we have
$\I{\psi_i} \neq \emptyset$
and
$\I{\varphi_i} \neq \emptyset$
for all $i \in I$.
Hence, it follows from Lemma~\ref{lem:top:char:fin}
that for each $i \in I$,
there are finite $d_i \in \I\PT$
and $e_i \in \I\PTbis$
such that
$\up d_i = \I{\varphi_i}$
and
$\up e_i = \I{\psi_i}$.
We thus have
$\I{\psi_i \realto \varphi_i} = \up(e_i \step d_i)$
for each $i \in I$,
so that
\[
\begin{array}{l l l}
  \I{\bigwedge_{i \in I} \left( \psi_i \realto \varphi_i \right)}
& =
& \bigcap_{i \in I} \up(e_i \step d_i)
\end{array}
\]

Assume 
$\I{\bigwedge_{i \in I} \left( \psi_i \realto \varphi_i \right)} = \emptyset$.
As recalled in \S\ref{sec:proof:sem:pure}
(see also \cite[Theorem 1.4.12]{ac98book}),
there is some $J \sle I$
such that
\[
\begin{array}{l l l !{\quad\text{and}\quad} l l l}
  \bigcap_{i \in I} \up e_i
& \neq
& \emptyset

& \bigcap_{i \in I} \up d_i
& =
& \emptyset
\end{array}
\]

But the induction hypothesis
yields either $\bigcap_{i \in I} \up e_i = \emptyset$
or $\bigcap_{i \in I} \up d_i \neq \emptyset$,
a contradiction.
\qed
\end{description}
\end{proof}

\begin{lemma}
\label{lem:proof:compl:fin:ded:c-false}
For all $\psi \in \Lang_\land(\PT)$,
\begin{enumerate}[(1)]
\item
\label{item:proof:compl:fin:ded:c-false:c}
if $\I\psi \neq \emptyset$,
then $\C(\psi)$ is derivable;

\item
\label{item:proof:compl:fin:ded:c-false:ded}
if $\I\psi = \emptyset$,
then $\psi \thesis \False$ is derivable.
\end{enumerate}
\end{lemma}

\begin{proof}
Both statements are proven by a simultaneous induction on the (finite!)
size of $\psi \in \Lang_\land$.

Note that $\C(\True)$ and $\False \thesis \False$
are always derivable,
so that we may always assume $\psi \neq \True$ in
item \eqref{item:proof:compl:fin:ded:c-false:c}
and $\psi \neq \False$ in
item \eqref{item:proof:compl:fin:ded:c-false:ded}.

We now reason by cases on $\PT$.
\begin{description}
\item[Case of $\BT$ with $\BT \in \Base$.]
We begin with item \eqref{item:proof:compl:fin:ded:c-false:c}.
If $\I\psi \neq \emptyset$,
then we must have $\I\psi = \{a\}$ for some $a \in \BT$.
Hence $\psi$ is $\thesisiff$-equivalent to $\form a$,
and we get $\C(\psi)$ since $\C(\form a)$ and $\form a \thesis \psi$.

We now turn to item \eqref{item:proof:compl:fin:ded:c-false:ded}.
If $\psi \neq \False$, then it must be the case that
$\psi$ is a finite conjunction containing (at least)
$\form a$ and $\form b$ for some $a \neq b$ in $\BT$.
We can thus conclude using the rule $\form a \land \form b \thesis \False$.

\item[Case of $\rec\TV.\PT$.]
If $\psi \neq \False$,
then $\psi$ is $\thesisiff$-equivalent to a formula of the form
$\bigwedge_{j \in J}\form\fold \psi_j$ for some finite $J$,
where each $\psi_j$ is smaller than $\psi$.
Let $\psi'$ be the formula
$\bigwedge_{j \in J} \psi_j$.
Note that $\psi \thesisiff \form\fold\psi'$ by Example~\ref{ex:log:modalnf},
so that
$\I\psi = \I{\form\fold}(\I{\psi'})$.

We first consider item \eqref{item:proof:compl:fin:ded:c-false:c}.
If $\I\psi \neq \emptyset$,
then $\I{\psi'} \neq \emptyset$.
If moreover $J \neq \emptyset$ (otherwise $\psi = \True$),
then
$\psi'$ is smaller than $\psi$ and the induction hypothesis
yields $\C(\psi')$.
Hence $\C(\form\fold \psi')$
and we get $\C(\psi)$ since $\form\fold \psi' \thesis \psi$.

We now turn to item \eqref{item:proof:compl:fin:ded:c-false:ded}.
If $\I\psi = \emptyset$ then
$\I{\psi'} = \emptyset$.
In this case, $J$ must be non-empty
(since otherwise $\I{\psi'} = \I{\True}$).
So $\psi'$ is smaller than $\psi$
and the induction hypothesis yields
$\psi' \thesis_{\PT[\rec\TV.\PT/\TV]} \False$,
so that
$\form\fold\psi' \thesis_{\rec\TV.\PT} \form\fold \False$.
Then we are done since $\form\fold \False \thesis \False$
(take $I = \emptyset$ in the rule
$\form\triangle \bigvee (-) \thesis \bigvee \form\triangle(-)$).

\item[Case of $\PT_1 \times \PT_2$.]
If $\psi \neq \False$,
then $\psi$ is $\thesisiff$-equivalent to a formula of the form
\( 
  (\bigwedge_{j \in J} \form{\pi_1}\psi_j)
  \land
  (\bigwedge_{k \in K} \form{\pi_2}\psi_k)
\)
where $\psi_j,\psi_k$ are smaller than $\psi$,
and
where we can assume w.l.o.g.\ $J \cap K = \emptyset$.
Let $\psi' = \bigwedge_{j \in J} \psi_j$
and $\psi'' = \bigwedge_{k \in K} \psi_k$,
so that $\psi \thesisiff \form{\pi_1}\psi' \land \form{\pi_2}\psi''$.

We first consider item \eqref{item:proof:compl:fin:ded:c-false:c}.
If $\I\psi \neq \emptyset$,
then $\I{\psi'} \neq \emptyset$
and $\I{\psi''} \neq \emptyset$.
If $J$ (resp.\@ $K$) is non-empty,
then $\psi'$ (resp. $\psi''$) is smaller than $\psi$
and the induction hypothesis applies to yield
$\C(\psi')$ (resp.\ $\C(\psi'')$).
If $J$ (resp.\ $K$) is empty, then
$\psi' = \True$ (resp.\ $\psi'' = \True$),
so that $\C(\psi')$ (resp.\ $\C(\psi'')$).
Hence, in any case we get $\C(\psi')$ and $\C(\psi'')$,
so that $\C(\pair{\psi',\psi''})$ and we are done.

We now turn to item \eqref{item:proof:compl:fin:ded:c-false:ded}.
If  $\I\psi = \emptyset$,
then we must have either $\I{\psi'} = \emptyset$
or $\I{\psi''} = \emptyset$,
say  $\I{\psi'} = \emptyset$ (the other case is symmetric).
Reasoning similarly as above
yields $\psi' \thesis_{\PT_1} \False$ by induction hypothesis,
and we conclude using
$\form{\pi_1}\False \thesis_{\PT_1 \times \PT_2} \False$.

\item[Case of $\PTbis \arrow \PT$.]
If $\psi \neq \False$,
then $\psi$ is $\thesisiff$-equivalent to a formula of the form
$\bigwedge_{i \in I} (\psi''_i \realto \psi'_i)$,
where $\psi'_i, \psi''_i$ are smaller than $\psi$.

Assume first that for some $i \in I$,
we have $\I{\psi'_i} = \emptyset$
with $\I{\psi''_i} \neq \emptyset$.
Then $\I{\psi''_i \realto \psi'_i} = \emptyset$
and $\I\psi = \emptyset$.
Hence we must be in the case of item \eqref{item:proof:compl:fin:ded:c-false:ded}.
Moreover, by induction hypothesis we have $\C(\psi''_i)$
and $\psi' \thesis \False$,
so that we can derive $\psi \thesis \False$ using the rule $\ax{C}$.

Otherwise, we have $\I{\psi'_i} \neq \emptyset$
for all $i \in I$ such that $\I{\psi''_i} \neq \emptyset$.

Given $i \in I$ such that $\I{\psi''_i} = \emptyset$,
by induction hypothesis we have $\psi''_i \thesis \False$,
and since $\True \thesis \left( \False \realto \varphi \right)$
for any formula $\varphi$ (Remark~\ref{rem:log:realto}),
we get $\True \thesis \left(\psi''_i \realto \psi'_i \right)$.
Hence
$\bigwedge_{j \neq i}(\psi''_j \realto \psi'_j) \thesisiff \psi$.

We can therefore reduce the case of 
$\bigwedge_{i \in I} (\psi''_i \realto \psi'_i)$
where 
$\I{\psi'_i} \neq \emptyset$ and $\I{\psi''_i} \neq \emptyset$
for all $i \in I$.
In, particular, the induction hypothesis yields $\C(\psi'_i)$
and $\C(\psi''_i)$ for all $i \in I$.

Regarding item \eqref{item:proof:compl:fin:ded:c-false:c},
if $\I\psi \neq \emptyset$,
then for all $J \sle I$ we have either
$\SI{\bigwedge_{j \in J} \psi''_j} = \emptyset$
or
$\SI{\bigwedge_{j \in J} \psi'_j} \neq \emptyset$,
so that either
$\bigwedge_{j \in J} \psi''_j \thesis \False$
or
$\C(\bigwedge_{j \in J} \psi'_j)$
by induction hypothesis.
We can thus obtain $\C(\psi)$ by using the last rule in eq.\ \eqref{eq:compl:cc}.

Concerning item \eqref{item:proof:compl:fin:ded:c-false:ded},
if $\I\psi = \emptyset$,
then there is some $J \sle I$
such that
$\SI{\bigwedge_{j \in J} \psi''_j} \neq \emptyset$
while
$\SI{\bigwedge_{j \in J} \psi'_j} = \emptyset$.
Hence
$\C(\bigwedge_{j \in J} \psi''_j)$
and
$\bigwedge_{j \in J} \psi'_j \thesis \False$
by induction hypothesis.
We have
\[
\begin{array}{l l l}
  \psi
& \thesis
& \bigwedge_{i \in I}
  \left(
  \left( \bigwedge_{j \in J} \psi''_j \right)
  \realto
  \psi'_i
  \right)
\end{array}
\]

\noindent
and thus
\[
\begin{array}{l l l}
  \psi
& \thesis
& \left( \bigwedge_{j \in J} \psi''_j \right)
  \realto
  \bigwedge_{j \in J} \psi'_j
\end{array}
\]

\noindent
and we obtain $\psi \thesis \False$ using the rule $\ax{C}$.
\qed
\end{description}
\end{proof}

\begin{lemma}
\label{lem:proof:compl:fin:ded:compl}
For all $\varphi,\psi \in \Lang_\land(\PT)$,
if $\I\psi \sle \I\varphi$,
then $\psi \thesis \varphi$ is derivable;
\end{lemma}

\begin{proof}
The proof is by induction on sum of the (finite!) 
sizes of $\psi$ and $\varphi$.

First, note that if $\varphi = \bigwedge_{i \in I}\varphi_i$
for some finite set $I$,
then $\I\psi \sle \I\varphi$ implies
$\I\psi \sle \I{\varphi_i}$ for all $i \in I$,
and we can obtain $\psi \thesis \varphi$ from the induction hypotheses.

Also, if $\I\psi = \emptyset$, then Lemma~\ref{lem:proof:compl:fin:ded:c-false}
yields that $\psi \thesis \False$,
so that $\psi \thesis \varphi$.
This in particular applies when $\I\varphi = \emptyset$,
since we must then have $\I\psi = \emptyset$ as well.

We can thus assume that $\varphi$ is not a conjunction
and that both $\I\varphi$ and $\I\psi$ are not empty.
We now reason by cases on $\PT$.
\begin{description}
\item[Case of $\BT$ with $\BT \in \Base$.]
In this case, we must have $\I\varphi = \{a\}$ for some $a \in \BT$,
and thus $\I\psi = \{a\}$ as well.
We can then obtain $\psi \thesis \varphi$ using the rule $\form a \thesis \form a$.

\item[Case of $\rec\TV.\PT$.]
Reasoning as in Lemma~\ref{lem:proof:compl:fin:ded:c-false},
we can assume that $\psi$ is of the form
$\bigwedge_{j \in J}\form\fold \psi_j$ for some finite $J$.
Let $\psi'$ be the formula
$\bigwedge_{j \in J} \psi_j$.
Note that $\psi \thesisiff \form\fold\psi'$ by Example~\ref{ex:log:modalnf},
so that
$\I\psi = \I{\form\fold}(\I{\psi'})$.

Moreover, we must have $\varphi = \form\fold \varphi'$.
Hence $\I\psi \sle \I\varphi$
implies $\I{\psi'} \sle \I{\varphi'}$.
Note that $\varphi'$ is smaller than $\varphi$
while $\psi'$ is not greater than $\psi$.
Hence the induction hypothesis yields
$\psi' \thesis \varphi'$, and we are done.

\item[Case of $\PT_1 \times \PT_2$.]
Reasoning as in Lemma~\ref{lem:proof:compl:fin:ded:c-false},
we can assume that $\psi$ is of the form
\( 
  (\bigwedge_{j \in J} \form{\pi_1}\psi_j)
  \land
  (\bigwedge_{k \in K} \form{\pi_2}\psi_k)
\)
with
$J \cap K = \emptyset$.
Let $\psi' = \bigwedge_{j \in J} \psi_j$
and $\psi'' = \bigwedge_{k \in K} \psi_k$,
so that $\psi \thesisiff \form{\pi_1}\psi' \land \form{\pi_2}\psi''$.

Moreover, we have $\varphi = \form{\pi_i}\varphi'$
say $i = 1$ (the case $i =2$ is symmetric).
But then we must have $\I{\psi'} \sle \I{\varphi'}$
and similarly as above (again), the induction hypothesis
yields $\psi' \thesis_{\PT_1} \varphi'$.
It is then easy to conclude.

\item[Case of $\PTbis \arrow \PT$.]
This is the most important case.
The proof is an adaptation to our setting
of the proof of~\cite[Proposition 10.5.2]{ac98book}.

First, note that $\varphi$ must be of the form $\varphi'' \realto \varphi'$.
If $\I{\varphi''} = \emptyset$,
then by Lemma~\ref{lem:proof:compl:fin:ded:c-false}
we obtain $\varphi'' \thesis \False$,
so that $\True \thesis (\varphi'' \realto \varphi')$
by Remark~\ref{rem:log:realto}.
Hence $\psi \thesis \varphi$ in this case.
We can thus assume $\I{\varphi''} \neq \emptyset$.
Since $\I\varphi \neq \emptyset$,
this implies that $\I{\varphi'} \neq \emptyset$ as well.
Hence, by Lemma~\ref{lem:top:char:fin}
there are finite $d''_i,d'_i$
such that $\I{\varphi''_i} = \up d''_i$
and $\I{\varphi'_i} = \up d'_i$.

On the other hand,
reasoning similarly as in Lemma~\ref{lem:proof:compl:fin:ded:c-false},
we can assume that $\psi$ is of the form
$\bigwedge_{i \in I} (\psi''_i \realto \psi'_i)$
for some finite set $I$,
with
$\I{\psi''_i} \neq \emptyset$
and $\I{\psi'_i} \neq \emptyset$ for all $i \in I$.
Hence, by Lemma~\ref{lem:top:char:fin},
for each $i \in I$ there are finite $e''_i,e'_i$
such that $\I{\psi''_i} = \up e''_i$
and $\I{\psi'_i} = \up e'_i$.
Moreover, $\bigvee_{i \in I}(e''_i \step e'_i)$
exists since $\I\psi \neq \emptyset$.

Hence $\I\psi \sle \I\varphi$ means
\[
\begin{array}{l l l}
  \up \bigvee_{i \in I} \left(e''_i \step e'_i \right)
& \sle
& \up \left( d'' \step d' \right)
\end{array}
\]

\noindent
which implies
\[
\begin{array}{l l l}
  d'' \step d'
& \leq
& \bigvee_{i \in I} \left(e''_i \step e'_i \right)
\end{array}
\]

\noindent
We thus have
\[
\begin{array}{l l l}
  d'
& \leq
& \bigvee_{d'' \geq e''_i} e_i
\end{array}
\]

\noindent
that is
\[
\begin{array}{l l l}
  \up \bigvee_{\up d'' \sle \up e''_i} e'_i
& \sle
& \up d'
\end{array}
\]

\noindent
In other words,
\[
\begin{array}{l l l}
  \I{\bigwedge_{\I{\varphi''} \sle \I{\psi ''_i}} \psi'_i}
& \sle
& \I{\varphi'}
\end{array}
\]

\noindent
and by induction hypothesis
\[
\begin{array}{l l l}
  \bigwedge_{\varphi'' \thesis \psi ''_i} \psi'_i
& \thesis
& \varphi'
\end{array}
\]

Hence we are done since
\[
\begin{array}{l l l}
  \psi
& \thesis
& \bigwedge_{i \in I} \left(
  \left( \bigwedge_{\varphi'' \thesis \psi ''_i} \psi''_i \right)
  \realto
  \psi'_i
  \right)
\end{array}
\]

\noindent
and thus
\[
\begin{array}{l l l}
  \psi
& \thesis
& \left( \bigwedge_{\varphi'' \thesis \psi ''_i} \psi''_i \right)
  \realto
  \bigwedge_{\varphi'' \thesis \psi ''_i} \psi'_i
\end{array}
\]

\noindent
while
\(
\varphi'' \thesis \bigwedge_{\varphi'' \thesis \psi ''_i} \psi''_i
\).
\qed
\end{description}
\end{proof}

We now turn to Theorem~\ref{thm:compl:fin},
namely the completeness for finite judgments.
While this result is essentially due to Abramsky \cite{abramsky91apal},
we nevertheless offer a proof since our system formally
differs from that of \cite{abramsky91apal}.
Let us recall the statement of Theorem~\ref{thm:compl:fin}.

\ThmComplFin*

\begin{proof}
First, note that Lemma~\ref{lem:reft} (\S\ref{sec:reft})
restricts to finite types, in the sense that if a type
$\RTter$ is finite, then there is some $\varphi \in \Lang_\land(\UPT\RTter)$
such that $\RTter \eqtype \reft{\UPT\RTter \mid \varphi}$.

Consider first the case of a sound judgment
$\Env \thesis M : \RT$
where 
$\Env = x_1:\RTbis_1,\dots,x_n:\RTbis_n$
is such that $\I{\RTbis_i} = \emptyset$ for some $i \in \{1,\dots,n\}$.
Since $\I{\UPT{\RTbis_i}}$ is not empty (as it is a Scott domain),
taking $\psi \in \Lang_\land(\UPT{\RTbis_i})$
such that $\RTbis_i \eqtype \reft{\UPT{\RTbis_i} \mid \psi}$,
we must have $\I\psi = \emptyset$
and thus $\psi \thesis \False$ by Lemma~\ref{lem:proof:compl:fin:ded:c-false}.
We can thus conclude by taking $I = \emptyset$ in the rule
\[
\dfrac{
  \begin{array}{l}
  \UPT{\Env'}, x:\PTbis, \UPT{\Env''} \thesis M : \UPT\RT
  \\
  \text{for each $i \in I$,}\quad
  \Env', x:\reft{\PTbis \mid \psi_i},\Env' \thesis M : \RT
  \end{array}}
  {\Env', x : \reft{\PTbis \mid \bigvee_{i \in I} \psi_i} , \Env'' \thesis M : \RT}
\]

Hence we can reduce the case of a sound judgment
$\Env \thesis M : \RT$ with
$\Env = x_1:\RTbis_1,\dots,x_n:\RTbis_n$
such that $\I{\RTbis_i} \neq \emptyset$ for all $i = 1,\dots,n$.
We now reason by induction on the typing derivation of $\UPT\Env \thesis M : \UPT\RT$.
\begin{description}
\item[Case of]
\[
\dfrac{(x:\UPT\RT) \in \UPT\Env}
  {\UPT\Env \thesis x:\UPT\RT}
\]

We have $(x : \RTbis) \in \Env$ for some type $\RTbis$ with $\UPT\RTbis = \UPT\RT$.
Let $\varphi,\psi \in \Lang_\land(\UPT\RT)$
such that $\RT \eqtype \reft{\UPT\RT \mid \varphi}$
and $\RTbis \eqtype \reft{\UPT\RT \mid \psi}$.
By assumption on $\Env \thesis M :\RT$,
we have $\I\psi \sle \I\varphi$.
Hence $\psi \thesis \varphi$
by Proposition~\ref{prop:compl:fin:ded}.
We then conclude by subtyping.

\item[Case of]
\[
\dfrac{\UPT\Env \thesis N_1 : \PT_1
  \qquad
  \UPT\Env \thesis N_2 : \PT_2}
  {\UPT\Env \thesis \pair{N_1,N_2} : \PT_1 \times \PT_2}
\]

\noindent
where $\UPT\RT = \PT_1 \times \PT_2$
and $M = \pair{N_1,N_2}$.

Let $\varphi \in \Lang_\land(\PT_1 \times \PT_2)$
such that $\RT \eqtype \reft{\PT_1 \times \PT_2 \mid \varphi}$.
Our assumption on $\Env \thesis M : \RT$
implies that $\I\varphi \neq \emptyset$.
Hence,
reasoning as in the proof of Lemma~\ref{lem:proof:compl:fin:ded:c-false}
yields that $\varphi \thesisiff \pair{\psi_1,\psi_2}$
for some $\psi_1 \in \Lang_\land(\PT_1)$
and some $\psi_2 \in \Lang_\land(\PT_2)$.

Since $\Env \thesis M :\RT$ is sound,
so are
$\Env \thesis N_1 : \reft{\PT_1 \mid \psi_1}$
and
$\Env \thesis N_1 : \reft{\PT_2 \mid \psi_2}$.
The induction hypotheses yield that
$\Env \thesis N_1 : \reft{\PT_1 \mid \psi_1}$
and
$\Env \thesis N_1 : \reft{\PT_2 \mid \psi_2}$
are derivable.
We can then conclude using the rules
\[
\dfrac{\Env \thesis N_i : \reft{\PT_i \mid \psi_i}
  \qquad
  \Env \thesis N_{3-i} : \PT_{3-i}}
  {\Env \thesis \pair{N_1,N_2} : \reft{\PT_1 \times \PT_2 \mid \form{\pi_i} \psi_i}}
\]

\noindent
for $i = 1$ and $i = 2$.

\item[Case of]
\[
\begin{array}{c}
\dfrac{\UPT\Env \thesis N : \PT \times \PTbis}
  {\UPT\Env \thesis \pi_i(N) : \UPT\RT}
\end{array}
\]

\noindent
where $i = 1,2$ and $M = \pi_1(N)$.
Assume w.l.o.g.\ $i = 1$ (so that $\UPT\RT = \PT$).
Let $\varphi \in \Lang_\land(\UPT\RT)$ such that
$\RT \eqtype \reft{\UPT\RT \mid \varphi}$.
Our assumption on $\Env \thesis M : \RT$
implies that $\I\varphi \neq \emptyset$,
and
since $\Env \thesis M : \RT$ is sound,
we get that
$\Env \thesis N : \reft{\PT \times \PTbis \mid \form{\pi_1}\varphi}$
is sound, and thus derivable by induction hypothesis.
We then conclude with the rule
\[
\dfrac{\Env \thesis N : \reft{\PT \times \PTbis \mid \form{\pi_1} \varphi}}
  {\Env \thesis \pi_1(N) : \reft{\PT \mid \varphi}}
\]

\item[Case of]
\[
\dfrac{\UPT\Env \thesis N : \PT[\rec\TV.\PT/\TV]}
  {\UPT\Env \thesis \fold(N) : \rec\TV.\PT}
\]

\noindent
where $\UPT\RT = \rec\TV.\PT$ and $M = \fold(M)$.

Let $\varphi \in \Lang_\land(\rec\TV.\PT)$
such that $\RT \eqtype \reft{ \rec\TV.\PT \mid \varphi}$.
Our assumption on $\Env \thesis M : \RT$
implies that $\I\varphi \neq \emptyset$.
Hence,
reasoning as in the proof of Lemma~\ref{lem:proof:compl:fin:ded:c-false}
yields that $\varphi \thesisiff \form{\fold}\psi$
for some $\psi \in \Lang_\land(\PT[\rec\TV.\PT/\TV])$.
Moreover, since $\Env \thesis M : \RT$ is sound,
so is
$\Env \thesis N : \reft{\PT[\rec\TV.\PT/\TV] \mid \psi}$.
We can thus conclude using the induction hypothesis and the rule
\[
\dfrac{\Env \thesis N : \reft{\PT[\rec\TV.\PT/\TV] \mid \psi}}
  {\Env \thesis \fold(N) : \reft{\rec\TV.\PT \mid \form\fold \psi}}
\]

\item[Case of]
\[
\dfrac{\UPT\Env \thesis N : \rec\TV.\PT}
  {\UPT\Env \thesis \unfold(N) : \PT[\rec\TV.\PT/\TV]}
\]

\noindent
where $\UPT\RT = \PT[\rec\TV.\PT/\TV]$
and $N = \unfold(M)$.

Let $\varphi \in \Lang_\land(\PT[\rec\TV.\PT/\TV])$
such that
$\RT \eqtype \reft{\PT[\rec\TV.\PT/\TV] \mid \varphi}$.
Our assumption on $\Env \thesis M : \RT$
implies that 
$\Env \thesis N : \reft{\rec\TV.\PT \mid \form\fold \varphi}$
is sound,
and we conclude using the induction hypothesis and the rule
\[
\dfrac{\Env \thesis N : \reft{\rec\TV.\PT \mid \form\fold \varphi}}
  {\Env \thesis \unfold(N) : \reft{\PT[\rec\TV.\PT/\TV] \mid \varphi}}
\]

\item[Case of]
\[
\dfrac{\UPT\Env,x:\PTbis \thesis N : \PT}
  {\UPT\Env \thesis \lambda x.N : \PTbis \arrow \PT}
\]

\noindent
where $\UPT\RT = \PTbis \arrow \PT$,
and where $M = \lambda x.N$.

Let $\varphi \in \Lang_\land(\PTbis \arrow \PT)$
such that
$\RT \eqtype \reft{\PTbis \arrow \PT \mid \varphi}$.
Our assumption on $\Env \thesis M : \RT$
implies $\I\varphi \neq \emptyset$.
Reasoning as in the proof of Lemma~\ref{lem:proof:compl:fin:ded:c-false}
yields that $\varphi \thesisiff \bigwedge_{i \in I}(\varphi''_i \realto \varphi'_i)$
for some finite set $I$.
Let $i \in I$.
The judgment
\(
  \Env
  \thesis
  \lambda x.N
  :
  \reft{\PTbis \arrow \PT \mid \varphi''_i \realto \varphi'_i}
\)
is sound,
and so is
\(
  \Env, x : \reft{\PTbis \mid \varphi''_i}
  \thesis
  N
  :
  \reft{\PT \mid \varphi'_i}
\).
Using the induction hypothesis, we derive
\(
  \Env
  \thesis
  \lambda x.N
  :
  \reft{\PTbis \arrow \PT \mid \varphi''_i \realto \varphi'_i}
\).
We can then derive
\(
  \Env
  \thesis
  \lambda x.N
  :
  \reft{\PTbis \arrow \PT \mid \varphi}
\).

\item[Case of]
\[
\dfrac{\UPT\Env \thesis N : \PTbis \arrow \PT
  \qquad
  \UPT\Env \thesis V : \PTbis}
  {\UPT\Env \thesis N V : \PT}
\]

\noindent
where $\UPT\RT = \PT$ and where $M = N V$.

Write $\Env = x_1:\RTbis_1,\dots,x_n:\RTbis_n$.
Given $i \in \{1,\dots,n\}$,
let $\psi_i \in \Lang_\land(\UPT{\RTbis_i})$
such that $\RTbis_i \eqtype \reft{\UPT{\RTbis_i} \mid \psi_i}$.
Moreover, by assumption we have $\I{\psi_i} \neq \emptyset$,
hence by Lemma~\ref{lem:top:char:fin}
there is a finite $e_i \in \I{\UPT{\RTbis_i}}$
such that $\I{\psi_i} = \up e_i$.

Similarly, let $\varphi \in \Lang_\land(\PT)$
such that $\RT \eqtype \reft{\PT \mid \varphi}$.
Our assumption on $\Env \thesis M : \RT$
implies that $\I\varphi \neq \emptyset$.
Hence, again by Lemma~\ref{lem:top:char:fin}
there is a finite $d \in \I{\PT}$
such that $\I{\varphi} = \up d$.

Since $\Env \thesis M : \RT$ is sound,
we have
$\I M(\vec e) \in \varphi$.
But note that
$\I M(\vec e) = \I N(\vec e)\left(\I V(\vec e) \right)$.

Now, since $\I\PTbis$ is a Scott domain, it is algebraic,
and $\I V(\vec e)$ is the directed l.u.b.\ of the finite $e \leq \I V(\vec e)$.
Since $\I N(\vec e)$ is Scott-continuous, we thus get
that $\I M(\vec e)$ is the l.u.b.\ of the directed set
\[
\left\{
  \I N(\vec e)(e)
  \mid
  \text{$e$ finite and $\leq \I V(\vec e)$}
\right\}
\]

\noindent
Since $d \leq \I M(\vec e)$ and since $d$ is finite,
it follows that we have
$d \leq \I N(\vec e)(e)$ for some finite $e \leq \I V(\vec e)$.
By Lemma~\ref{lem:top:char:fin},
there is a formula $\psi \in \Lang_\land(\PTbis)$
such that $\I\psi = \up e$.

Since $d \leq \I N(\vec e)(e)$, we have
$(e \step d) \leq \I N(\vec e)$,
so that
$\I N(\vec e) \in \I{\psi \realto \varphi}$.
Since $\I N$ is monotone, it follows that 
$\Env \thesis N : \reft{\PTbis \arrow \PT \mid \psi \realto \varphi}$
is sound.
Hence, this judgment is derivable by induction hypothesis.

Similarly, since $e \leq \I V(\vec e)$,
we obtain that the judgment
$\Env \thesis V : \reft{\PTbis \mid \psi}$
is sound and thus derivable.

We can then easily derive $\Env \thesis M : \reft{\PT \mid \varphi}$
and $\Env \thesis M : \RT$.

\item[Case of]
\[
\dfrac{\UPT\Env,x:\UPT\RT \thesis N : \UPT\RT}
  {\UPT\Env \thesis \fix x.N : \UPT\RT}
\]

\noindent
where $M = \fix x.N$.

Write $\Env = x_1:\RTbis_1,\dots,x_n:\RTbis_n$.
Similarly as above, for each $i \in \{1,\dots,n\}$
there is a finite $e_i \in \I{\UPT{\RTbis_i}}$
such that $\I{\RTbis_i} = \up e_i$.
Similarly, there are $\varphi \in \Lang_\land(\UPT\RT)$
such that $\RT \eqtype \reft{\UPT\RT \mid \varphi}$,
and a finite $d \in \I{\UPT\RT}$ such that $\I\varphi = \up d$.

Let $f \colon \I{\UPT\RT} \to \I{\UPT\RT}$
be the Scott-continuous function which takes
$a \in \I{\UPT\RT}$ to $\I N(\vec e,a)$.
We have
\[
\begin{array}{l l l}
  \I{\fix x.N}(\vec e)
& =
& \bigvee_{k \in \NN}
  f^k(\bot)
\end{array}
\]

Since $d \leq \I{\fix x.N}(\vec e)$ with $d$ finite,
there is some $k \in \NN$
such that
$d \leq f^k(\bot)$.
Write $d_k$ for $d$.
By induction, for each $j = k-1,\dots,0$,
there is some finite $d_j$ such that 
$d_{j+1} \leq f(d_j)$
and
$d_j \leq f^{j}(\bot)$.
In particular, $d_0 = \bot$.
For each $j = 0,\dots,k$,
let $\varphi_j$ such that $\I{\varphi_j} = \up d_j$.
Note that $\varphi_k = \varphi$.
Moreover, since $d_0 = \bot$,
we can take $\varphi_0 = \True$.

Again reasoning similarly as above,
we obtain that 
$\Env, x : \reft{\UPT\RT \mid \varphi_j} \thesis N : \reft{\UPT\RT \mid \varphi_{j+1}}$
is sound and thus derivable for each $j = 0,\dots,k-1$.
Moreover,
$\Env \thesis \fix x.N : \reft{\UPT\RT \mid \varphi_0}$
is derivable.
We can then derive $\Env \thesis \fix x.N : \reft{\UPT\RT \mid \varphi}$
by iterated applications of the rule
\[
\dfrac{\Env \thesis \fix x.N : \reft{\PT \mid \psi}
  \quad
  \Env, x: \reft{\PT \mid \psi} \thesis N : \reft{\PT \mid \psi'}}
  {\Env \thesis \fix x.N : \reft{\PT \mid \psi'}}
~(\psi,\psi' \in \Lang_\land(\PT))
\]

\item[Case of]
\[
\dfrac{}
  {\UPT\Env \thesis a : \BT}
\]

Let $\varphi \in \Lang_\land(\BT)$ such that $\RT \eqtype \reft{\BT \mid \varphi}$.
By assumption on $\Env \thesis M :\RT$,
we have $a \in \I\varphi$,
so that $\I{\form a} \sle \I\varphi$.
Hence $\form a \thesis \varphi$
by Proposition~\ref{prop:compl:fin:ded}.
We can then conclude by subtyping and
\[
\dfrac{}
  {\Env \thesis a : \reft{\BT \mid \form a}}
\]

\item[Case of]
\[
\dfrac{ \UPT\Env \thesis N : \BT
  \qquad\text{for each $a \in \BT$,\quad} \UPT\Env \thesis N_a : \UPT\RT}
  {\UPT\Env \thesis \cse\ N\ \copair{a \mapsto N_a \mid a \in \BT} : \UPT\RT}
\]

We reason similarly as in the cases of $\fix x.N$ and $N V$ above.

Write $\Env = x_1:\RTbis_1,\dots,x_n:\RTbis_n$.
For each $i \in I$
there is a finite $e_i \in \I{\UPT{\RTbis_i}}$
such that $\I{\RTbis_i} = \up e_i$.
Also, there is $\varphi \in \Lang_\land(\UPT\RT)$
such that $\RT \eqtype \reft{\UPT\RT \mid \varphi}$.

Assume first that $\I\varphi = \I\True$,
so that $\varphi \thesisiff \True$
by Proposition~\ref{prop:compl:fin:ded}.
Then we have $\RT \eqtype \UPT\RT$ and we easily
derive $\Env \thesis M : \RT$.

Otherwise, we must have $\bot \notin \I\varphi$,
so that $\I M(\vec e) \neq \bot$
and thus $\I N(\vec e) \neq \bot$.
Hence $\I N(\vec e) = b$ for some $b \in \BT$.
Since $\I M(\vec e) = \I{N_b}(\vec e)$,
we obtain that the judgment
$\Env \thesis N_b : \RT$ is sound and thus derivable.
Moreover,
$\Env \thesis N : \reft{\BT \mid \form b}$ is sound and thus
derivable.
We can then conclude using the rule
\[
\dfrac{
  \Env \thesis N : \reft{\BT \mid \form b}
  \qquad
  \Env \thesis N_b : \RT
  \qquad
  \text{for each $a \in A$,\quad} \UPT\Env \thesis N_a : \UPT\RT}
  {\Env \thesis \cse\ N\ \copair{a \mapsto N_a \mid a \in \BT} : \RT}
\]
\qed
\end{description}
\end{proof}

\subsection{\nameref{sec:main}}
\label{sec:proof:main}

Note that Lemma~\ref{lem:compl:nf} and Theorem~\ref{thm:main}
are proven in \S\ref{sec:app:main}.
We prove Proposition~\ref{prop:main:eta}.

\PropMainEta*

The proof of Proposition~\ref{prop:main:eta}
relies on the following
Lemmas~\ref{lem:proof:main:eta:prod} and \ref{lem:proof:main:eta:fun}.

\begin{lemma}
\label{lem:proof:main:eta:prod}
A (not necessarily normal) judgment
$\Env \thesis M : \RT_1 \times \RT_2$
is sound (resp.\ derivable)
if, and only if,
so are $\Env \thesis \pi_1 M : \RT_1$
and $\Env \thesis \pi_2 M : \RT_2$.
\end{lemma}

\begin{proof}
By Lemma~\ref{lem:reft},
there are formulae
$\varphi_1 \in \Lang(\UPT{\RT_1})$
and
$\varphi_2 \in \Lang(\UPT{\RT_2})$
such that
$\RT_1 \eqtype \reft{\UPT{\RT_1} \mid \varphi_1}$
and
$\RT_2 \eqtype \reft{\UPT{\RT_2} \mid \varphi_2}$.
Hence
\(
  \RT_1 \times \RT_2
  \eqtype
  \reft{\UPT{\RT_1} \times \UPT{\RT_2} \mid \pair{\varphi_1,\varphi_2}}
\).

It follows that
$\Env \thesis M : \RT_1 \times \RT_2$
is sound if, and only if,
so are $\Env \thesis \pi_1 M : \RT_1$
and $\Env \thesis \pi_2 M : \RT_2$.

It is clear that $\Env \thesis \pi_1 M : \RT_1$
and $\Env \thesis \pi_2 M : \RT_2$
are derivable whenever so is
$\Env \thesis M : \RT_1 \times \RT_2$.

For the converse, assume that
$\Env \thesis \pi_1 M : \RT_1$
and $\Env \thesis \pi_2 M : \RT_2$
are derivable.
We first show that
$\Env \thesis M : \RT_1 \times \UPT{\RT_2}$.
We reason by induction on the derivation of $\Env \thesis \pi_1 M : \RT_1$
and by cases on the last possible rule.
\begin{description}
\item[Case of]
\[
\dfrac{
  \begin{array}{l}
  \UPT\Env \thesis \pi_1 M : \UPT{\RT_1}
  \\
  \text{for each $i \in I$,}\quad \Env \thesis \pi_1 M : \reft{\UPT{\RT_1} \mid \psi_i}
  \end{array}}
  {\Env \thesis \pi_1 M : \reft{\UPT{\RT_1} \mid \bigwedge_{i \in I} \psi_i}}
\]

\noindent
where $\varphi_1 = \bigwedge_{i \in I}\psi_i$.
Then by induction hypothesis and subtyping,
for all $i \in I$ we can derive
\[
\begin{array}{*{5}{l}}
  \Env
& \thesis
& M
& :
& \reft{\UPT{\RT_1} \times \UPT{\RT_2} \mid \form{\pi_1}\psi_i}
\end{array}
\]

\noindent
and thus
\[
\begin{array}{*{5}{l}}
  \Env
& \thesis
& M
& :
& \reft{\UPT{\RT_1} \times \UPT{\RT_2} \mid \bigwedge_{i \in I} \form{\pi_1}\psi_i}
\end{array}
\]

\noindent
We can then conclude using subtyping and Example~\ref{ex:log:modalnf}.

\item[Case of]
\[
\dfrac{
  \begin{array}{l}
  \UPT\Env, x:\PTbis, \UPT{\Env'} \thesis \pi_1 M : \UPT{\RT_1}
  \\
  \text{for each $i \in I$,}\quad
  \Env, x:\reft{\PTbis \mid \psi_i},\Env' \thesis \pi_1 M : \RT_1
  \end{array}}
  {\Env, x : \reft{\PTbis \mid \bigvee_{i \in I} \psi_i} , \Env' \thesis \pi_1 M : \RT_1}
\]

By induction hypothesis.

\item[Case of]
\[
\dfrac{
  \Env \subtype \Env'
  \quad 
  \RT'_1 \subtype \RT_1
  \quad
  \Env' \thesis \pi_1 M : \RT'_1}
  {\Env \thesis \pi_1 M : \RT_1}
\]

By subtyping we obtain
$\Env' \thesis \pi_2 M : \RT_2$
and the induction hypothesis yields
$\Env' \thesis M : \RT'_1 \times \UPT{\RT_2}$.
Then we are done since
$\RT'_1 \times \UPT{\RT_2} \subtype \RT_1 \times \UPT{\RT_2}$
and $\Env \subtype \Env'$.

\item[Case of]
\[
\dfrac{\Env \thesis M : \reft{\UPT{\RT_1} \times \UPT{\RT_2} \mid \form{\pi_1} \varphi_1}}
  {\Env \thesis \pi_1 M : \reft{\UPT{\RT_1} \mid \varphi_1}}
\]

Since
\(
  \reft{\UPT{\RT_1} \times \UPT{\RT_2} \mid \form{\pi_1} \varphi_1}
  \eqtype
  \RT_1 \times \UPT{\RT_2}
\).

\item[Case of]
\[
\dfrac{\Env \thesis M : \RT_1 \times \RT_2}
  {\Env \thesis \pi_1 M : \RT_1}
\]

Since
$\RT_1 \times \RT_2 \subtype \RT_1 \times \UPT{\RT_2}$.
\end{description}

\noindent
We similarly obtain
$\Env \thesis M : \UPT{\RT_1} \times \RT_2$.
Using subtyping, we then get
\[
\begin{array}{*{5}{l}}
  \Env
& \thesis
& M
& :
& \reft{\UPT{\RT_1} \times \UPT{\RT_2} \mid \form{\pi_1}\varphi_1}
\\

  \Env
& \thesis
& M
& :
& \reft{\UPT{\RT_1} \times \UPT{\RT_2} \mid  \form{\pi_2}\varphi_2}
\end{array}
\]

\noindent
from which we get
\[
\begin{array}{*{5}{l}}
  \Env
& \thesis
& M
& :
& \reft{\UPT{\RT_1} \times \UPT{\RT_2} \mid \pair{\varphi_1,\varphi_2}}
\end{array}
\]

\noindent
and thus
$\Env \thesis M : \RT_1 \times \RT_2$.
\qed
\end{proof}

\begin{lemma}
\label{lem:proof:main:eta:fun}
A (not necessarily normal) judgment
$\Env \thesis M : \RTbis \arrow \RT$
is sound (resp.\ derivable)
if, and only if,
so is
$\Env, x:\RTbis \thesis M x : \RT$.
\end{lemma}

\begin{proof}
By Lemma~\ref{lem:reft},
there are formulae
$\varphi \in \Lang(\UPT{\RT})$
and
$\psi \in \Lang(\UPT{\RTbis})$
such that
$\RT \eqtype \reft{\UPT{\RT} \mid \varphi}$
and
$\RTbis \eqtype \reft{\UPT{\RTbis} \mid \psi}$.
Hence
\(
  \RTbis \arrow \RT
  \eqtype
  \reft{\UPT{\RTbis} \arrow \UPT{\RT} \mid \psi \realto \varphi}
\).

It follows that
$\Env \thesis M : \RTbis \arrow \RT$
is sound if, and only if,
so is $\Env,x:\RTbis \thesis M x : \RT$.

It is clear that $\Env, x:\RTbis \thesis M x : \RT$
is derivable whenever so is
$\Env \thesis M : \RTbis \arrow \RT$.

For the converse, assume that
$\Env,x:\RTbis \thesis M x : \RT$
is derivable.
We show that
$\Env \thesis M : \RTbis \arrow \RT$
is derivable by induction on the derivation of
$\Env,x:\RTbis \thesis M x : \RT$.
We reason by cases on the last possible rule.
\begin{description}
\item[Case of]
\[
\dfrac{
  \begin{array}{l}
  \UPT\Env, x : \UPT\RTbis \thesis M x : \UPT{\RT}
  \\
  \text{for each $i \in I$,}\quad
  \Env,x:\RTbis \thesis M x : \reft{\UPT{\RT} \mid \varphi_i}
  \end{array}}
  {\Env,x:\RTbis \thesis M x : \reft{\UPT{\RT} \mid \bigwedge_{i \in I} \varphi_i}}
\]

\noindent
where $\varphi = \bigwedge_{i \in I}\varphi_i$.
Then by induction hypothesis and subtyping,
for all $i \in I$ we can derive
\[
\begin{array}{*{5}{l}}
  \Env
& \thesis
& M
& :
& \reft{\UPT{\RTbis} \arrow \UPT{\RT} \mid \psi \realto \varphi_i}
\end{array}
\]

\noindent
and thus
\[
\begin{array}{*{5}{l}}
  \Env
& \thesis
& M
& :
& \reft{\UPT{\RTbis} \arrow \UPT{\RT} \mid \bigwedge_{i \in I}(\psi \realto \varphi_i)}
\end{array}
\]

\noindent
We can then conclude by subtyping since
\[
\begin{array}{l !{\quad\thesis\quad} l}
  \bigwedge_{i \in I}(\psi \realto \varphi_i)
& \left( \psi \realto \bigwedge_{i \in I}\varphi_i \right)
\end{array}
\]

\item[Case of]
\[
\dfrac{
  \begin{array}{l}
  \UPT\Env, y:\PTbis, \UPT{\Env'},x : \UPT\RTbis \thesis M x : \UPT{\RT}
  \\
  \text{for each $i \in I$,}\quad
  \Env, y:\reft{\PTbis \mid \psi_i},\Env',x : \RTbis \thesis M x : \RT
  \end{array}}
  {\Env, y : \reft{\PTbis \mid \bigvee_{i \in I} \psi_i} , \Env', x : \RTbis
  \thesis M x : \RT}
\]

By induction hypothesis.

\item[Case of]
\[
\dfrac{
  \begin{array}{l}
  \UPT\Env, x:\UPT\RTbis \thesis M x  : \UPT{\RT}
  \\
  \text{for each $i \in I$,}\quad
  \Env, x:\reft{\UPT\RTbis \mid \psi_i} \thesis M x : \RT
  \end{array}}
  {\Env, x : \reft{\UPT\RTbis \mid \bigvee_{i \in I} \psi_i} \thesis M x : \RT}
\]

\noindent
where $\psi = \bigvee_{i \in I}\psi_i$.
By induction hypothesis and subtyping,
for all $i \in I$ we can derive
\[
\begin{array}{*{5}{l}}
  \Env
& \thesis
& M
& :
& \reft{\UPT{\RTbis} \arrow \UPT{\RT} \mid \psi_i \realto \varphi}
\end{array}
\]

\noindent
and thus
\[
\begin{array}{*{5}{l}}
  \Env
& \thesis
& M
& :
& \reft{\UPT{\RTbis} \arrow \UPT{\RT} \mid \bigwedge_{i \in I}(\psi_i \realto \varphi)}
\end{array}
\]

We can then conclude by subtyping since
\[
\begin{array}{l !{\quad\thesis\quad} l}
  \bigwedge_{i \in I}(\psi_i \realto \varphi)
& \left( \bigvee_{i \in I} \psi_i \right) \realto \varphi
\end{array}
\]

\item[Case of]
\[
\dfrac{
  \Env \subtype \Env'
  \quad
  \RTbis \subtype \RTbis'
  \quad 
  \RT' \subtype \RT
  \quad
  \Env', x:\RTbis' \thesis M x : \RT'}
  {\Env,x:\RTbis \thesis M x : \RT}
\]

By induction hypothesis we obtain
$\Env' \thesis M : \RTbis' \arrow \RT'$.
Then we are done since
$\RTbis' \arrow \RT' \subtype \RTbis \arrow \RT$
and $\Env \subtype \Env'$.

\item[Case of]
\[
\dfrac{\Env \thesis M : \RTbis \arrow \RT
  \qquad
  \Env, x: \RTbis \thesis x : \RTbis}
  {\Env, x: \RTbis \thesis M x : \RT}
\]

Trivial.
\qed
\end{description}
\end{proof}

We can now prove Proposition~\ref{prop:main:eta}.

\begin{proof}[of Proposition~\ref{prop:main:eta}]
We reason by induction on the fonf type $\RT$.
If $\RT$ is normal, then the result is trivial
since 
$\eta(\Env \thesis M : \RT) = \left\{ \Env \thesis M : \RT \right\}$.
In the cases of $\RT_1 \times \RT_2$
and $\RTbis \arrow \RT$ (with $\RTbis$ normal)
we conclude by induction hypothesis
and Lemmas~\ref{lem:proof:main:eta:prod} and \ref{lem:proof:main:eta:fun}, respectively.
\qed
\end{proof}

\subsection{\nameref{sec:compl:general}}
\label{sec:proof:compl:general}

We prove Lemma~\ref{lem:compl:nf:wf}.

\LemComplNfWf*

\begin{proof}
The proof is by induction on $\varphi$.
In the case of $\bigwedge$ and $\bigvee$,
we conclude by induction hypothesis and Example~\ref{ex:log:distr}.
In the case of $\form\triangle\varphi$ ($\triangle$ either $\pi_1$, $\pi_2$ or $\fold$),
we conclude by induction hypothesis and Example~\ref{ex:log:modalnf}.

Consider now the case of $\psi \realto \varphi$.
By induction hypothesis we can assume $\varphi \in \Norm$.
By combining the induction hypothesis with 
Example~\ref{ex:log:distr},
we can assume that $\psi$
is a $\bigvee$ of $\bigwedge$'s of formulae in $\Lang_\land$.
Since
\[
\begin{array}{r !{\quad\thesisiff\quad} l}
  \bigwedge_{i \in I}\left(\psi \realto \varphi_i \right)
& \psi \realto \left(\bigwedge_{i \in I} \varphi_i\right)
\\

  \bigwedge_{i \in I}\left( \psi_i \realto \varphi \right)
& \left(\bigvee_{i \in I} \psi_i \right) \realto \varphi
\end{array}
\]

\noindent
we can reduce to the case of
$\psi = \bigwedge_{i \in I}\psi_i$
and
$\varphi = \bigvee_{k \in K}\varphi_k$
with $\varphi_k,\psi_i \in \Lang_\land$.

Now, note that we can derive
\[
\begin{array}{r !{\quad\thesisiff\quad} l}
  \left( \bigwedge_{i \in I} \psi_i \right)
  \realto
  \varphi
& \bigvee_{\text{$J \sle I$, $J$ finite}}
  \left(
  \left( \bigwedge_{j \in J} \psi_j \right)
  \realto
  \varphi
  \right)
\end{array}
\]

\noindent
Indeed, the $\thesis$ direction is given by the rule $\ax{WF}$.
The converse is derivable using the left-rule for $\bigvee$,
since $\bigwedge_{i \in I}\psi_i \thesis \bigwedge_{j \in J}\psi_j$
whenever $J \sle I$.

It follows that we can actually assume $\psi \in \Lang_\land$
(still with $\varphi = \bigvee_{k \in K}\varphi_k$
where $\varphi_k \in \Lang_\land$).
If $K \neq \emptyset$, then we can conclude using the rule
$\ax{F}$ in Figure~\ref{fig:log:ded}.

Otherwise, $K = \emptyset$ and $\varphi = \False$.
If $\C(\psi)$ then we conclude using the rule $\ax{C}$
in eq.\ \eqref{eq:compl:cc}.
Otherwise, by Proposition~\ref{prop:compl:fin:ded}
we have $\psi \thesis \False$,
and we are done since
$\True \thesis (\False \realto \False)$
by Remark~\ref{rem:log:realto}.
\qed
\end{proof}

}

\opt{full}{\newpage}
\opt{full}{\tableofcontents}

\end{document}